\documentclass[%authoryear,
preprint,review,12pt]{elsarticle}

\usepackage{amssymb}
\usepackage{amsthm}

\usepackage[titletoc,toc,title]{appendix}
\usepackage{dsfont}                     %
\usepackage{amsbsy}					    %
\usepackage{amsfonts}					%
\usepackage{amsmath}					%
\usepackage{amssymb}					%
\usepackage{amsthm}					    %
\usepackage{bbm}						%
\usepackage{graphicx}                   %
\usepackage{multirow}					%
\usepackage{mathtools} %
\usepackage{enumerate}
\usepackage{rotating}
\usepackage[colorlinks,citecolor=blue,urlcolor=blue]{hyperref}
\DeclareMathOperator{\sign}{sign}

\newcommand{\1}{\mathds{1}}

\newcommand{\cX}{X^{\circ}}
\newcommand{\cx}{x^{\circ}}
\newcommand{\cY}{Y^{\circ}}

\newcommand{\eff}{\text{eff}} %
\newcommand{\E}{\mathds{E}} %
\newcommand{\var}{\text{Var}} %
\newcommand{\ex}{\text{ex}} %
\newcommand{\ceil}[1]{\lceil #1 \rceil} 	%
\newcommand{\mat}[1]{\boldsymbol{#1}}      %
\newcommand{\norm}[1]{\| #1 \|}       	%
\newcommand{\R}{\mathds{R}}                	%

\newcommand{\x}{\mathbbmss{x}}             	%
\newcommand{\Z}{\mathds{Z}}                	%
\newcommand{\as}{a.s.}				  	 	%
\newcommand{\ase}{\mbox{ASE}}		%
\newcommand{\cdf}{c.d.f.}				 	%
\newcommand{\iid}{i.i.d.}				   	%
\newcommand{\pdf}{p.d.f.}				 	%
\newcommand{\rv}{r.v.}                	%
\newcommand{\abs}[1]{\left| #1 \right|}   	%
\newcommand{\ch}[1]{\left\{ #1 \right\}}   	%
\newcommand{\co}[1]{\left[ #1 \right]}	   	%
\newcommand{\innprod}[1]{ \langle #1 \rangle } %
\newcommand{\pa}[1]{\left( #1 \right)}	   	%
\newcommand{\comment}[1]{}			   	%

\newenvironment{lalign*}{\linenomath\csname align*\endcsname}{\csname endalign*\endcsname\endlinenomath}
\newtheorem{theorem}{Theorem} %[chapter] %
\newtheorem{corollary}{Corollary} %[chapter] %
\newtheorem{definition}{Definition} %[chapter] %[teorema]
\newtheorem{lemma}{Lemma} %[chapter] %[teorema]
\newtheorem{remark}{Remark} %

\begin{document}

\begin{frontmatter}



\title{Wavelet-based estimation of power densities of size-biased data}


\author[label1]{Michel H. Montoril}
\ead{michel@ufscar.br}
\author[label2]{Alu\'{\i}sio Pinheiro}
\ead{pinheiro@ime.unicamp.br}
\author[label3]{Brani Vidakovic}
\ead{brani@stat.tamu.edu}
\affiliation[label1]{organization={Department of Statistics, Federal University of São Carlos}, country = {Brazil}}
\affiliation[label2]{organization={Department of Statistics, University of Campinas}, country = {Brazil}}
\affiliation[label3]{organization={Department of Statistics, Texas A\&M University}, country = {USA}}

%
%\affiliation{organization={},%Department and Organization

\begin{abstract}
We propose a new wavelet-based method for density estimation when the data are size-biased. More specifically, we consider a power of the density of interest, where this power exceeds 1/2. Warped wavelet bases are employed, where warping is attained by some continuous cumulative distribution function. A special case is the conventional orthonormal wavelet estimation, where the warping distribution is the standard continuous uniform. We show that both linear and nonlinear wavelet estimators are consistent, with  optimal and/or near-optimal rates. Monte Carlo simulations are performed to compare four special settings which are easy to interpret in practice.  An application with a real dataset on fatal traffic accidents involving alcohol illustrates the method. We observe that warped bases provide more flexible and superior estimates for both simulated and real data. Moreover, we find that estimating the power of a density (for instance, its square root) further improves the results.
\end{abstract}



\end{frontmatter}


\section{Introduction}\label{sec_intro}
Frequently one may be interested in the probability density function (\pdf) $ f $ of some random variable (\rv) $ X $. The estimation of $ f $ is done based on a sample $ X_1, X_2, \ldots, X_n $,  usually consisting of independent and identically distributed (\iid) \rv's. In this scenario, there is a wide range of solutions \citep[see, e.g.][]{Fan.Gijbels-1996, Efromovich-1999, Klemela-2009, Scott-2015}. Sometimes it may be impossible to collect such a sample. Instead, observing $ X = x $ happens under the interference of some biasing device that imposes weights according to the magnitude (size) of $ x $. In this situation one observes a sample $ Y_1, Y_2, \ldots, Y_n $ of $ Y $, which has \pdf\ $ g $. This is a \textit{biased sample} and its \pdf\ is related to $ f $ by
\begin{equation}\label{eq:pdf}
g(y) = \dfrac{w(y) f(y)}{\mu},
\end{equation}
where $ g $ is known to be the biased \pdf, $ w $ is a weighting function, and $ \mu = \E\co{w(X)} $.

The problem of biased data is introduced by \cite{Cox-1969}, which proposes
\[ \hat{F}(x) = \dfrac{\hat{\mu}}{n} \sum_{i = 1}^{n} w^{-1}(Y_i) \1(Y_i \leq x) \]
as the estimator of the cumulative distribution function (\cdf) $ F $, where $ w^{-1}(y) = 1/w(y) $, and $ \1(A) $ is one if $ A $ is true, and zero, otherwise. Moreover,
\begin{equation}\label{eq:mu-hat}
\hat{\mu} = \dfrac{n}{\sum_{i = 1}^{n} w^{-1}(Y_i)}.
\end{equation}

Since then, studies involving biased data have gained attention, especially  because of their relevance to a wide range of applications. Consider the following example \cite{Efromovich-1999, Ramirez.Vidakovic-2010-JSPI}. We are interested in the  distribution of the concentration of alcohol in the blood of intoxicated drivers. This data is usually available from routine police reports on arrested drivers charged with driving under the influence. Drivers with higher levels of intoxication have a higher  chance of being arrested, so the collected data are size-biased toward higher concentration of alcohol in the blood. Several other similar examples can be found on the literature. See, e.g. \cite{Efromovich-2004-JSPI, Efromovich-2004-AS, Ramirez.Vidakovic-2010-JSPI} and the references therein.

In terms of the estimation methodology, different approaches have been used to estimate $ f $. For example, \cite{Vardi-1982-AS} considered a nonparametric maximum likelihood approach; \cite{Jones-1991-B} analyzed the mean square error properties of a kernel estimation method;  \cite{ElBarmi.Simonoff-2000-JNS} proposed a simple transformation approach; \cite{Efromovich-2004-AS, Efromovich-2004-JSPI} studied the asymptotic properties of $ f $ and $ F $, respectively, via Fourier series; \cite{Brunel.etal-2009-T} considered projection estimator methods for right censored data; and \cite{Borrajo.etal-2017-JoNS} proposed bandwidth selection methods for the estimation of $ f $ when the kernel approach is used. Also, in the context of biased data, \cite{Tenzer.etal-2021-JotASA} proposed two approaches to test independence, both based on resampling methods.

In density estimation problems, wavelet bases are strong competitors to other orthonormal bases, such as Fourier and Hermite, among others. Wavelet bases are known to possess several optimality properties, such as adaptive simultaneous localization in space and scale/frequency, and have been used to solve several statistical problems. In the general density estimation context: \cite{Doukhan-1988-CSI, Doukhan.Leon-1990-CSIM} introduce linear wavelet estimators; \cite{Kerkyacharian.Picard-1992-SPL, Kerkyacharian.Picard-1993-S&PL} explore linear wavelet estimator in Besov spaces; \cite{Donoho.etal-1995-JRSSSBM, Donoho.etal-1996-AS} consider nonlinear estimators and studies their minimax properties in Besov spaces; \cite{Pinheiro.Vidakovic-1997-CSDA} proposes the estimation based on the square root of the density, which is useful to control positiveness and $ L_1 $-norm for the density estimate (the density estimate to integrate to 1); and \cite{Gine.Nickl-2009-AP} derives several uniform limits for the linear wavelet estimator.

Several studies have been developed on wavelet estimation of densities for biased data. Papers \cite{Ramirez.Vidakovic-2010-JSPI} and \cite{Chesneau.etal-2012-JoNS} consider wavelet-based methods to estimate the density of stratified biased data under the assumption that the data is independent and associated, respectively, and \cite{Cutillo.etal-2014-JoSPaI} derives asymptotic properties in $ L_2 $-sense for linear and nonlinear wavelet-based estimators. Paper \cite{Guo.Kou-2019-JIA} exploits pointwise estimation, while \cite{Kou.Guo-2018-JIA} and \cite{Guo.Kou-2019-JoCaAM} study the asymptotic properties of wavelet estimators for the density of multivariate (strong mixing and independent) biased data. Papers \cite{Yu-2020-RM} and \cite{Yu.Liu-2020-JoCaAM} consider the case of biased data with multiple change-points.

The novelty of the proposed approach is that we consider the estimation of the power density for biased data, say $ f^a $, $ a \geq 1/2 $. The standard approach for the direct density estimation is a special case when the power is $ a = 1 $. Another special case we should mention is $ a = 1/2 $, considered by \cite{Pinheiro.Vidakovic-1997-CSDA} for ``unbiased'' \iid\ data. In the case of $ a = 1/2 $, there is an advantage of dealing with orthonormal bases. Projection estimators can  be constructed to ensure non-negative density estimates that integrate to one (see the aforementioned reference for more details). Moreover, no $L_2$ assumption on $f$ is required.

Another contribution of this paper is the use of warped wavelet bases in this context. This can be useful in stabilizing numerical estimates for finite data, specially in the regions with sparse observations, which is quite common given the biasing function. These warped wavelet bases provided good performance \cite{Montoril.etal-2018-IJWMIP}. Some other references associated to warped wavelets are \cite{Cai.Brown-1998-AS, Cai.Brown-1999-SPL, Kerkyacharian.Picard-2004-B}.

This paper is organized as follows. In Section \ref{sec:wav-est} we propose and analyze the wavelet-based estimation method. Some theoretical results, special cases and computational aspects are discussed there. In Section \ref{sec:numerical}, we evaluate the performance of the methodology, using four special cases, through Monte Carlo simulation studies and a real dataset application. Some comments and conclusions are made in Section \ref{sec:conclusions}. %The proofs of the theoretical results are deferred to the \ref{sec:appendix}.

\section{Wavelet-based estimator}\label{sec:wav-est}

\subsection{A brief review of wavelets}

Wavelet bases are systems of functions capable of an efficient and parsimonious representation of other square integrable functions. Specifically, any function $ f \in L_2([0,1]) $ can be represented in $L_2([0,1])$-norm as
\[ f(x) = \sum_{k \in \Z} c_{j_0 k}\phi_{j_0k}(x) +
\sum_{j=j_0}^{\infty}\sum_{k \in \Z} d_{jk}\psi_{jk}(x), \]
where $ \phi_{j_0k}(x) = 2^{j_0/2} \phi(2^{j_0}x - k) $ and $
\psi_{jk}(x) = 2^{j/2} \psi(2^{j}x - k) $ are generated by the scaling
$ \phi $, and embedded  on a  multiresolution analysis of $L_2([0,1])$
\citep{Mallat-2008}. $\phi$ is called the scaling function or father wavelet. $\psi$ is called mother wavelet or simply wavelet, and it is also generated by $\phi$.

Since we assume that $ f $ is defined on $ [0,1] $, we consider the periodized version of the wavelet bases, whose atoms can be written as
\[
\phi_{jk}^p(x) = \sum_{l} \phi_{jk}(x-l), \; \psi_{jk}^p(x) = \sum_{l} \psi_{jk}(x-l), \; x \in [0,1],
\]
where $ k = 0, \ldots, 2^j - 1 $, $ j \in \Z $. One can show that, when $ \{ \phi_{Jk} \}_k $ generates an orthonormal basis, then $ \{ \phi_{Jk}^p \}_k $ will be orthonormal as well. Furthermore, if we consider compactly supported Daubechies wavelet bases, their periodized version shares most of their properties, with the advantage of dealing with the boundary problems \citep{Restrepo.Leaf-1997-IJNME}. In the sequel we adopt the periodized wavelets, and drop the superscript $ p $ for notational convenience. Thus, it is easy to see that shifts are bounded within the scales,
\[ f(x) = \sum_{k=0}^{2^{j_0}-1} c_{j_0 k}\phi_{j_0k}(x) +
\sum_{j=j_0}^{\infty}\sum_{k=0}^{2^j-1} d_{jk}\psi_{jk}(x), \]
where the Fourier coefficients can be written as
\[ c_{j_0k} = \int_{0}^{1} \phi_{j_0k}(x) f(x) dx \quad \text{and} \quad d_{jk} = \int_{0}^{1} \psi_{jk}(x) f(x) dx.\]

We now denote the $a$-th power of $f$ as $ f^a(x) = [f(x)]^a $, $ a \in \R $, let $ h $ be the \pdf\ associated to the continuous \cdf\ $ H $, consider $ H^* $ as the inverse of $ H $, and take $ r^a = f^a \circ H^* $ and $ y = H(x) $, $ x \in [0, 1] $. Then,
\begin{eqnarray}\label{eq:w-analysis}
f^a(x) & = & f^a\pa{H^{*}\pa{H(x)}} \equiv r^a(y) \nonumber \\
& = & \sum_{k = 0}^{2^{j_0}-1} c_{j_0k} \phi_{j_0k}(y) + \sum_{j \geq j_0} \sum_{k = 0}^{2^j-1} d_{jk} \psi_{jk} (y) \\
& = & \sum_{k = 0}^{2^{j_0}-1} c_{j_0k} \phi_{j_0k}\co{H(x)} + \sum_{j \geq j_0} \sum_{k = 0}^{2^j-1} d_{jk} \psi_{jk} \co{H(x)}. \nonumber
\end{eqnarray}
The wavelet basis in \eqref{eq:w-analysis} is ``warped'' by $ H $, and the expansion can be seen as a generalization of the ordinary wavelet analysis. Observe that, when $ H(x) = x $, \eqref{eq:w-analysis} reduces to the usual case. Furthermore, as discussed in Section \ref{sec_intro}, this warped representation may be advantageous for
statistical analyses of irregularly spaced data \citep{Montoril.etal-2018-IJWMIP}.

\subsection{Linear wavelet-based estimation}

We consider $ f^a $, $ a \geq 1/2$ for the size-biased data problem, where $ f $ is defined as in \eqref{eq:pdf}. Assuming that $ f \in L_{2a}([0,1]) $ (or, equivalently, $ f^a \in L_2([0,1]) $), $ f^a $ can be approximated by its orthogonal projection on some multiresolution space $ V_{J_0} $, say $ f_{J_0}^a $, for any arbitrary resolution level $ J_0 $, which results in
\begin{equation}
f_{J_0}^a(x) = \sum_{k = 0}^{2^{J_0}-1} c_{J_0k} \phi_{J_0k}\co{H(x)}. \label{eq:rep1} %\\
%& \equiv & \sum_{k = 0}^{2^{j_0}-1} c_{j_0k} \phi_{j_0k}\co{H(x)} + \sum_{j = j_0}^{J-1} \sum_{k = 0}^{2^j-1} d_{jk} \psi_{jk} \co{H(x)}. \label{eq:rep2}
\end{equation}
%It is important remember that both representations above are equivalent, i.e., the bases $ \{\phi_{Jk}; k = 0, 1, \ldots, 2^J-1\} $ and $ \{\phi_{j_0k}, \psi_{jl}; k = 0, 1, \ldots, 2^{j_0}-1, l = 0, 1, \ldots, 2^j - 1, j = j_0, \ldots, J-1 \} $ spam $ V_J $. For the sake of simplicity, let us focus initially on representation \eqref{eq:rep1}.
The Fourier coefficients satisfy
\begin{eqnarray} \label{eq:cjk}
\begin{split}
c_{J_0k} & = \innprod{f^a \circ H^*, \phi_{J_0k}} %= \int_{0}^{1} r^a(y) \psi_{jk}(y) dy
%= \int_{0}^{1} f^a(x) \phi_{J_0k}\co{H(x)} h(x) dx
=  \int_{0}^{1} \phi_{J_0k}\co{H(x)} \dfrac{\mu^a g^a(x)}{w^a(x)} h(x) dx \\
%& =  \mu^a \int_{0}^{1} \phi_{J_0k}\co{H(x)} \dfrac{g^a(x)}{w^a(x)} h(x) dx \\
& =  \mu^a \E\ch{ \dfrac{\phi_{J_0k}\co{H(Y)} g^{a-1}(Y)}{w^a(Y) h(Y)} },
\end{split}
\end{eqnarray}
where $ k = 0, 1, \ldots, 2^{J_0} - 1 $.

Based on \eqref{eq:cjk}, the coefficients could be estimated by moment matching, resulting in
\[ \bar{c}_{J_0k} = \dfrac{\mu^a}{n} \sum_{i = 1}^{n} \dfrac{\phi_{J_0k}\co{H(Y_i)} g^{a-1}(Y_i) h(Y_i)}{w^a(Y_i)}. \]
Such estimator  is not useful in practical situations because both $ \mu $ and $ g $ are unknown. This problem can be solved by plugging in their estimates. Observe that $ g $ can be easily estimated from the biased data. Kernel-based and wavelet-based estimators are just two efficient methodologies. Let us denote this estimator by $ \hat{g} $. We can use $ \hat{\mu} $ as defined by \eqref{eq:mu-hat}. Therefore, the linear wavelet estimator of $ f^a $ can be written as
\begin{eqnarray}\label{est:linear}
\begin{split}
\hat{f}_{J_0}^a(x) & = \sum_{k = 0}^{2^{J_0}-1} \hat{c}_{J_0k} \phi_{J_0k}\co{H(x)}, \\
\hat{c}_{J_0k} & = \dfrac{\hat{\mu}^a}{n} \sum_{i = 1}^{n} \dfrac{\phi_{J_0k}\co{H(Y_i)} \hat{g}^{a-1}(Y_i) h(Y_i)}{w^a(Y_i)}.
\end{split}
\end{eqnarray}

One can then estimate $ f $ by
\begin{eqnarray}
\hat{f}_{J_0}(x) = \co{\hat{f}_{J_0}^a(x)}^{1/a}.\label{eq_linearwavelet}
\end{eqnarray}

\subsection{Regularized wavelet-based estimation}\label{sec:reg-ests}

Choosing the resolution level $J_0$ is a well-known problem in statistical analysis by wavelets \citep[see, e.g.][for details]{Vidakovic-1999, Morettin.etal-2017}.  Larger values of $J_0$ lead to larger variances, whilst smaller values yield fewer coefficients leading to oversmoothing. Balancing bias and variance may be attained by employing more detail coefficients.  Regularization of these ``extra'' coefficients helps reducing oversmoothing and providing adaptive estimates. We consider a projection on $ V_{J_1} $:
\begin{eqnarray}
f_{J_1}^a(x) & = & \sum_{k = 0}^{2^{J_0}-1} c_{J_0k} \phi_{J_0k}\co{H(x)}  + \sum_{j = j_0}^{J_1-1} \sum_{k = 0}^{2^j-1} d_{jk} \psi_{jk} \co{H(x)} \\
\label{eq:rep2} %\\
& = & f_{J_0}^a(x) + \sum_{j = J_0}^{J_1-1} \sum_{k = 0}^{2^j-1} d_{jk} \psi_{jk} \co{H(x)}.
%& \equiv & \sum_{k = 0}^{2^{j_0}-1} c_{j_0k} \phi_{j_0k}\co{H(x)}
\end{eqnarray}

Analogously to \eqref{eq:cjk}, one has
\begin{equation}\label{djk}
%c_{J_0k} & = & \mu^a \E\ch{ \dfrac{\phi_{J_0k}\co{H(Y)} g^{a-1}(Y) h(Y)}{w^a(Y)} } \\
d_{jk} = \mu^a \E\ch{ \dfrac{\psi_{jk}\co{H(Y)} g^{a-1}(Y) h(Y)}{w^a(Y)} },
\end{equation}
$ k = 0, 1, \ldots, 2^j-1 $, $ j = J_0, \ldots, J_1 - 1 $. The detail coefficients can be estimated as
\[ \hat{d}_{jk} = \dfrac{\hat{\mu}^a}{n} \sum_{i = 1}^{n} \dfrac{\psi_{jk}\co{H(Y_i)} \hat{g}^{a-1}(Y_i) h(Y_i)}{w^a(Y_i)}. \]

We shrink $\hat{d}_{jk}$ by
\[
\hat{d}_{jk}^* = \lambda_{jk} \hat{d}_{jk},
\]
where
$ 0 \leq \lambda_{jk} \leq 1 $ plays the role of thresholding regularizer. There are several regularization methods that satisfy this representation. Two of the most famous are the hard- and the soft-thresholding approaches, where the latter is written as $ \hat{d}_{jk}^* = \sign(\hat{d}_{jk}) (|\hat{d}_{jk}| - \lambda)_+ $, and the former satisfies $ \hat{d}_{jk}^* = \hat{d}_{jk} \1(|\hat{d}_{jk}| > \lambda) $. See \cite{Vidakovic-1999} for details.

The proposed regularized nonlinear wavelet estimator is then given by
\begin{equation}\label{est:nonlinear}
\tilde{f}_{J_1}^a(x) = \hat{f}_{J_0}^a(x) + \sum_{j = J_0}^{J_1-1} \sum_{k = 0}^{2^j-1} \hat{d}_{jk}^* \psi_{jk} \co{H(x)}.
\end{equation}

We show in Section \ref{sec:theory} that proposals \eqref{est:linear} and \eqref{est:nonlinear} both result in consistent estimators.

\subsection{Theoretical results}\label{sec:theory}

In this section we discuss the mean integrated square error (MISE) consistency of $ \hat{f}_{J_0}^a $ and $ \tilde{f}_{J_1}^a $. For instance, we say that $ \hat{f}_{J} $ is MISE-consistent estimating $ f$ if $ \lim_n \E\norm{\hat{f}_{J} - f}_2 = 0 $, where $ \norm{h}_p = \pa{\int_{0}^{1} h^p(x) dx}^{1/p} $, $ 1 \leq p < \infty $.

Usually $f$ possess some degree of smoothness. Specifically we assume that $f$ belongs to a Sobolev space.

\begin{definition}
Let $ m \in \{0, 1, \ldots\} $ and $ 1 \leq p \leq \infty $. The Sobolev space corresponds to the set of functions $ W_p^m([0,1]) = \ch{f \in L_p([0,1]): \ f^{(m)} \in L_p([0,1])} $. It is equipped with the norm $ \norm{f}_{W_p^m} = \norm{f}_p + \norm{f^{(m)}}_p $.
\end{definition}

Let us focus on the Sobolev ball
\[ \tilde{W}_p^m(U) = \ch{f \in W_p^m([0,1]) : f \text{ is a } \pdf, \ \norm{f^{(m)}}_p \leq U}. \]
This class of functions is similar, for example, to the class used by \cite{Hardle.etal-1998} (Theorem 10.1) or \cite{Wang.etal-2013-JIA}.

Further, we impose some regularity conditions. First some notation is required. For two sequences of positive numbers $ a_n $ and $ b_n, $  we say that $ a_n \lesssim b_n $, if the ratio is uniformly bounded, and $ a_n \asymp b_n $, if $ a_n \lesssim b_n $ and $ b_n \lesssim a_n $.

\subsubsection*{Assumptions}
\begin{enumerate}
\item[(a1)] $ f $ in \eqref{eq:pdf} is bounded away from zero and infinity and $ f^a \in \tilde{W}_2^m(U) $, for $ a \geq 1/2 $, $ 0 < U < \infty $ and $ m = 1, 2, \ldots $.
\item[(a2)] $ w $ in \eqref{eq:pdf} is bounded away from zero and infinity.
%\item[(a3)] $ J_0 \equiv J_0(n) $ and $ J_1 \equiv J_1(n) $ are positive integers such that $ J_0 \leq J_1 $, $ 2^{J_0} \asymp 2^{J_1} $ and $ J_12^{J_1}/n \to 0 $ as $ n \to \infty $;
\item[(a3)] The \cdf\ $ H $ used to warp the wavelet basis is continuous
and strictly monotone. Its \pdf\ $ h $ is bounded away from zero and infinity uniformly on $ [0, 1] $.
\item[(a4)] The employed wavelet basis is a periodized version of some Daubechies compactly supported wavelet basis, with at least $ m $ vanishing moments.
\end{enumerate}

The above assumptions are frequently used in the literature. For example, (a2) is used in \cite{Efromovich-2004-JSPI, Efromovich-2004-AS, Wang.etal-2013-JIA} and (a3) is considered by \cite{Montoril.etal-2018-IJWMIP}.

\begin{remark}
The assumptions (a1) and (a2) ensure that $ g $ in \eqref{eq:pdf} is also bounded away from zero and infinity.
\end{remark}

\begin{remark}
The proofs of the results presented in this section are available in the Supplementary Material. % \ref{sec:appendix}.
\end{remark}

\begin{theorem}\label{theo:lin-a1}
Suppose assumptions (a1) -- (a4) hold. Furthermore, assume that $ J_0 \equiv J_0(n) $ is an increasing sequence of positive integers such that $ 2^{J_0}/n \to 0 $. Then, for $ a = 1 $, $ \hat{f}_{J_0}^a $ in \eqref{est:linear} is MISE-consistent. Its rate of convergence is given by
\[
\sup_{f^a \in \tilde{W}_2^m(U)} \E \norm{\hat{f}_{J_0}^a - f^a}_2^2 \lesssim \dfrac{2^{J_0}}{n} + 2^{-2mJ_0}.
\]
\end{theorem}

Theorem \ref{theo:lin-a1} states that the use of warping wavelets in the conventional case ($ a = 1 $) does not impact the minimax rate of convergence. Therefore, the rate of convergence will be minimized if we consider $ 2^{J_0} \asymp n^{\frac{1}{2m + 1}} $. In this case, one can see that $$ \sup_{f^a \in \tilde{W}_2^m(U)} \E \norm{\hat{f}_{J_0}^a - f^a}_2^2 \lesssim n^{-\frac{2m}{2m + 1}}. $$

\begin{theorem}\label{theo:lin-aa}
Suppose assumptions (a1) -- (a4) hold. Furthermore, assume that $ J_0 \equiv J_0(n) $ is an increasing sequence of positive integers. If there exists a positive sequence $ D_n $ such that $ 2^{J_0}D_n \to 0 $ as $ n \to \infty $ and\linebreak $ \sup_{y \in [0,1]} \abs{\hat{g}(y) - g(y)}^2 \lesssim D_n $, then for $ a \neq 1 $, $ \hat{f}_{J_0}^a $ given by \eqref{est:linear} is MISE-consistent. Its rate of convergence will be
\[
\sup_{f^a \in \tilde{W}_2^m(U)} \E \norm{\hat{f}_{J_0}^a - f^a}_2^2 \lesssim 2^{J_0}D_n + \dfrac{2^{J_0}}{n} + 2^{-2mJ_0}.
\]
\end{theorem}

Theorem \ref{theo:lin-aa} states that  the rate of convergence will no longer be minimax when we consider a non-trivial power of $ f $. As expected, the rate of convergence is slower. This happens because the estimators of the coefficients $ c_{J_0k} $ in \eqref{est:linear} depend on another estimator, specifically $ \hat{g} $. We need not only convergence of $ \hat{c}_{J_0k} $ to  $ c_{J_0k} $, but  the sup-norm convergence between $ \hat{g} $ and $ g $ as well.

We should also note regarding Theorem \ref{theo:lin-aa}  that its stated  MISE-consistency depends on $ 2^{J_0} D_n \to 0$ as $ n \to \infty $, where $D_n$ is $ \hat{g} $'s rate of convergence. One finds several sup-norm convergence results for kernel density estimators \citep{Silverman-1978-AS, Gine.Guillou-2002-AIHP} and wavelet-based estimators, for both linear and nonlinear approaches \citep{Gine.Nickl-2009-AP}.

We illustrate this issue by considering $ \hat{g} $ a linear wavelet-based estimator with some resolution level $ j_n $, i.e.,
\begin{eqnarray}
\begin{split}
\hat{g}(x) & = \sum_{k = 0}^{2^{j_n} - 1} \hat{\alpha}_{j_nk} \phi_{j_nk}(x), \\
\hat{\alpha}_{j_nk} & = \dfrac{1}{n} \sum_{i = 1}^{n} \phi_{j_nk}(Y_i).
\end{split}
\end{eqnarray}
The resolution level $ j_n $ is assumed to be an increasing sequence of $ n $ that satisfies
\begin{equation}\label{ass:jn}
\dfrac{j_n 2^{j_n}}{n} \to 0, \quad \dfrac{\log\log n}{n} \to 0 \quad \text{and} \quad \sup_{n \geq n_0} (j_{2n} - j_n) \leq \tau
\end{equation}
for some $ \tau \geq 1 $ and some $ n_0 < \infty $. The necessary conditions for Theorem \ref{theo:lin-aa} to hold  are guaranteed by Theorem \ref{theo:gine.et.al}.

\begin{theorem}[\citep{Gine.Nickl-2009-AP}]\label{theo:gine.et.al}
Suppose that $ g \in W_2^m([0,1]) $. If the assumptions \eqref{ass:jn} and (a4) hold, then
\[ \sup_y |\hat{g}(y) - g(y)| \lesssim \sqrt{\dfrac{j_n 2^{j_n}}{n}} + 2^{-mj_n}. \]
\end{theorem}

The original version of Theorem \ref{theo:gine.et.al} is more general, and $g$ is considered to live in a Besov space. This is not an issue here because Soboloev spaces are covered as well. See \cite{Gine.Nickl-2009-AP} Remarks 3 and 8 for details. If we take $ j_n = J_0 $ Corollary \ref{corol:lin-aa} summarizes the consistency results for linear warped wavelet estimators.

\begin{corollary}\label{corol:lin-aa}
Suppose assumptions (a1) -- (a4). Furthermore, assume that $ J_0 \equiv J_0(n) $ is an increasing sequence of positive integers satisfying \eqref{ass:jn} and $ J_0 2^{2J_0}/n \to 0 $. If $ g $ satisfies suppositions in Theorem \ref{theo:gine.et.al}, then for $ a \neq 1 $, $ \hat{f}_{J_0}^a $ given by \eqref{est:linear} is MISE-consistent. Its rate of convergence is given by
\[
\sup_{f^a \in \tilde{W}_2^m(U)} \E \norm{\hat{f}_{J_0}^a - f^a}_2^2 \lesssim \dfrac{J_0 2^{2J_0}}{n} + 2^{-(2m-1)J_0}.
\]
\end{corollary}

The rate of convergence above will be optimal if $ 2^{J_1} \asymp (n/\log n)^{1/(2m + 1)} $, where \[ \sup_{f^a \in \tilde{W}_2^m(U)} \E \norm{\hat{f}_{J_0}^a - f^a}_2^2 \lesssim \pa{\dfrac{\log n}{n}}^{\frac{2m-1}{2m+1}}. \]

Theorems \ref{theo:nonlin-a1} and \ref{theo:nonlin-aa} state that it is possible for nonlinear warped wavelet estimators to attain the same MISE convergence rates obtained in Theorems \ref{theo:lin-a1} and \ref{theo:lin-aa} (and Corollary \ref{corol:lin-aa}) for linear warped wavelet estimators.

\begin{theorem}\label{theo:nonlin-a1}
Suppose assumptions (a1) -- (a4). Furthermore, assume that $ J_0 \equiv J_0(n) $ and $ J_1 \equiv J_1(n) $ are increasing sequences of positive integers such that $ J_0 \leq J_1 $, $ 2^{J_0} \asymp 2^{J_1} $ and $ 2^{J_1}/n \to 0 $. Then, for $ a = 1 $, $ \tilde{f}_{J_1}^a $ given by \eqref{est:nonlinear} is MISE-consistent. Its rate of convergence is
\[
\sup_{f^a \in \tilde{W}_2^m(U)} \E \norm{\tilde{f}_{J_1}^a - f^a}_2^2 \lesssim \dfrac{2^{J_1}}{n} + 2^{-2mJ_1}.
\]
\end{theorem}

\begin{theorem}\label{theo:nonlin-aa}
Suppose assumptions (a1) -- (a4). Furthermore, assume that $ J_0 \equiv J_0(n) $ and $ J_1 \equiv J_1(n) $ are increasing sequences of positive integers such that $ J_0 \leq J_1 $, $ 2^{J_0} \asymp 2^{J_1} $, and there exists a positive sequence $ D_n $ such that $ 2^{J_1}D_n \to 0 $, as $ n \to \infty $, with $ \sup_{y \in [0,1]} \abs{\hat{g}(y) - g(y)}^2 \lesssim D_n $. If $ g $ satisfies Theorem \ref{theo:gine.et.al}, then for $ a \neq 1 $, $ \hat{f}_{J_1}^a $ given by \eqref{est:nonlinear} is MISE-consistent. Its rate of convergence is given by
\[
\sup_{f^a \in \tilde{W}_2^m(U)} \E \norm{\tilde{f}_{J_1}^a - f^a}_2^2 \lesssim 2^{J_1}D_n + \dfrac{2^{J_1}}{n} + 2^{-2mJ_1}. %\dfrac{J_1 2^{2J_1}}{n} + 2^{-(2m-1)J_1}.
\]
\end{theorem}

\begin{corollary}\label{corol:nonlin-aa}
Suppose assumptions (a1) -- (a4) hold. Furthermore, assume that $ J_0 \equiv J_0(n) $ and $ J_1 \equiv J_1(n) $ are increasing sequences of positive integers such that $ J_0 \leq J_1 $, $ 2^{J_0} \asymp 2^{J_1} $ and $ J_1 2^{2J_1}/n \to 0 $. If $ g $ behaves as stated in Theorem \ref{theo:gine.et.al}, then for $ a \neq 1 $, $ \hat{f}_{J_1}^a $ given by \eqref{est:nonlinear} is MISE-consistent. Its rate of convergence is given by
\[
\sup_{f^a \in \tilde{W}_2^m(U)} \E \norm{\tilde{f}_{J_1}^a - f^a}_2^2 \lesssim \dfrac{J_1 2^{2J_1}}{n} + 2^{-(2m-1)J_1}.
\]
\end{corollary}

It is usual in the literature to demonstrate that nonlinear wavelet-based estimators are asymptotically minimax up to a logarithmic term \citep[see, e.g.][]{Donoho.etal-1996-AS}, where the finest resolution level does not depend on unknown quantities such as regularity parameters of function spaces. In practice, as mentioned by \cite{Gine.Nickl-2009-AP}, one can choose $ J_1 $ sufficiently large (and independent of $ m $) and regularize by shrinking or thresholding selected wavelet coefficient estimates in resolution levels from $ J_0 $ to $ J_1 $. In the case of the regularized wavelet estimators proposed here, when $ a = 1 $, assumption (a3) ensures that the adaptive rate of convergence for the nonlinear warped estimator can be easily derived based on results known in the literature (under similar arguments presented in \eqref{asymp:norm-rel}, proof of Theorem \ref{theo:lin-aa}, in the Supplementary Material). An adaptive rate that could be taken into account is presented in Theorem 4.1 of \cite{Chesneau-2010-JotKSS}, where the author consider a block thresholding approach. The case where $ a \neq 1 $ is more problematic. Observe that the rate of convergence in Theorem \ref{theo:nonlin-aa} depends on the sequence $ D_n $, which is quite generic and makes the development of the results unfeasible. Even in specific situations, such as the one illustrated in Corollary \ref{corol:lin-aa}, the assumptions necessary to derive adaptive rates of convergence tend to be unrealistic.

By the arguments presented above, we simply focus on showing that the proposed nonlinear wavelet estimators \eqref{est:nonlinear} are still MISE-convergent, although not in an adaptive way. Therefore, Theorems \ref{theo:nonlin-a1} and \ref{theo:nonlin-aa} guarantee that the regularized warped wavelet estimators can both rely on a sparse representation and attain the same rates of convergence as their linear versions, where the optimal rate of the former still depend on resolution levels which, in sequel, depend on the regularity parameter $ m $. In practice, this can be seen as a drawback, because $ m $ is unknown. On the other hand, one can easily perform an empirical analysis to estimate the finest resolution level. We illustrate it for Theorem \ref{theo:nonlin-a1} and Corollary \ref{corol:nonlin-aa}.

%Although Theorems \ref{theo:nonlin-a1} and \ref{theo:nonlin-aa} guarantee that the regularized warped wavelet estimators attain the same rates of convergence from their linear versions, the former still chooses resolution levels which depend on the regularity parameter of the Sobolev space. In practice, this can be seen as a drawback, because $ m $ is unknown. On the other hand, one can easily perform an empirical analysis to estimate the finest resolution level. We illustrate this empirical analysis for Theorem \ref{theo:nonlin-a1} and Corollary \ref{corol:nonlin-aa}.

Let us focus initially on the case where $ a = 1 $. The rate of convergence of $ \sup_{f^a \in \tilde{W}_2^m(U)} \E \norm{\tilde{f}_{J_1}^a - f^a}_2^2 $, in the case where one chooses $ 2^{J_1} \asymp n^{\frac{1}{2k+1}} $ is $ R_m^k(n) $. Thus, one can see that $ R_m^k(n) = n^{-2(k \wedge m)/(2k+1)} $. Observe that $ R_m^m(n) \lesssim n^{-2m/(2m+1)} $. Therefore, it is possible to compare the performance of a ``bad'' choice of resolution level with respect to the optimal rate. In this case, let us denote $ \eff(k, m) = R_m^k(n)/R_m^m(n) = n^{-2(k \wedge m)/(2k+1) + 2m/(2m+1)} $, which play the role of a kind of asymptotic relative efficiency. For the case where $ a \neq 1 $, one can consider the choice of $ 2^{J_1} \asymp (n/\log n)^{\frac{1}{2k+1}} $ and obtain $ R_m^k(n) = ((\log n)/n)^{\frac{2(k \wedge m) - 1}{2k + 1}} $, which provides $ \eff(k, m) = ((\log n)/n)^{\frac{2(k \wedge m) - 1}{2k + 1} - \frac{2m-1}{2m+1}} $.

The asymptotic relative efficiencies based on Theorems \ref{theo:nonlin-a1} and \ref{theo:nonlin-aa} are presented in Figure \ref{fig:asseff}. In the case where $ a = 1 $, for choices $k \leq m $ the efficiency is closer to one, and it increases quickly for choices $ k > m $. On the other hand, in the case where $ a \neq 1 $, smaller values of $ k $ are interesting when the $ f^a $ is not too regular ($ m $ small). The more regular $ f^a $ becomes (greater values of $ m $), the choice $ k = 1 $ tends to increase $ \eff(1, m) $. Therefore, $ k = 2 $ or $ k = 3 $ seem to be good candidates to provide nearly optimal rates of convergence.

\begin{figure}[!ht]
\centering
\includegraphics[angle=0,width=0.49\linewidth]{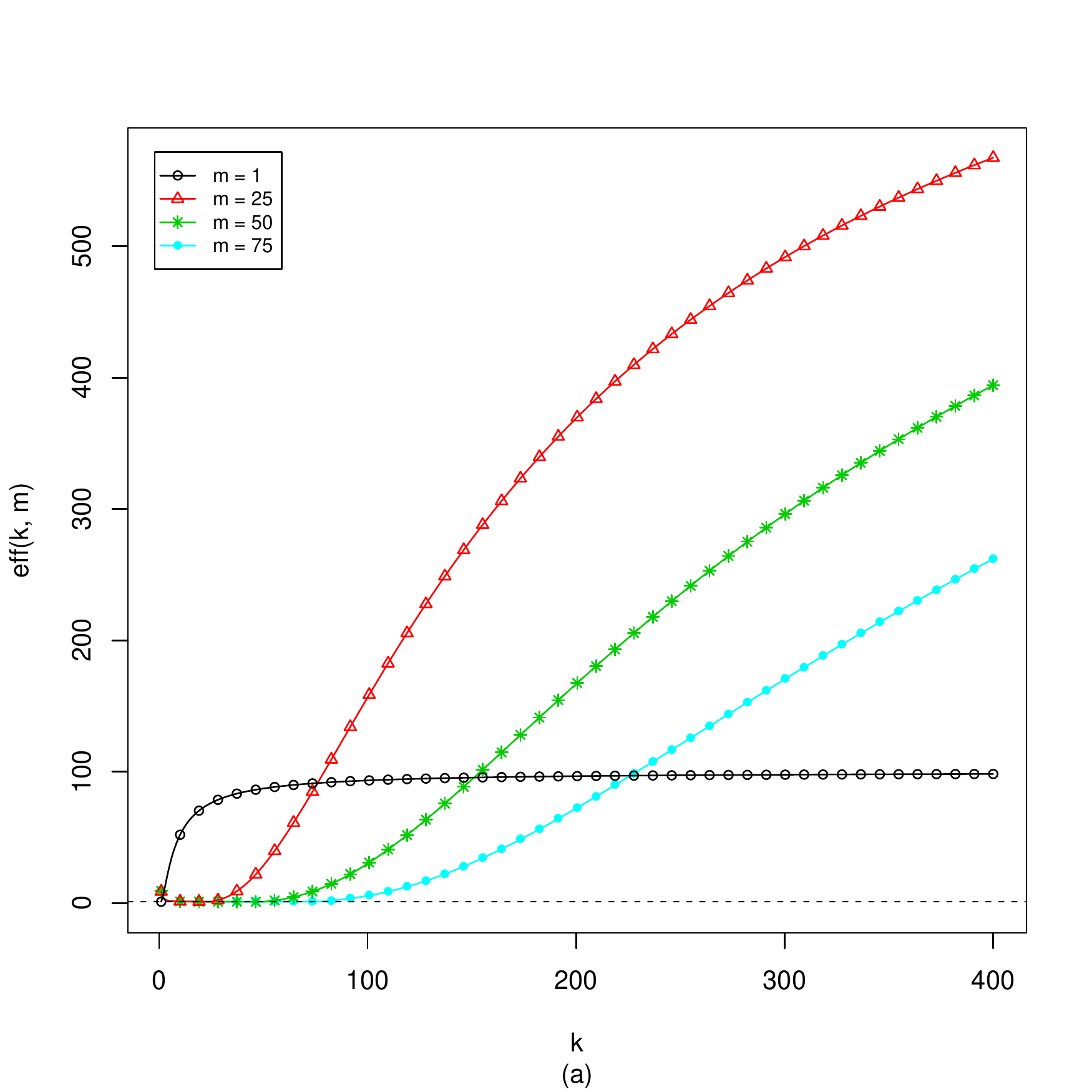}
\includegraphics[angle=0,width=0.49\linewidth]{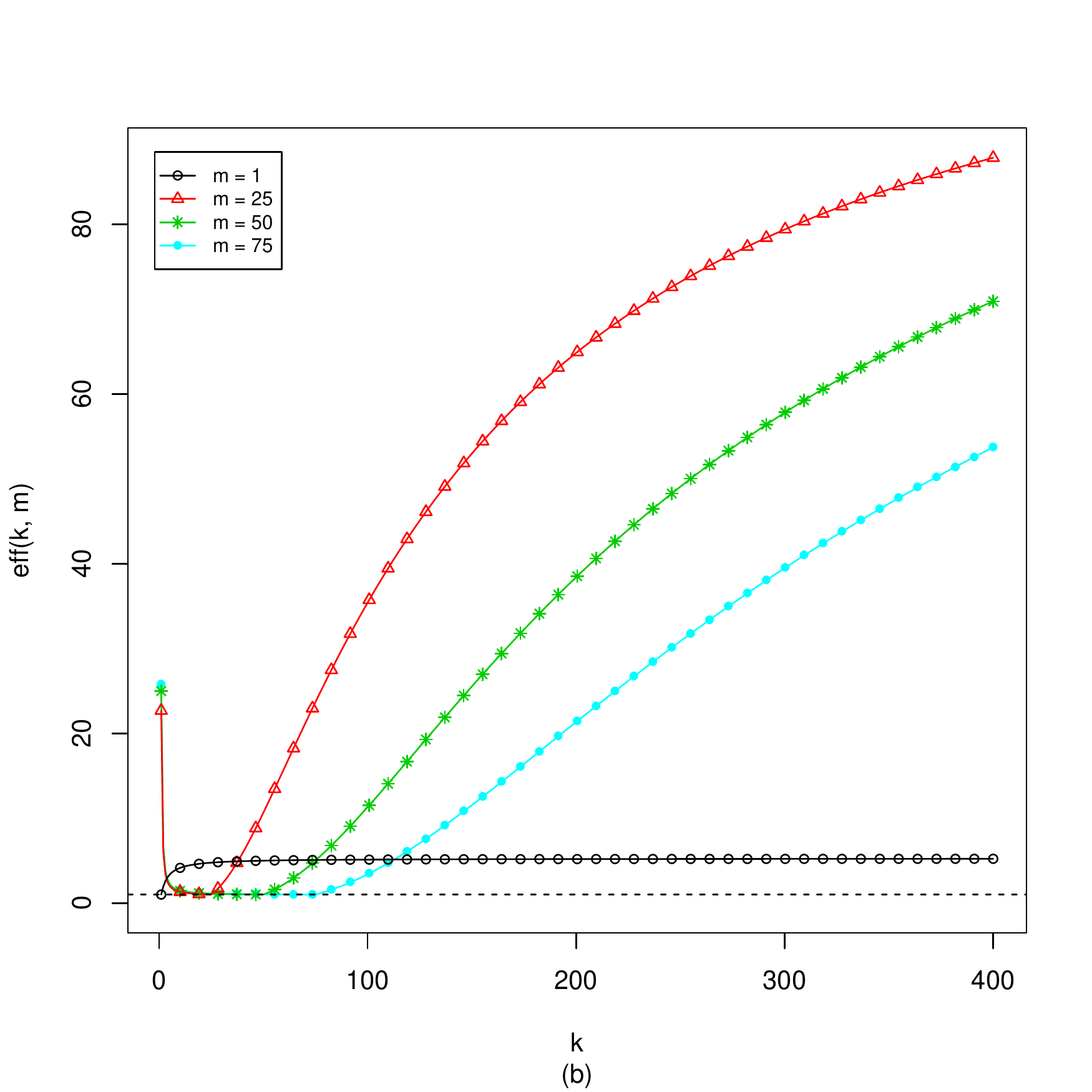}
\caption{Relative efficiency (Theorem \ref{theo:nonlin-a1} and Corollary \ref{corol:nonlin-aa}) $ \eff(k, m) $, for $ 1 \leq k \leq 400 $, using $ n = 1000 $ and the regularity parameter $ m = 1, 25, 50, 75 $. The horizontal dashed line is a reference where $ \eff(k,m) = 1 $. Figures (a) and (b) correspond to the $ \eff(k, m) $ in the cases $ a = 1 $ and $ a \neq 1 $, respectively.}
\label{fig:asseff}
\end{figure}

\subsection{A few special cases}\label{sec:special-cases}

There are infinitely many estimators for the density $ f $ based on \eqref{est:linear}, because of many choices for the power $ a $ and the warping function $ H $. When $ H(x) = x $, $ 0 < x < 1 $, the case where $ a = 1 $ represents the standard estimator \citep[see, e.g.][]{Ramirez.Vidakovic-2010-JSPI, Cutillo.etal-2014-JoSPaI}. Still in the case of an identity warping function, $ a = 1/2 $ can be seen as a natural generalization of \cite{Pinheiro.Vidakovic-1997-CSDA} to the context of biased data. On the other hand, one can explore the possibility of warping the wavelet basis, taking into account some different function $ H $. This can be seen as an attempt to improve the estimates, specially in regions where the weighting function is closer to zero. An interesting case corresponds to $ H(x) = G(x) $, i.e., we consider the \cdf\ of the biased data. Since, as mentioned before, $ G $ (and, consequently, $ g $) is unknown, and requires estimation. A natural candidate is the empirical \cdf, which we denote by $ \hat{G} $. With respect to $ h $, a natural candidate is $ \hat{g} $, the same estimator employed for $ g $. Therefore, the wavelet coefficient estimator in \eqref{est:linear} becomes
\[ \hat{c}_{J_0k} = \dfrac{\hat{\mu}^a}{n} \sum_{i = 1}^{n} \dfrac{\phi_{J_0k}\co{\hat{G}(Y_i)} \hat{g}^{a}(Y_i)}{w^a(Y_i)}. \]

%In this work we consider four cases to be explored, which will represent a method of estimation.
In this work we consider four cases to be explored (see Section \ref{sec:numerical}). These methods of estimation correspond to
\begin{enumerate}
\item[$ \mat{m_1} $:] when $ a = 1/2 $ and $ H(x) = x $;
\item[$ \mat{m_2} $:] when $ a = 1 $ and $ H(x) = x $;
\item[$ \mat{m_3} $:] when $ a = 1/2 $ and $ H(x) =\hat{G}(x) $;
\item[$ \mat{m_4} $:] when $ a = 1 $ and $ H(x) = \hat{G}(x) $.
\end{enumerate}

The procedure of regularization is analogous to that discussed in Section \ref{sec:reg-ests}. Hereafter, we refer to $ m_k $, $ k = 1, 2, 3, 4 $ as the regularized method.

\subsection{Computational aspects}\label{sec:com-asp}

In the real world the range of data is seldom the unit interval. We briefly discuss here how we employ the proposed wavelet estimators the domain is arbitrary. For appropriate $ q $ and $ s $, take $ \cY = (Y-q)/s $ and $ \cX = (X - q)/s $, with densities $ g_{\cY} $ and $ f_{\cX} $, respectively. Hence, it is easy to see that $ f(x) = f_{\cX}\pa{(x-q)/s}/s $ and, after some algebra, one can derive
\begin{eqnarray*}
\hat{f}_{J_0}(x) & = & \frac{1}{s}\co{\hat{f}_{\cX}^a(\cx)}^{1/a}, \\
\hat{f}_{\cX}^a(\cx) & = & \sum_{k = 0}^{2^{J_0} - 1} \hat{c}_{J_0k}^{\circ} \phi_{J_0k}\co{H(\cx)}, \\
\hat{c}_{J_0k}^{\circ} & = & \dfrac{\hat{\mu}^a}{n} \sum_{i = 1}^{n} \dfrac{\phi_{J_0k}\co{H(\cY_i)} \hat{g}_{\cY}^{a-1}(\cY_i) h(\cY_i)}{w^a(Y_i)}.
\end{eqnarray*}
It is important to mention that $ \hat{\mu} $ and $ w $ are associated with $ Y_i $'s, and not with $ \cY_i $'s.

The periodized wavelets impose a periodic analysis to the function of interest. Therefore, a solution is transforming the data into $ [\epsilon, 1-\epsilon] $,
for some $ 0 \leq \epsilon < 1 $. We denote the ordered sample $ y_{(1)} < y_{(2)} < \ldots < y_{(n)} $, and  $ r = y_{(n)} - y_{(1)} $. The adequate constants are given then by $ q = y_{(1)} - \epsilon r/(1 - 2\epsilon) $ and $ s = r $. Following \cite{Montoril.etal-2019-SB}, we consider $ \epsilon = 1.9^{-J_1} $.

\begin{remark}
The transformation of the data, as described above, has impact only in the cases where the wavelet bases are not warped. In fact, the empirical \cdf\ will always be $ k/n $, $ k = 1, \ldots, n $, for the observed sample.
\end{remark}

Regularization is performed analogously to Section \ref{sec:reg-ests}.

\section{Numerical studies}\label{sec:numerical}

We present now some Monte Carlo simulations as well as a real data application. The dataset is not equally spaced. With the exception of the Haar basis, compactly supported orthonormal wavelets
do not posses analytic expressions, so we need some numerical interpolation. We employ the Daubechies-Lagarias algorithm \citep{Daubechies.Lagarias-1991-SJMA, Daubechies.Lagarias-1992-SJMA}, which can attain any preassigned precision \citep[see, e.g.][]{Vidakovic-1999}.  Analyses are performed with Symmlets S10 (Daubechies least asymmetric 20-tap filter).

We consider methods of estimation $ m_k $, $ k = 1,2,3,4 $, as presented in Section \ref{sec:special-cases}.
This gives us an idea of how the density's square root estimate ($a=1/2$) can improve the ordinary approach
($ a = 1 $), as well as if a warped wavelet basis can provide a better performance. For the methods $ m_1 $ and $ m_3 $, we estimate $ g $ by a Gaussian kernel with bandwidth selected according to \cite{Sheather.Jones-1991-JRSSB}\comment{Other methods, such as linear and regularized wavelet-based estimates, were tried to estimate $ g $. However, the final estimates of the densities of interest were not as good as the ones provided when using the adopted kernel method.}. For the methods $ m_3 $ and $ m_4 $, the wavelet basis is warped by the empirical \cdf\ of the data, linearly interpolated.

Regularization is done by the universal hard threshold, i.e.,\linebreak $ \lambda = \hat{\sigma} \sqrt{2\log 2^{J_1-1}}$, where $ \hat{\sigma} $ is the median absolute deviation of the detail coefficients in the finest resolution level \citep{Donoho.Johnstone-1994-B}.

\subsection{Simulation studies}\label{sec:simul}

The performance of the estimation methods is evaluated by Monte Carlo simulation studies. For such a task, we consider three different examples described below.

\begin{proof}[Example 1]
We assume that $ X \sim \text{Beta}(2.5, 2.5) $ and $ w(y) = y^{-2}(1-y)^{-2} $. In this case, $ Y \sim \text{Beta}(0.5, 0.5) $. Therefore, the shapes of $ f $ and $ g $ are ``inverted'', as it can be observed in the first row of Figure \ref{fig:densities}.
\end{proof}

\begin{proof}[Example 2]
Let us denote by $ \beta(x, a, b) $ the density of a beta distribution with parameters $ a $ and $ b $ evaluated at $ x $. Thus, in this example we consider a mixture of three betas for the density of interest:
\[ f(x) = (1/3)\beta(x, 20, 3) + (1/3)\beta(x, 40, 40) + (1/3)\beta(x, 3, 20).  \]
The weighting function is $ w(y) = y $. This results in a biased sample from the density
\[ g(y) = (40/69)\beta(y, 21, 3) + (1/3)\beta(x, 41, 40) + (2/23)\beta(x, 4, 20). \]
In this example, the biased density remains a mixture of betas, but now with ``unbalanced'' weights, as presented in the second row of Figure \ref{fig:densities}.
\end{proof}

%\subsubsection*{Example 3 ($ ex_1 $)}

\begin{proof}[Example 3]
In this example we consider a piecewise linear density for $ f $, where
\[ f(x) = \begin{cases}
\dfrac{64x + 1}{9}, & 0 \leq x < 0.25, \\
\dfrac{32(1 - 2x) + 1}{9}, & 0.25 \leq x < 0.5, \\
\dfrac{x(32x - 31) + 12}{9}, & 0.5 \leq x < 0.75, \\
\dfrac{x(65 - 32x) - 24}{9}, & 0.75 \leq x \leq 1.
\end{cases} \]
As biasing function, we employed in this example $ w(y) = 0.1 + 2x^2 $. We do not present the cumbersome resulting biased density of $ Y $. However, it can be seen in the third row of Figure \ref{fig:densities} that the biased density presents a different shape (still not smooth, but no longer piecewise linear). This provides a challenge to estimate $ f $. Data values are numerically generated by accept-reject algorithm \cite[Section 3.6]{Efromovich-1999}.
\end{proof}

\begin{figure}[!h]
\centering
\includegraphics[angle=0,width=0.39\linewidth]{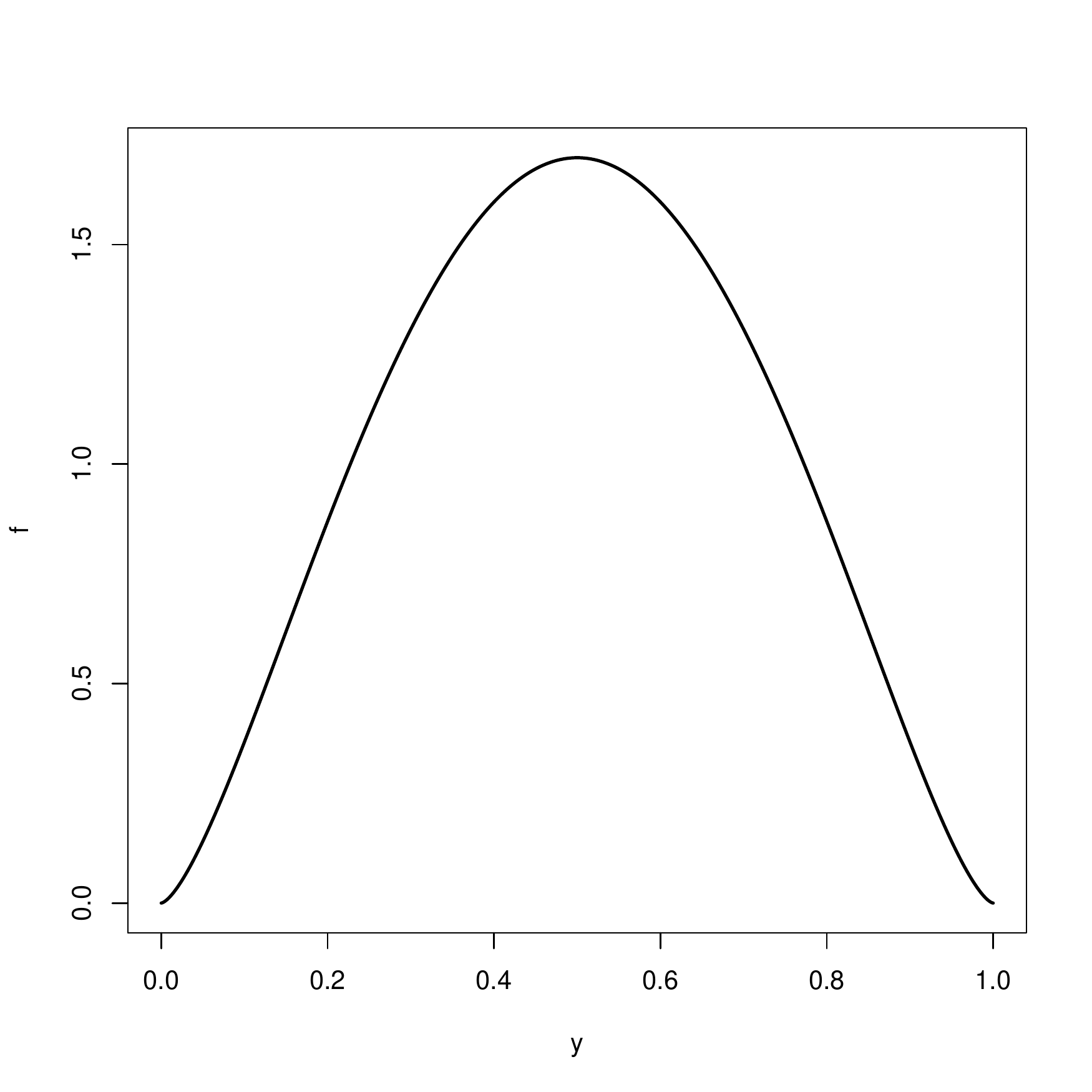}
\includegraphics[angle=0,width=0.39\linewidth]{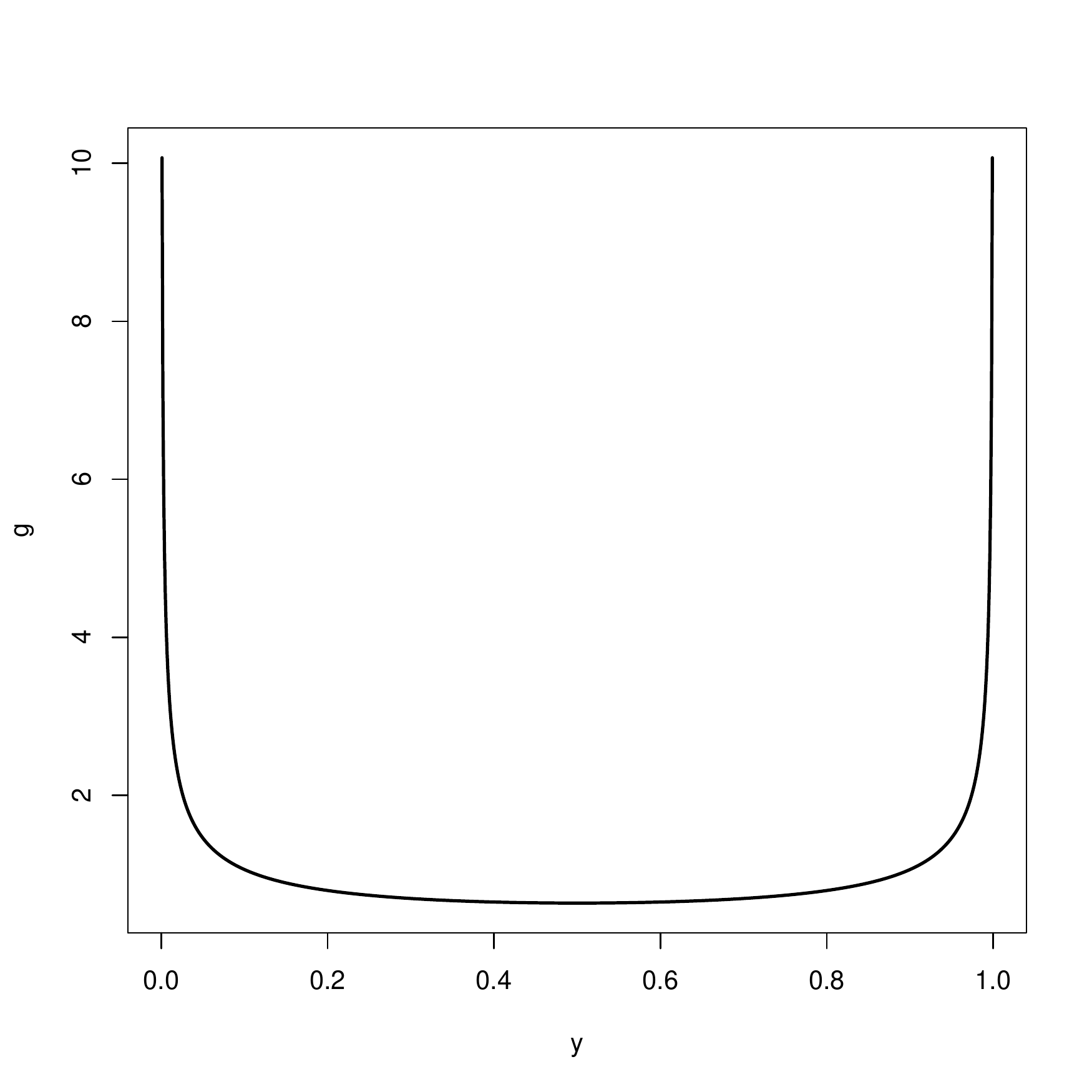} \\
\includegraphics[angle=0,width=0.39\linewidth]{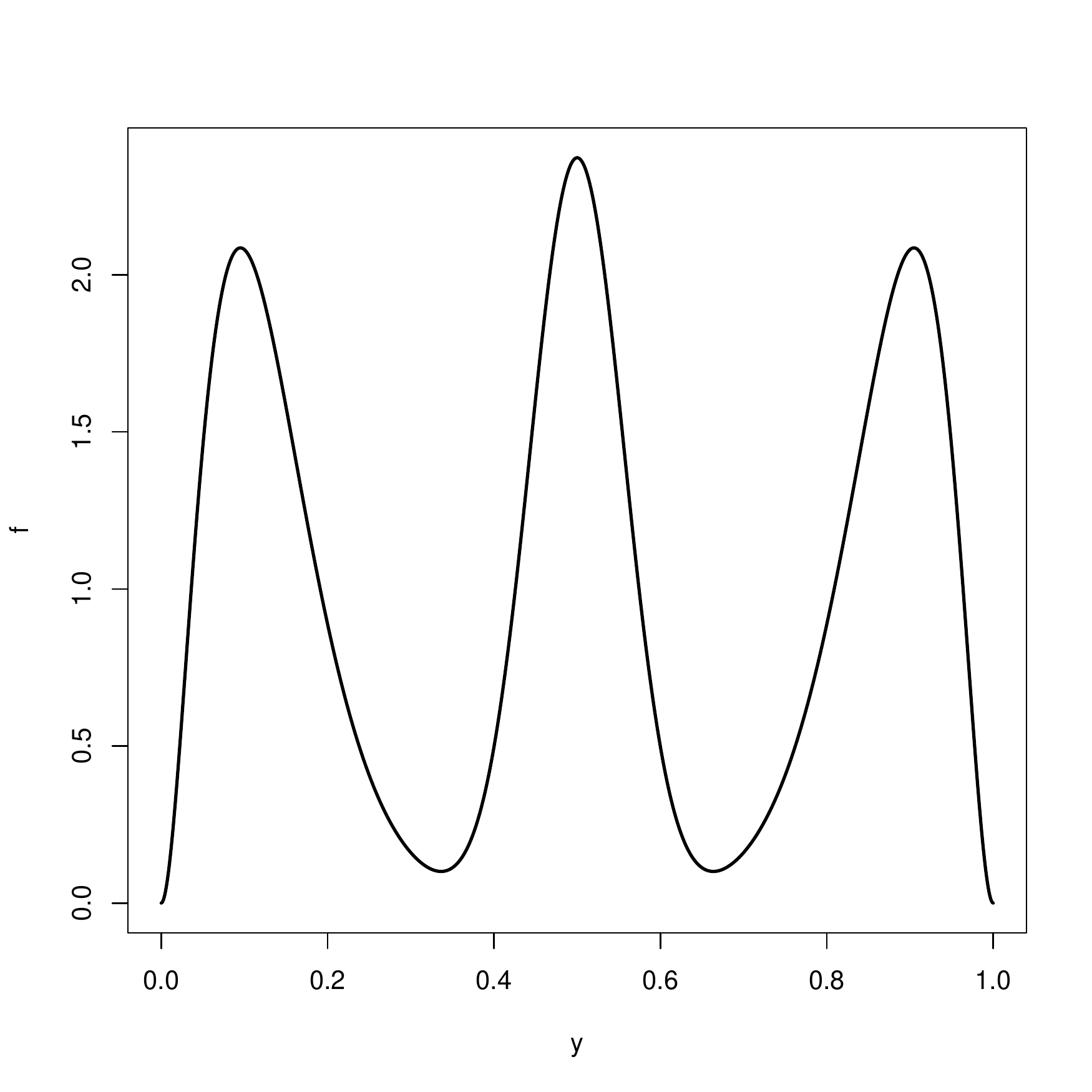}
\includegraphics[angle=0,width=0.39\linewidth]{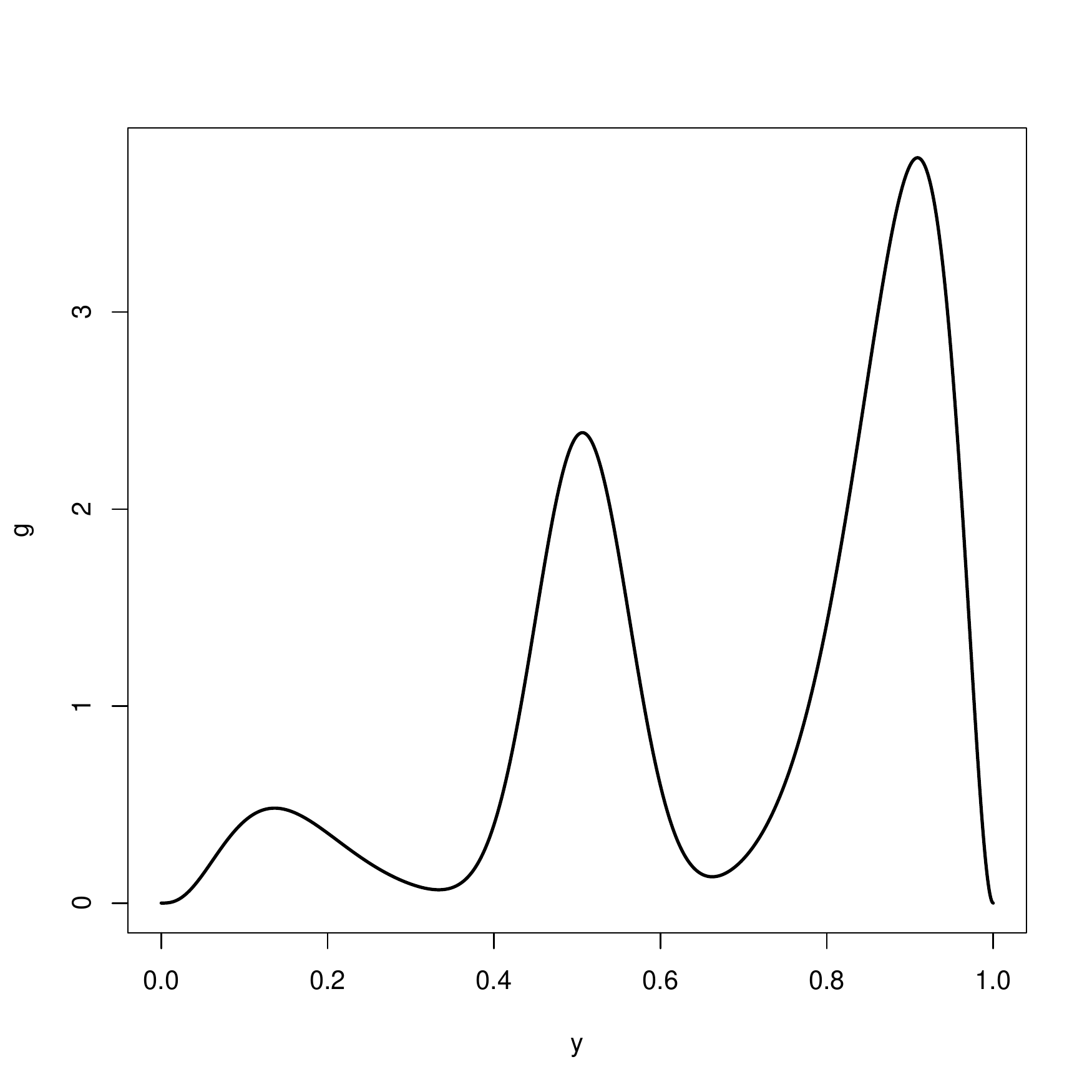} \\
\includegraphics[angle=0,width=0.39\linewidth]{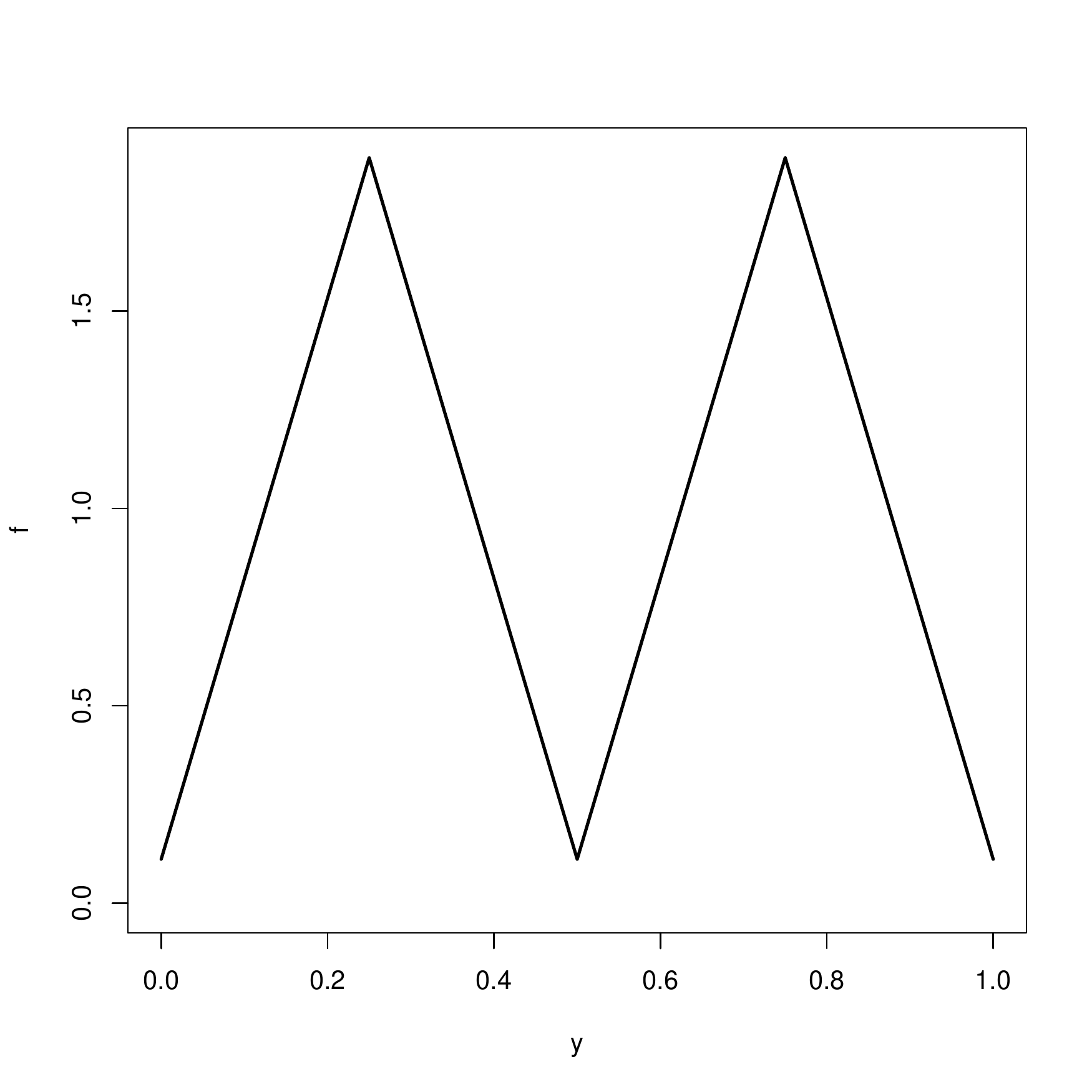}
\includegraphics[angle=0,width=0.39\linewidth]{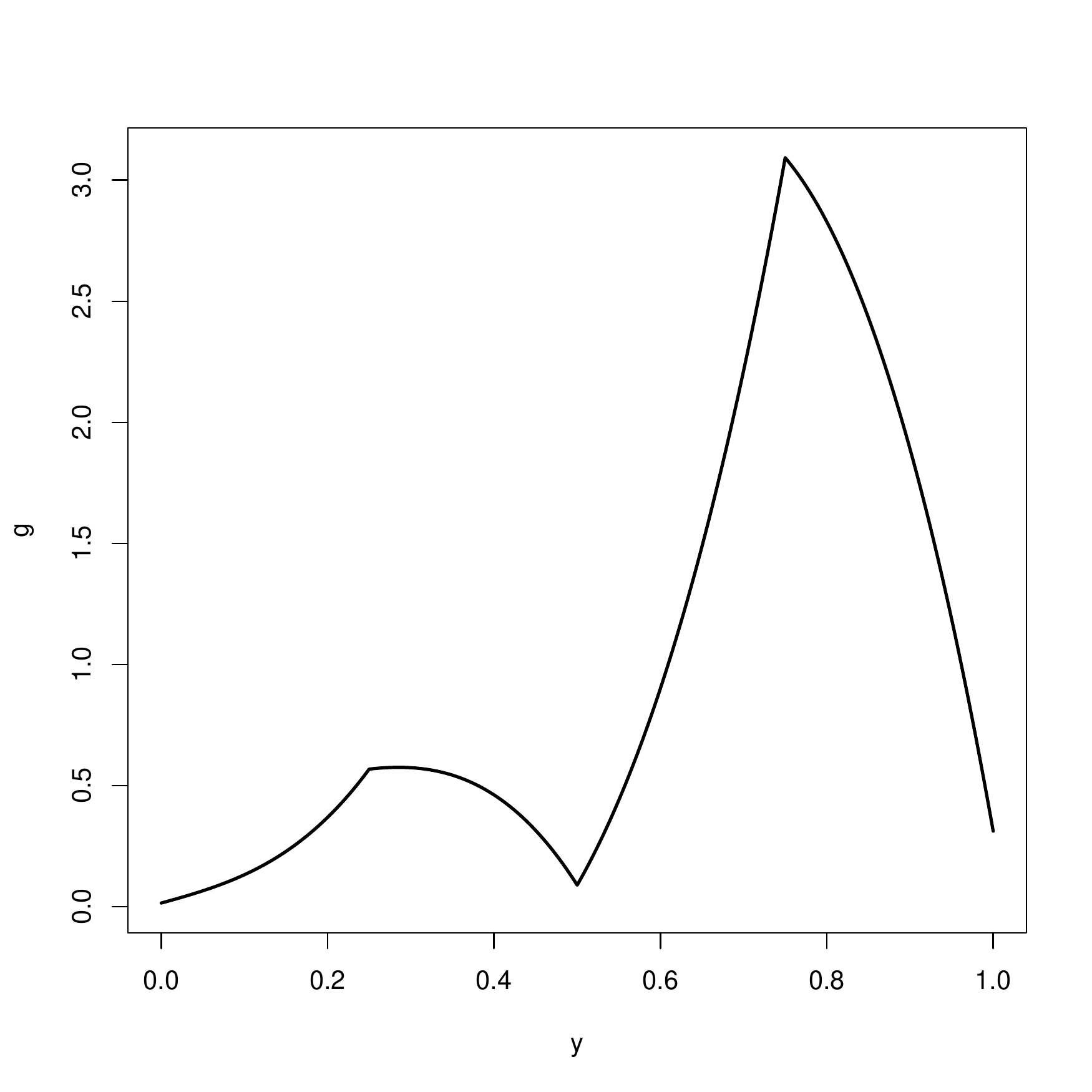}
\caption{Densities used in the simulations (Examples $ \text{ex}_1 $-$ \text{ex}_3 $). The first and second columns represent the densities $ f $ and $ g $, respectively. The $ i $-th row corresponds to the simulation example $ \text{ex}_i $, $ i = 1, 2, 3 $.}
\label{fig:densities}
\end{figure}

The examples above will be denoted hereafter by $ \ex_1 $, $ \ex_2 $ and $ \ex_3 $. We generate $ 1,000 $ biased samples with sizes $ n = 250, 500, 750, 1000 $. For the finest resolution level, we consider the cases $ J_1 = \ceil{p \log_2n} $, where $ p = 0.20, 0.45, 0.70, 0.95 $ and $ \ceil{x} $ represents the smallest integer greater than or equal to $ x $. We adopt as coarsest resolution level $ J_0 = 0 $. %In this work we use the hard threshold to regularize the estimates.

We numerically evaluate the estimate's closeness to the real function of interest by the average square error (ASE), defined as
\[ \ase(\hat{f}, f) = \dfrac{1}{n} \sum_{i = 1}^{n_{\text{grid}}} [\hat{f}(x_i) - f(x_i)]^2, \]
where $ x_1, \ldots, x_{n_{\text{grid}}} $ correspond to a grid of $ n_{\text{grid}} = 250 $ equally spaced points inside the unit interval. Since we have $1,000$ samples for each sample size, data range vary. In order to make the ASEs comparable, we consider the maximum among the minimums and the minimum of the maximums of the datasets to represent $ x_1 $ and $ x_{n_{\text{grid}}} $, respectively.

Performance of the finest resolution level candidates can be observed in Figure \ref{fig:jcomp}. For $ m_1 $ and $ m_2 $ (ordinary wavelet basis), these methods tend to provide poor estimates for larger values of $ J_1 $ ($ p = 0.70 $ and $ 0.95 $). In $ \ex_1 $, $ p = 0.20 $ show a performance slightly superior to  $ p = 0.45 $. On the other hand, one sees considerable improvement when changing from $ p = 0.20 $ to $ p = 0.45 $ for $ \ex_2 $. Finally, in $ \ex_3 $, $ p = 0.45 $ provides a slight improvement for the estimates, when compared to $ p = 0.20 $. This suggests that $ J_1 = \ceil{0.45\log_2n} $ seems to be a good choice for the finest level. Furthermore, when comparing these two methods, $ m_1 $ presents the worst estimates, with greater mean and variability, sometimes providing negative density estimates.

\begin{figure}[!htb]
\centering
\includegraphics[angle=0,width=0.24\linewidth]{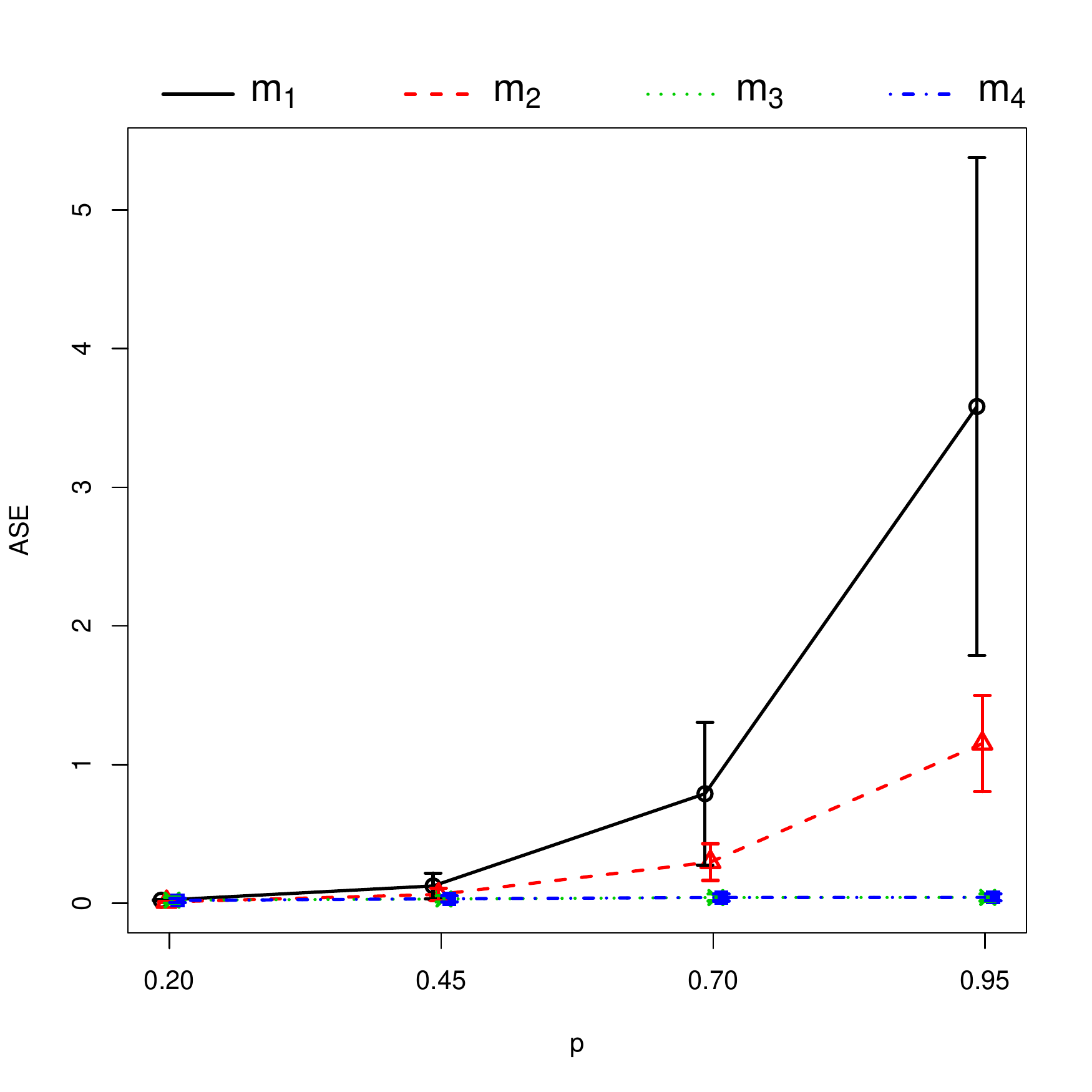}
\includegraphics[angle=0,width=0.24\linewidth]{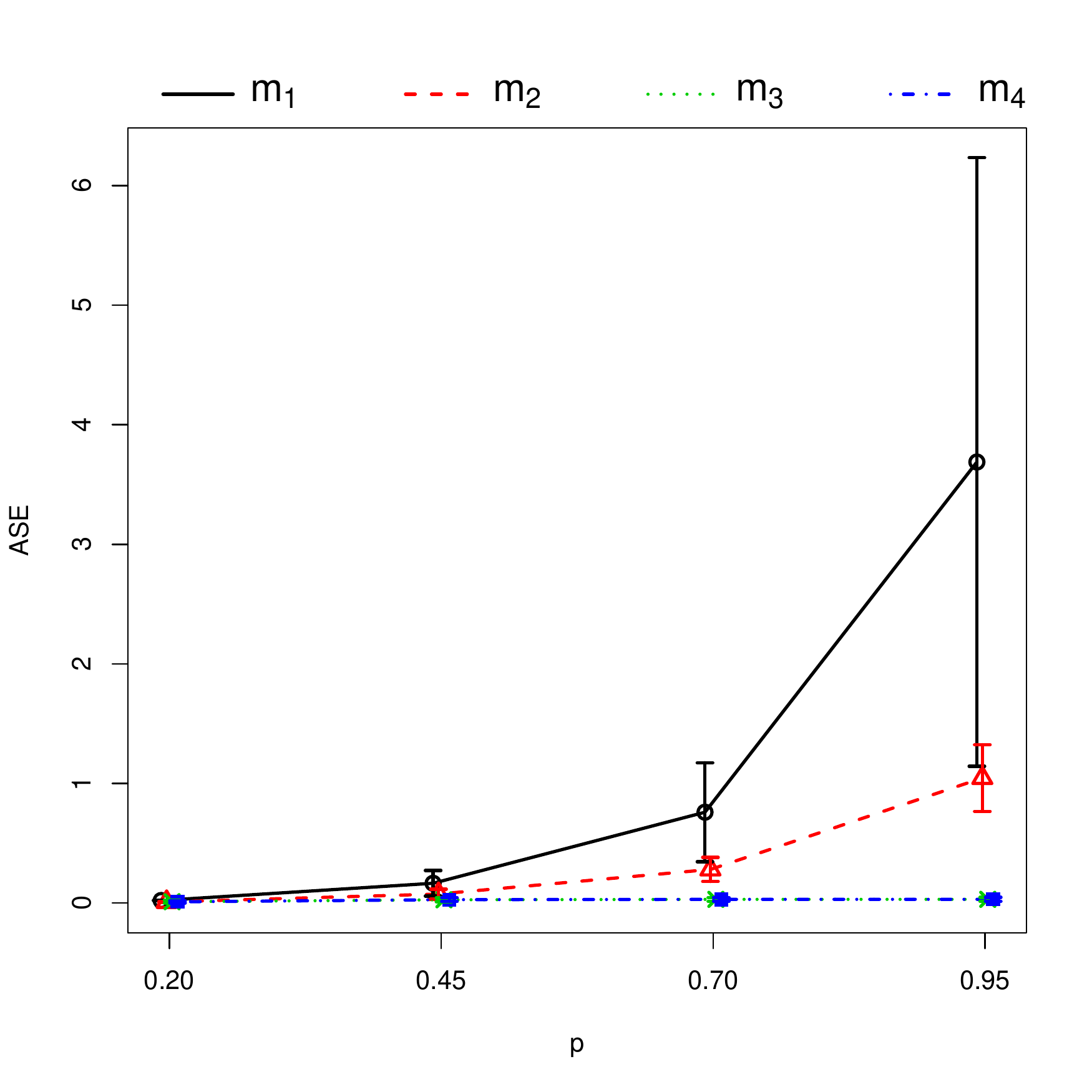}
\includegraphics[angle=0,width=0.24\linewidth]{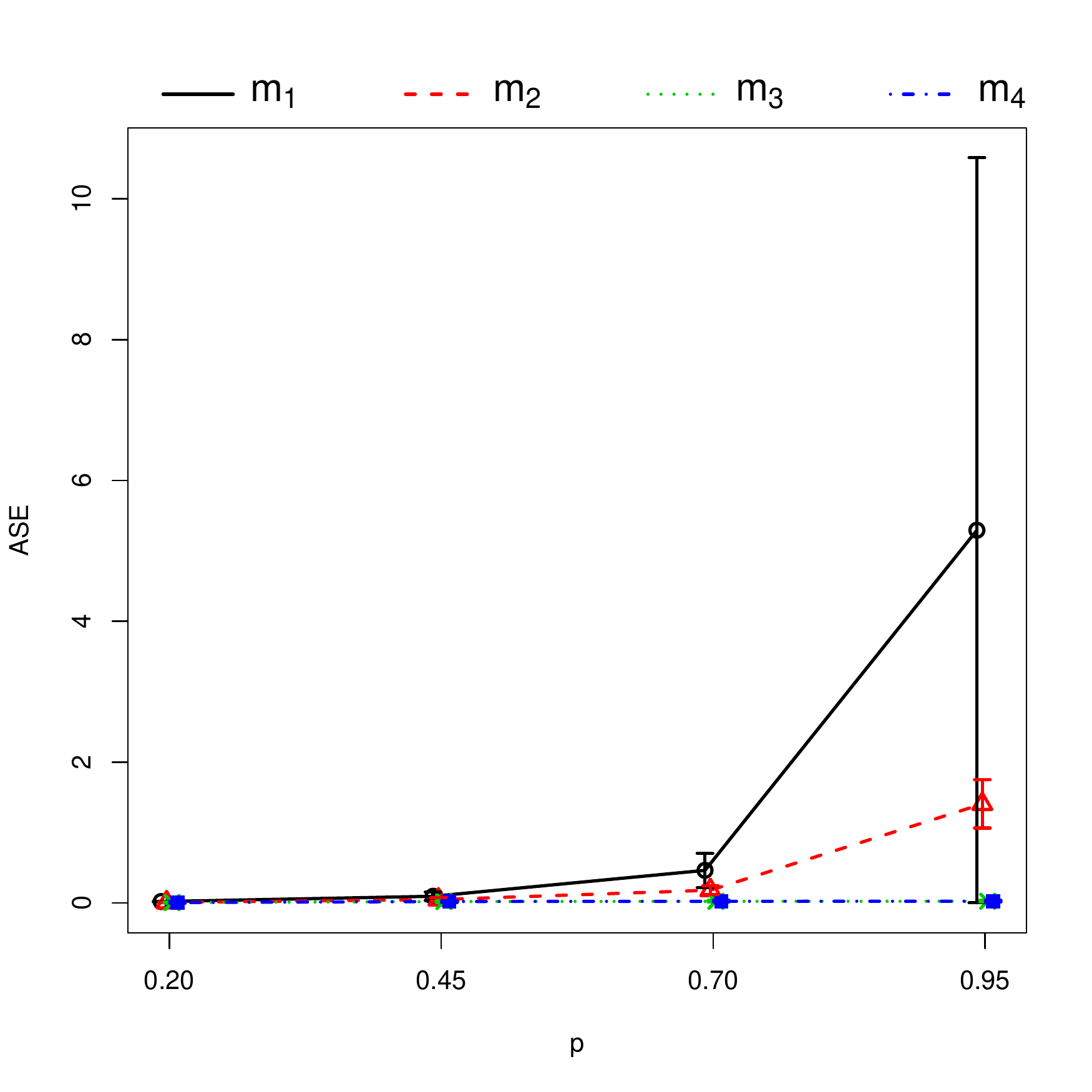}
\includegraphics[angle=0,width=0.24\linewidth]{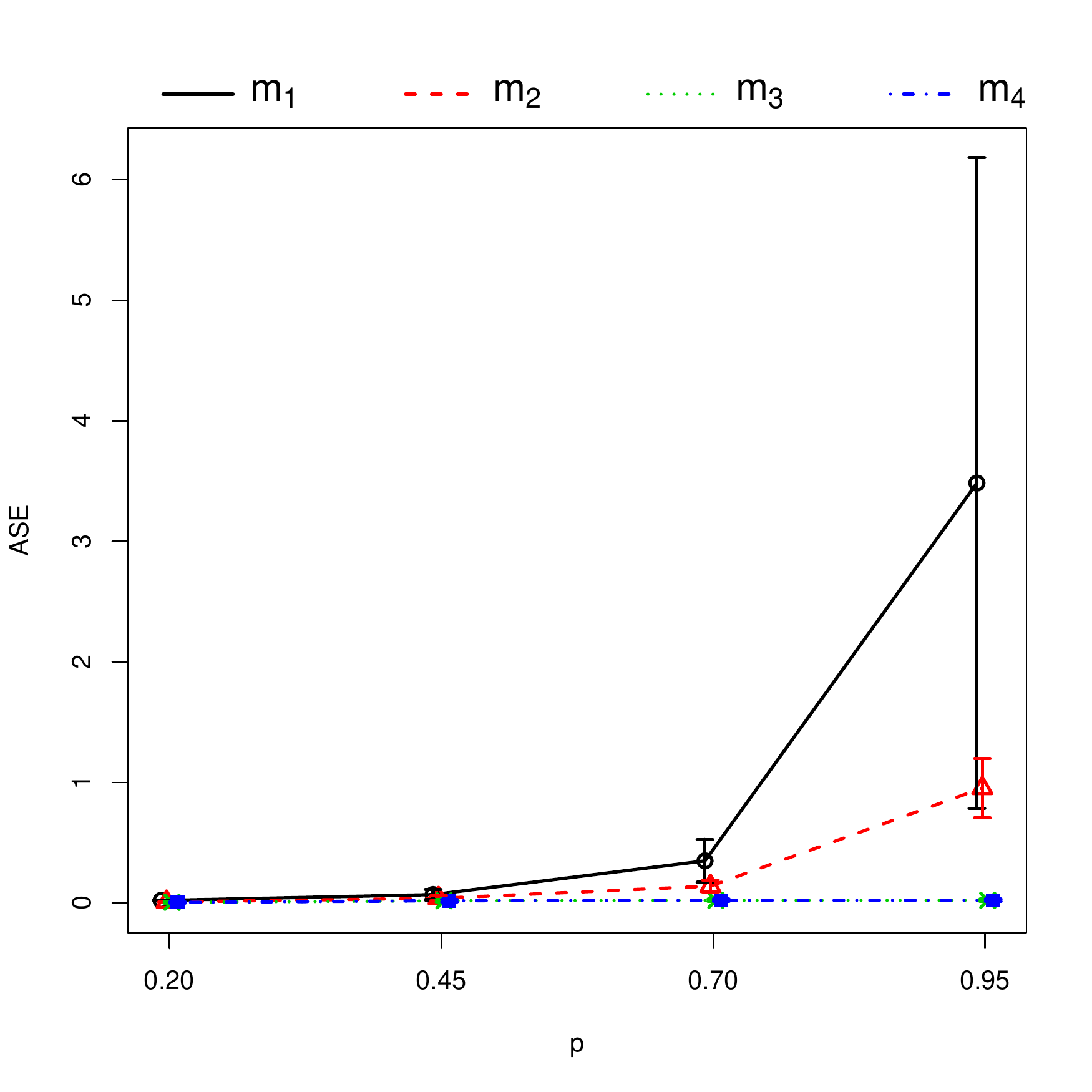} \\
\includegraphics[angle=0,width=0.24\linewidth]{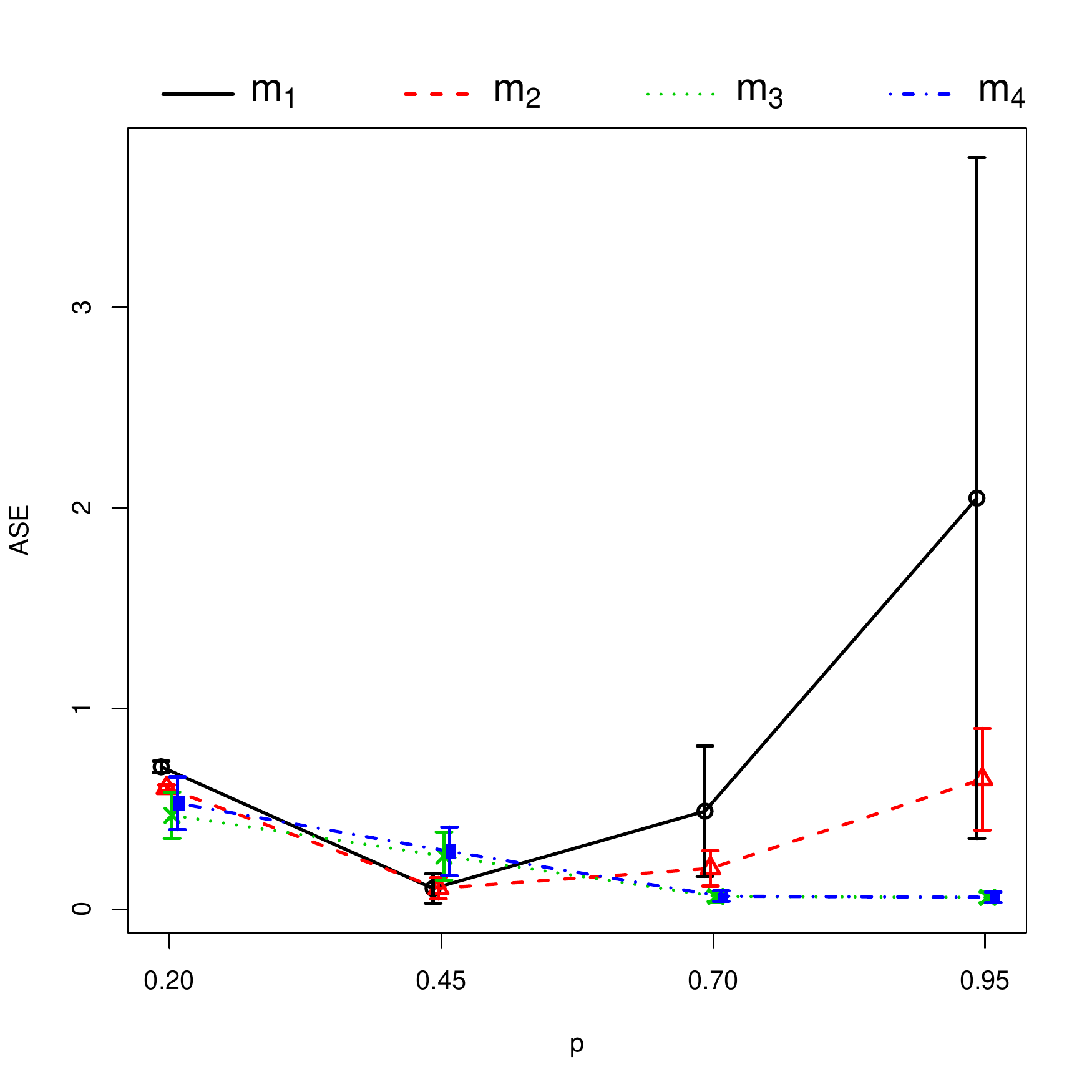}
\includegraphics[angle=0,width=0.24\linewidth]{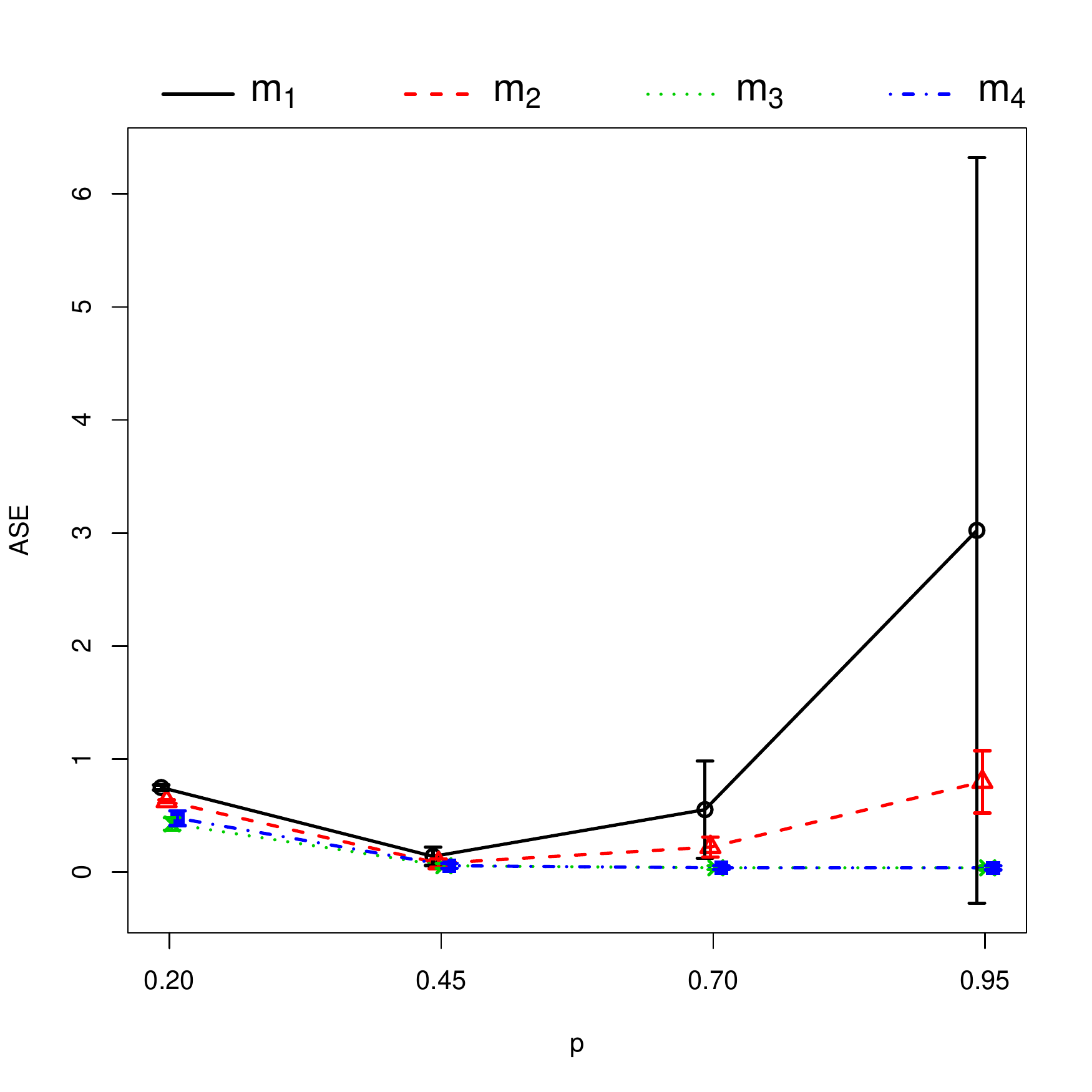}
\includegraphics[angle=0,width=0.24\linewidth]{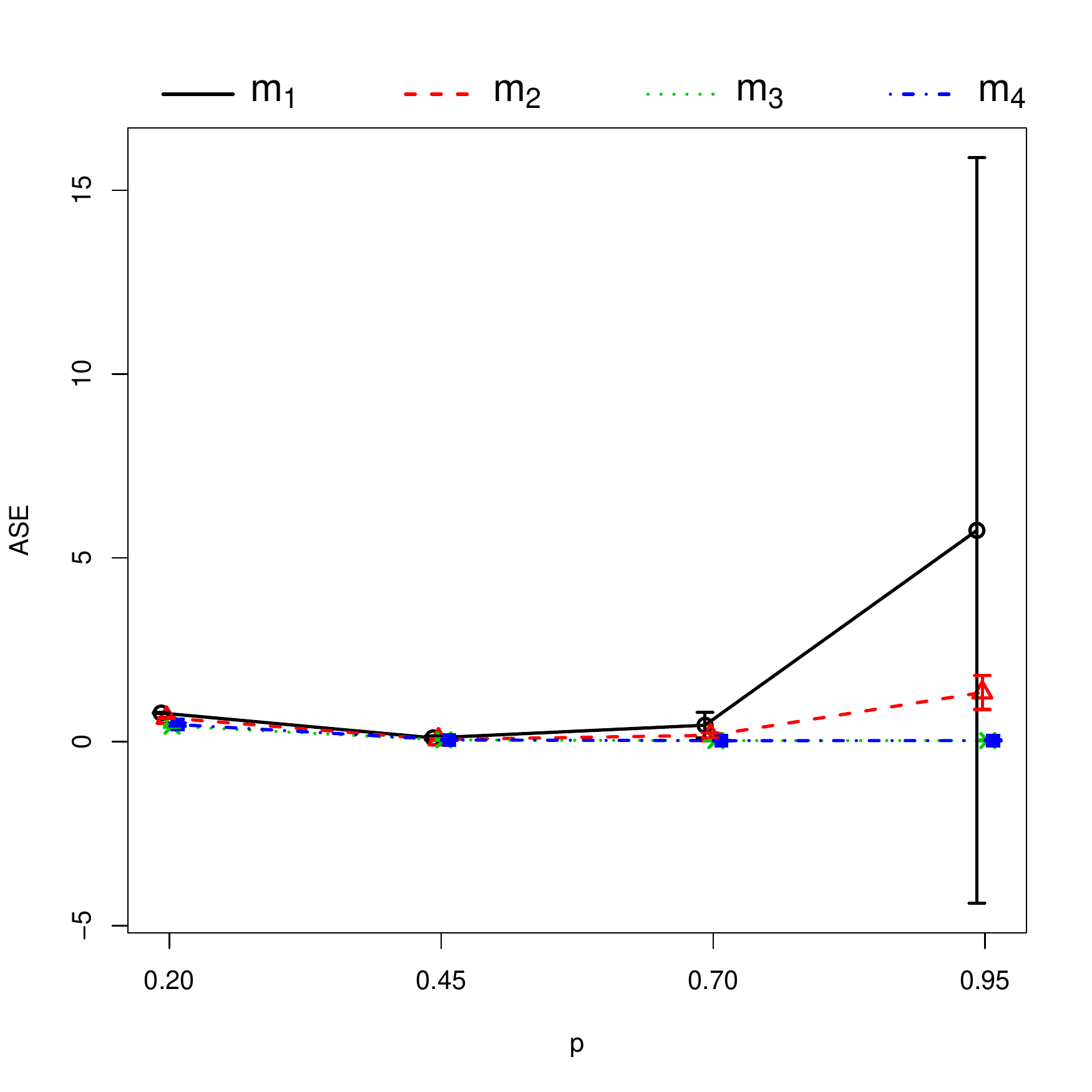}
\includegraphics[angle=0,width=0.24\linewidth]{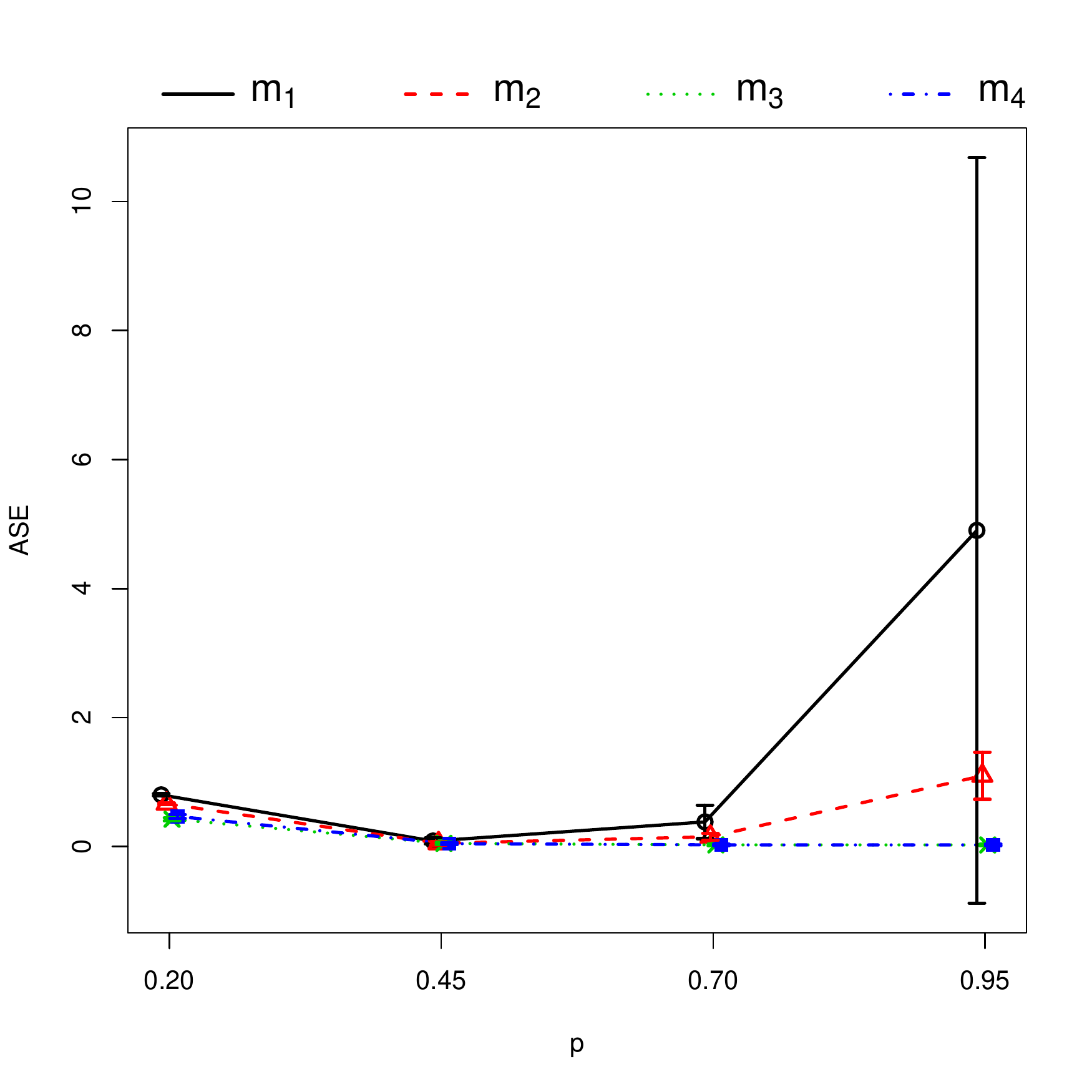} \\
\includegraphics[angle=0,width=0.24\linewidth]{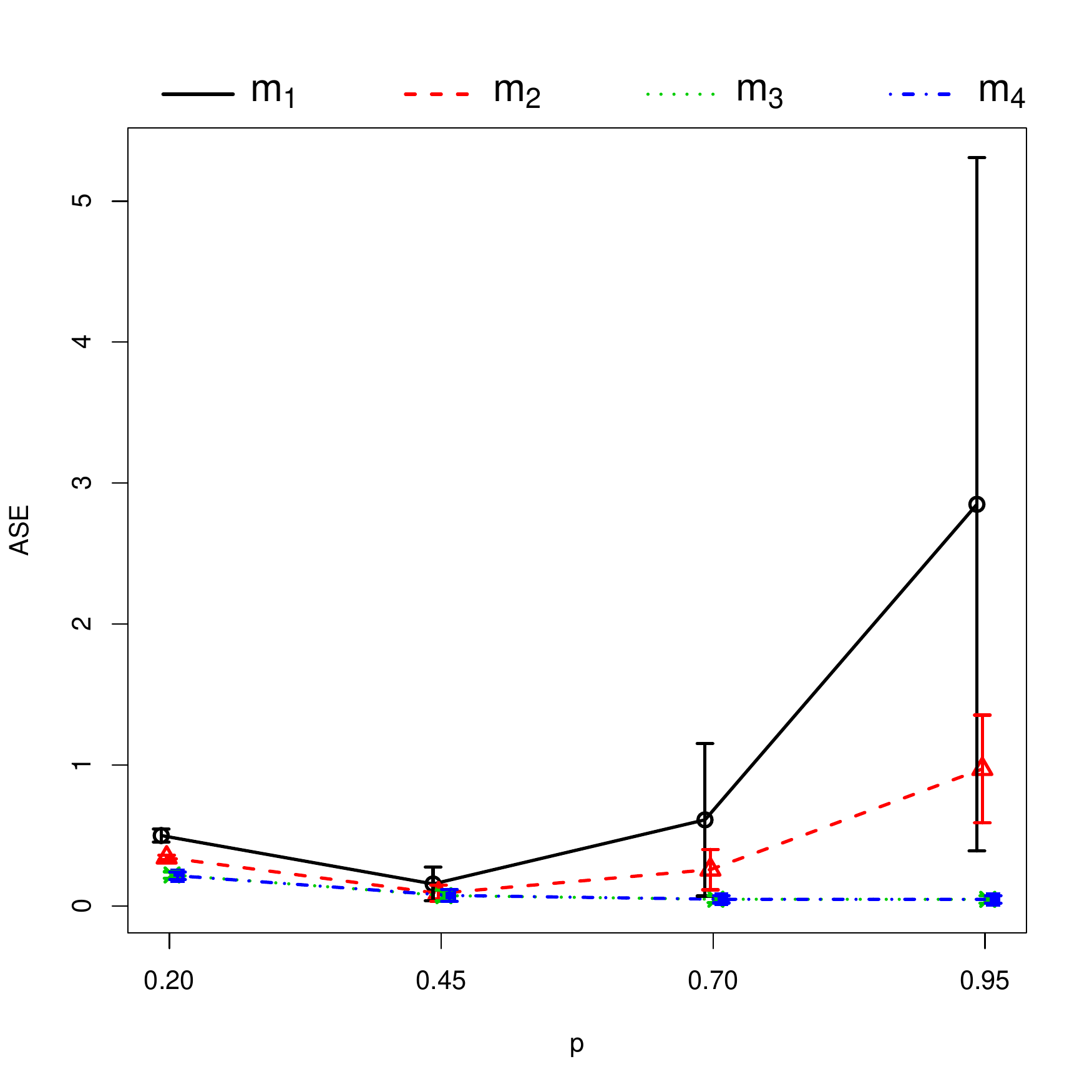}
\includegraphics[angle=0,width=0.24\linewidth]{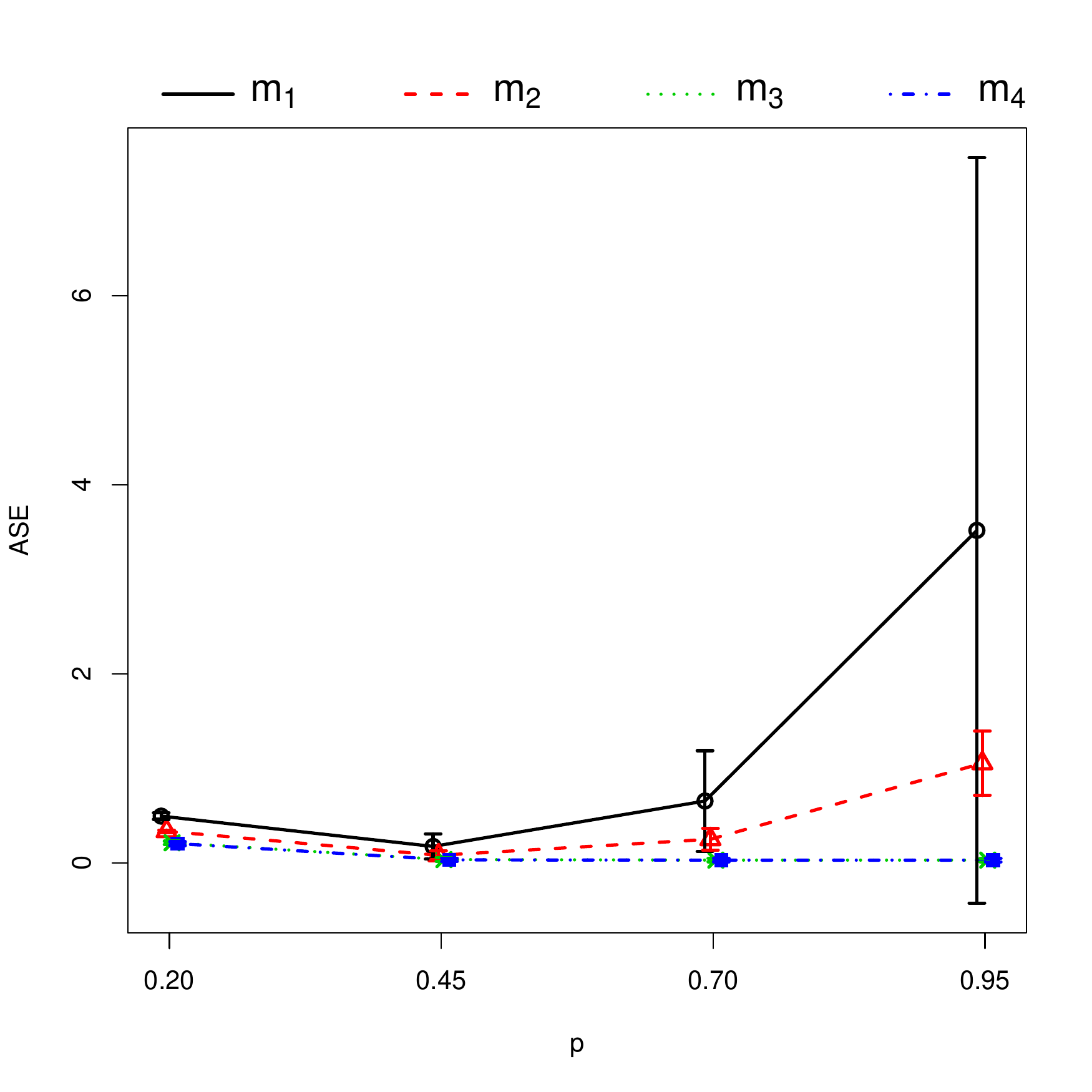}
\includegraphics[angle=0,width=0.24\linewidth]{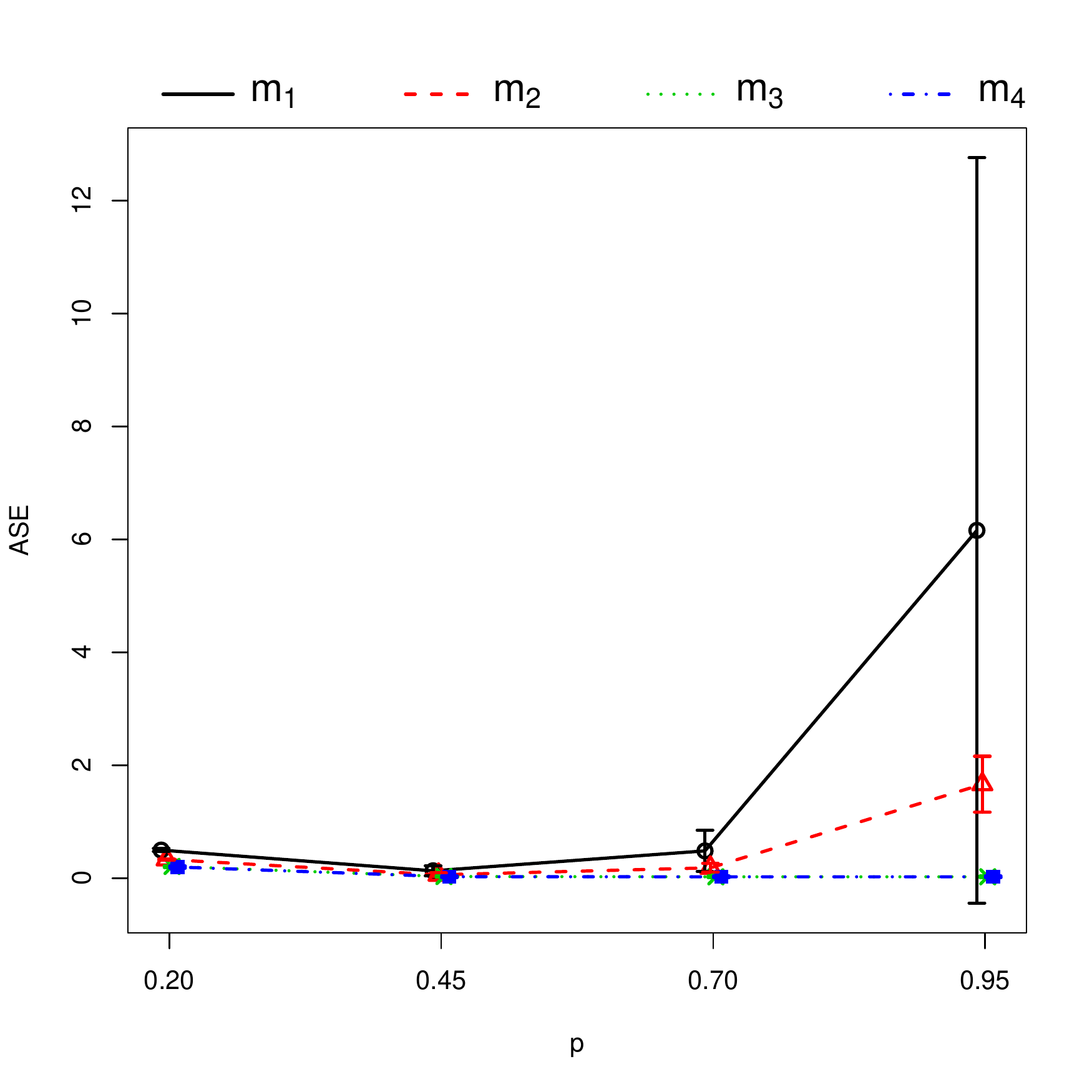}
\includegraphics[angle=0,width=0.24\linewidth]{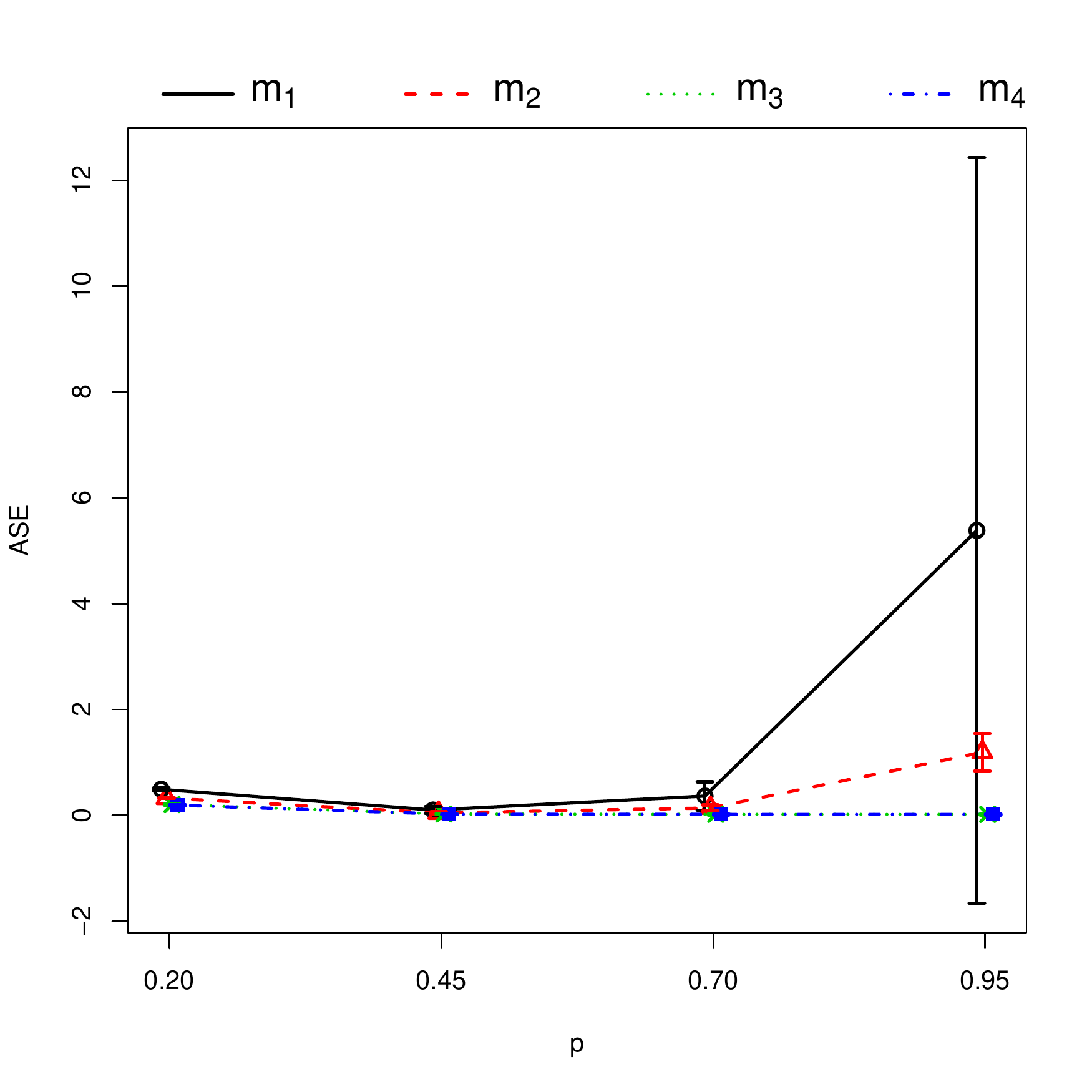}
\caption{Average plus/minus a standard deviation of the ASEs. The full lines, dashed lines, dotted lines and dot-dashed lines represent the ASEs obtained by the methods $ m_1 $--$ m_4 $, respectively. Columns 1--4 are related to the datasets with sample sizes $ n = 250, 500, 750, 1000 $, respectively. The $ i $-th row corresponds to the example $ i $, $ i = 1, 2, 3 $.}
\label{fig:jcomp}
\end{figure}

We also see in Figure \ref{fig:jcomp} that the estimates obtained by the warped wavelet basis ($ m_3 $ and $ m_4 $) present an opposite behavior to $ m_1 $ and $ m_2 $. Larger values of $ J_1 $ yield generally better estimates. This suggests that the finest resolution level for $ J_1 = \ceil{0.95\log_2n} $ is a good alternative for methods of estimation based on the warped basis. Moreover, the estimates of $ m_3 $ and $ m_4 $ look very similar (almost equal). Finally, one can see that the warped basis provide an improvement on the estimates, with average ASEs closer to zero as well as smaller standard deviations.

Figure \ref{fig:ncomp} presents the estimator's performance for fixed resolution levels, as sample size increases. In general, one can see that all the four methods of estimation provide estimates that seem to converge to the density of interest. One clear exception corresponds to the case where $ J_1 = \ceil{0.20 \log_2 n} $ in $ \ex_2 $. This reinforces that such a resolution level is not indicated, which is consistent with some arguments above for Figure \ref{fig:jcomp}. Also, one can clearly see the superiority of estimates based on warped wavelets, which provide ASEs with smaller averages and standard deviations.

\begin{figure}[!ht]
\centering
\includegraphics[angle=0,width=0.24\linewidth]{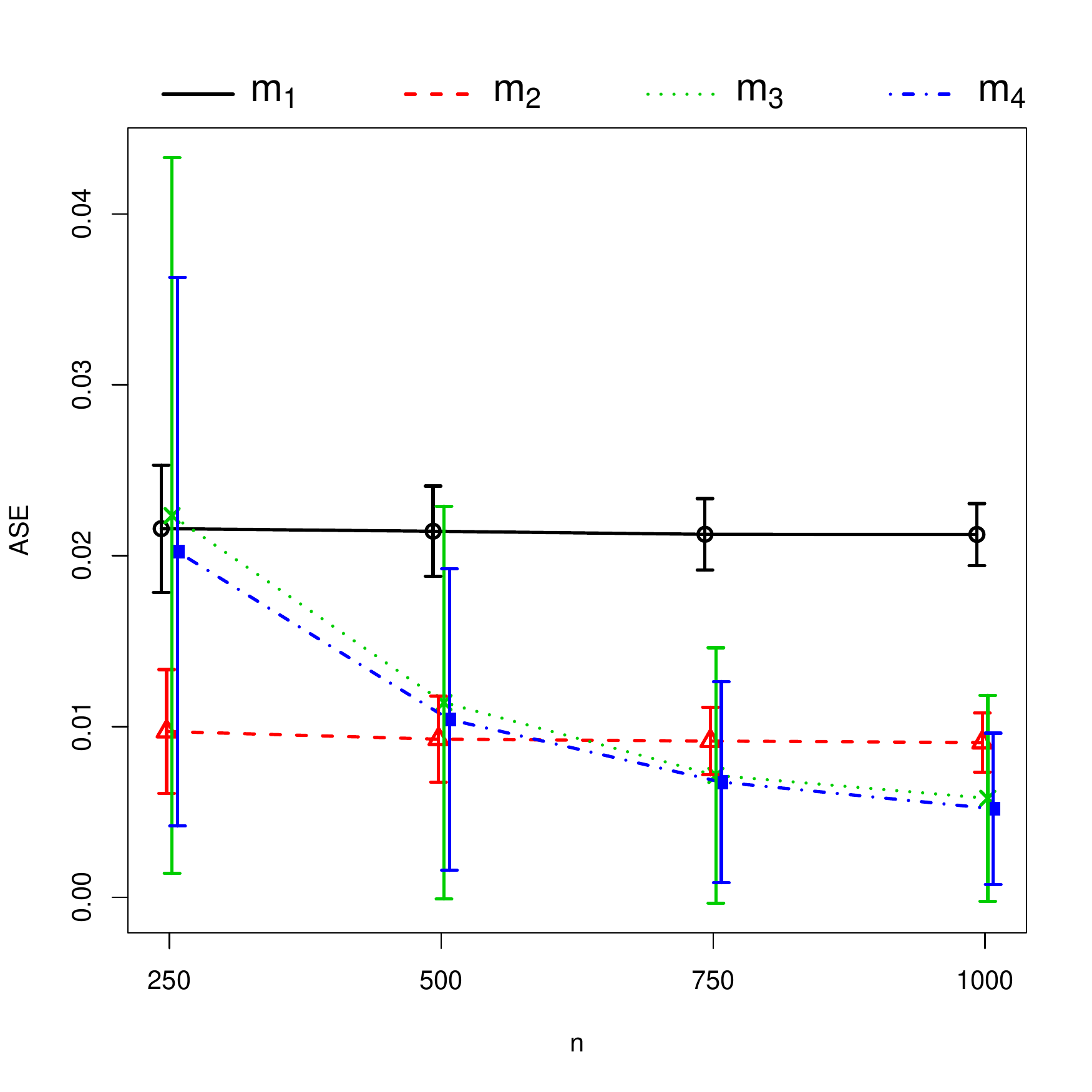}
\includegraphics[angle=0,width=0.24\linewidth]{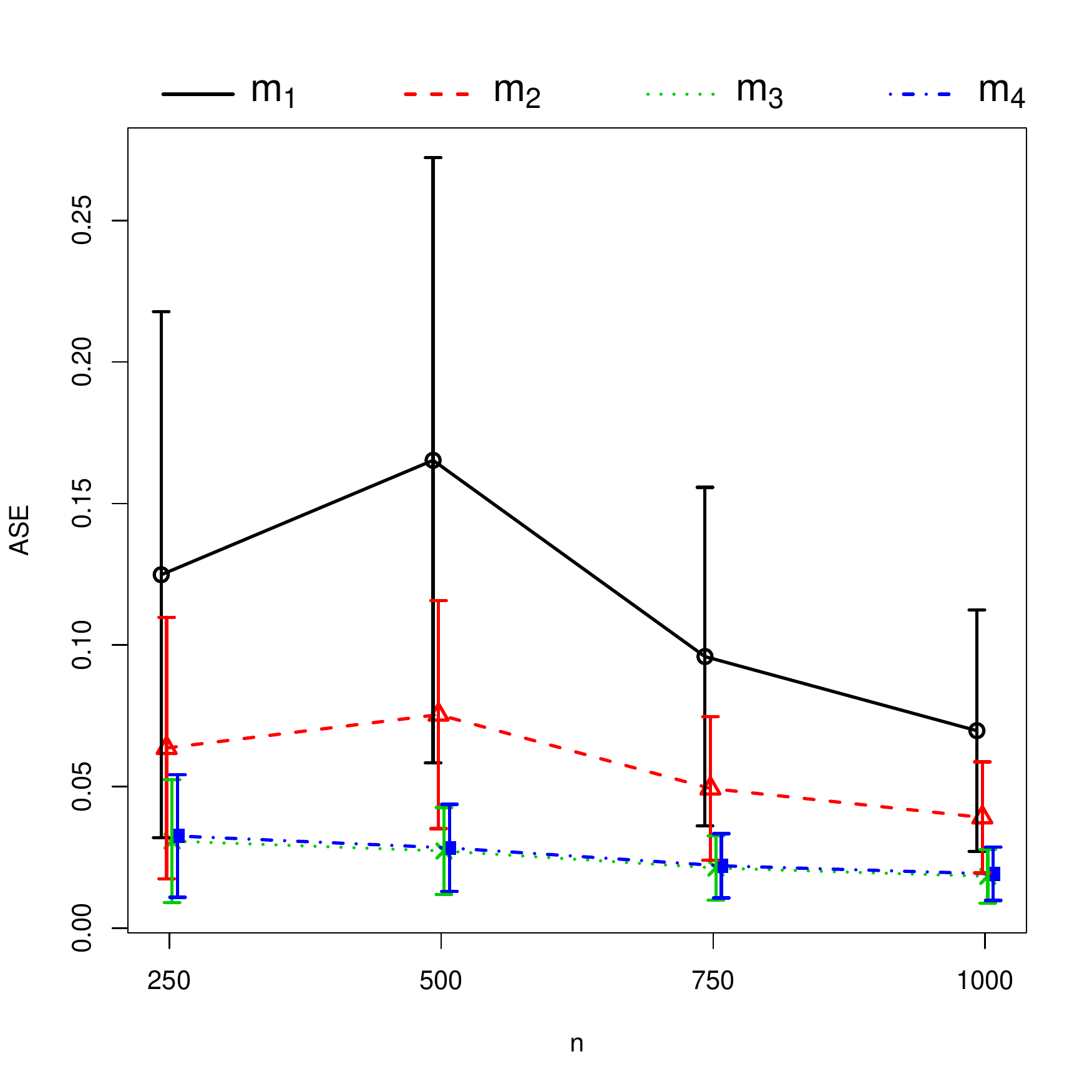}
\includegraphics[angle=0,width=0.24\linewidth]{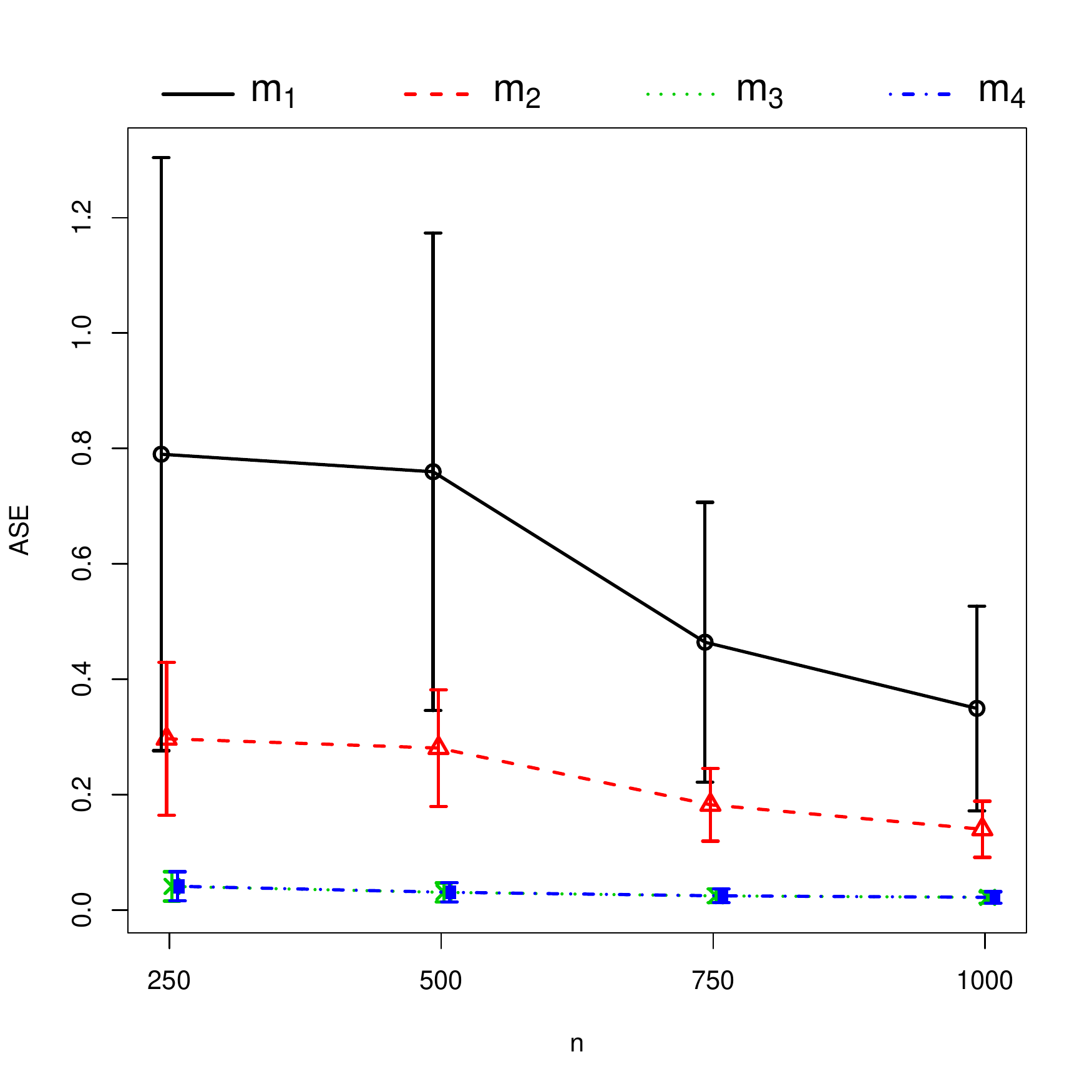}
\includegraphics[angle=0,width=0.24\linewidth]{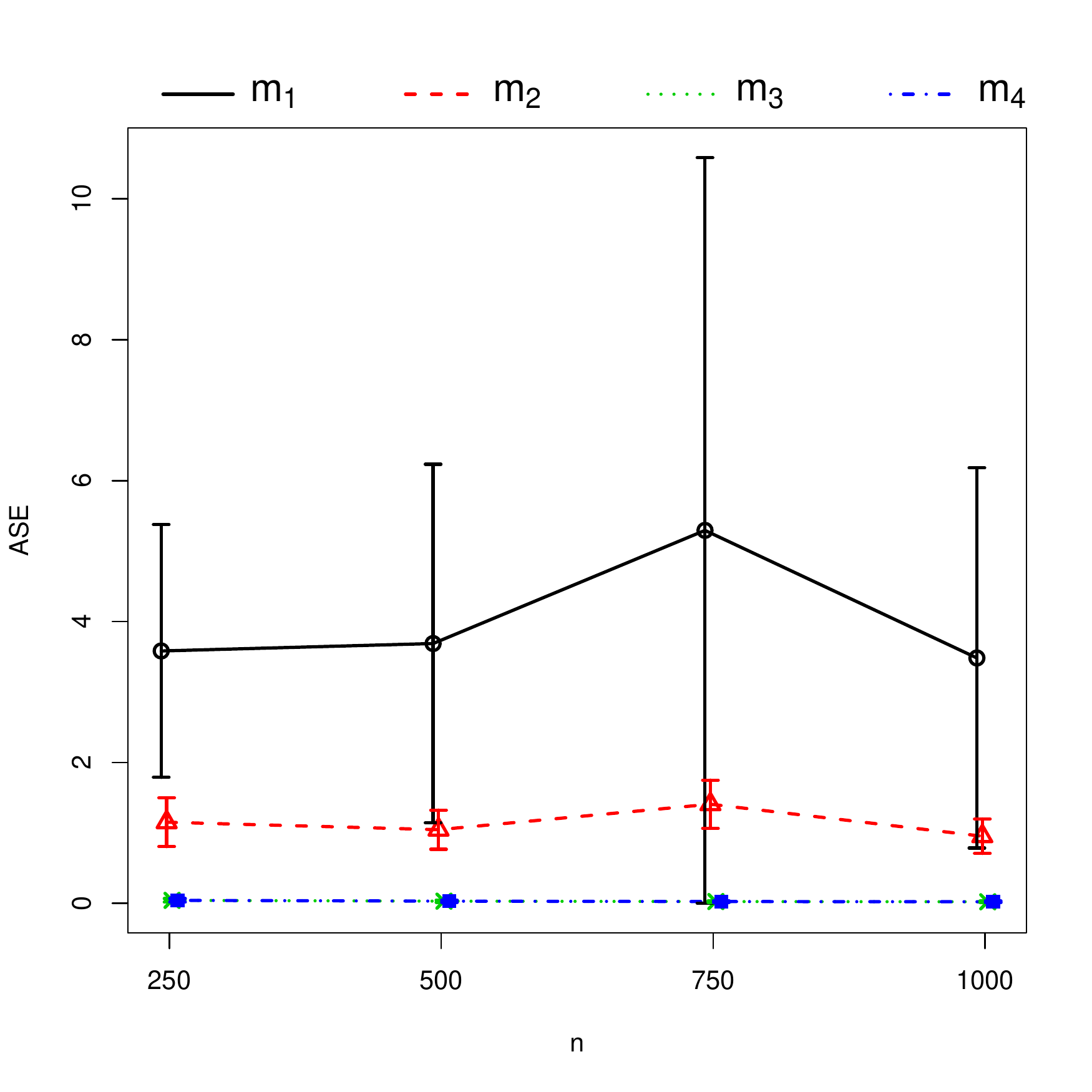} \\
\includegraphics[angle=0,width=0.24\linewidth]{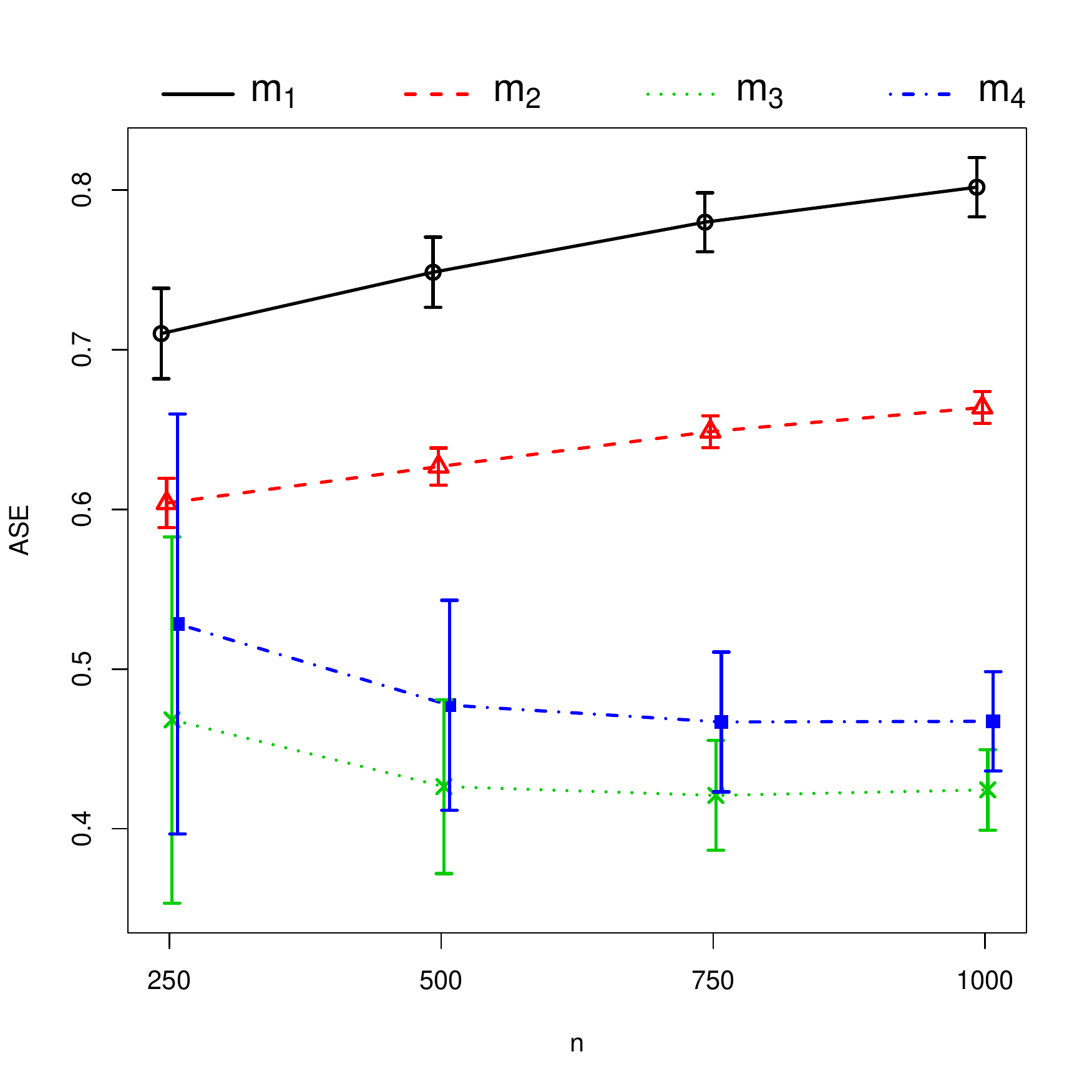}
\includegraphics[angle=0,width=0.24\linewidth]{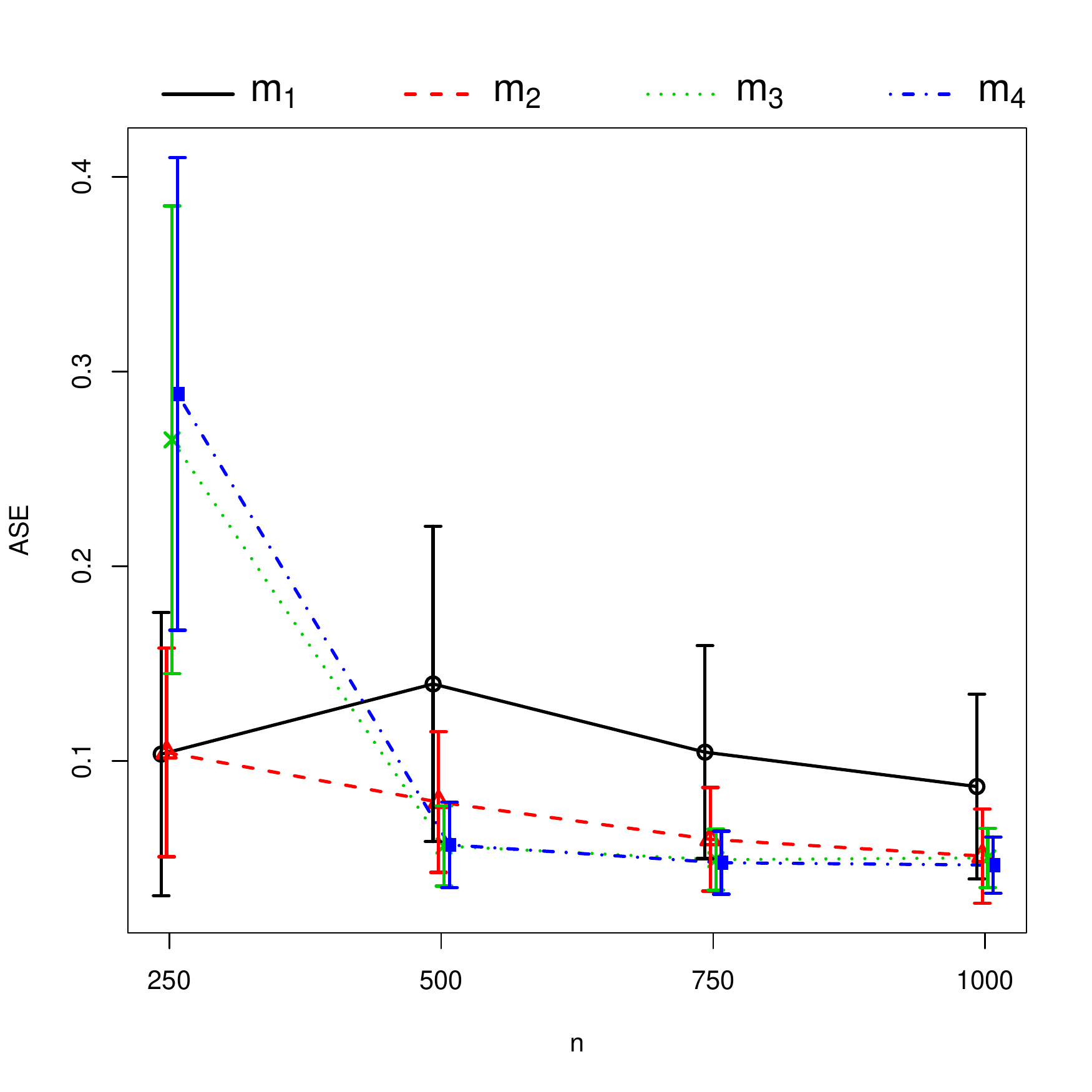}
\includegraphics[angle=0,width=0.24\linewidth]{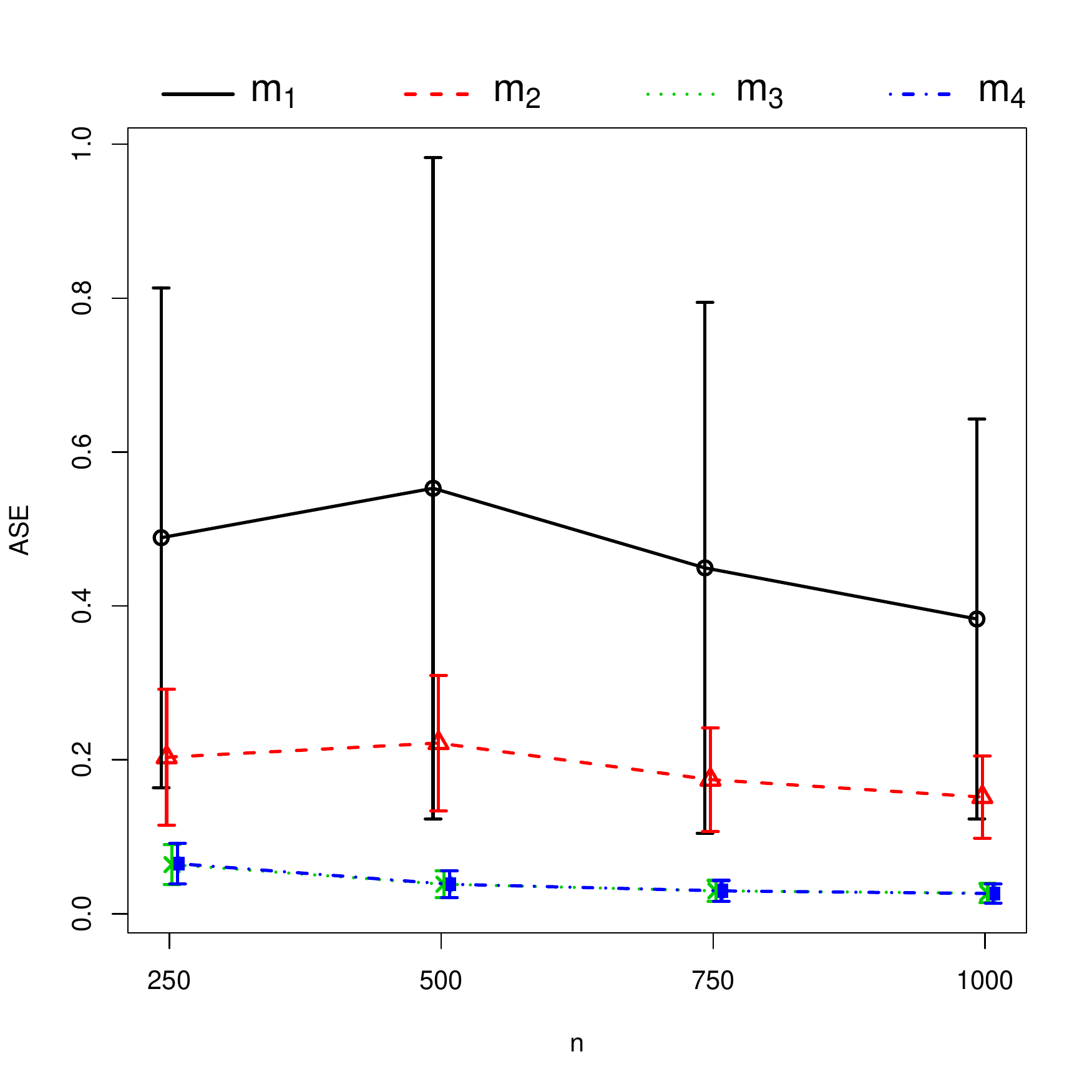}
\includegraphics[angle=0,width=0.24\linewidth]{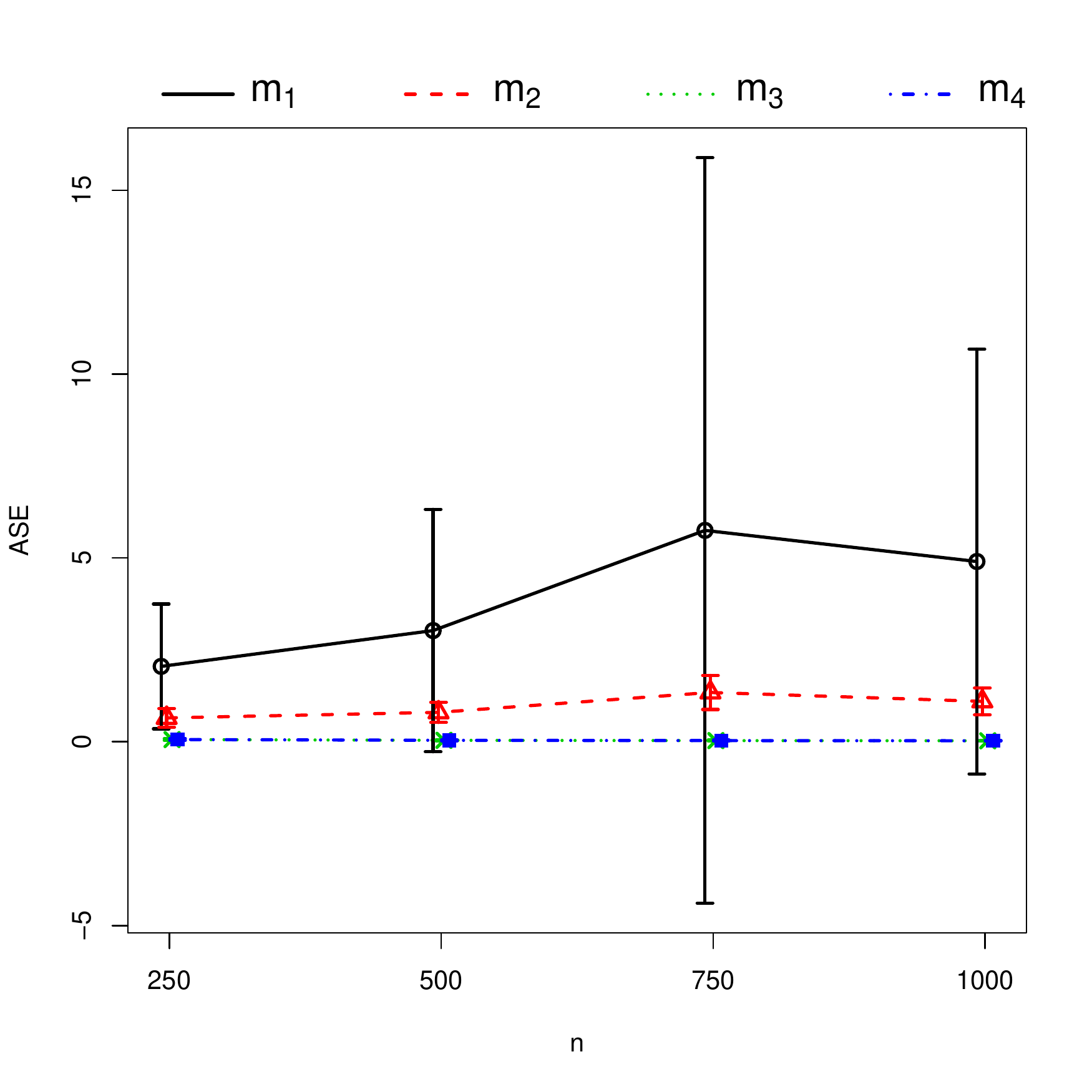} \\
\includegraphics[angle=0,width=0.24\linewidth]{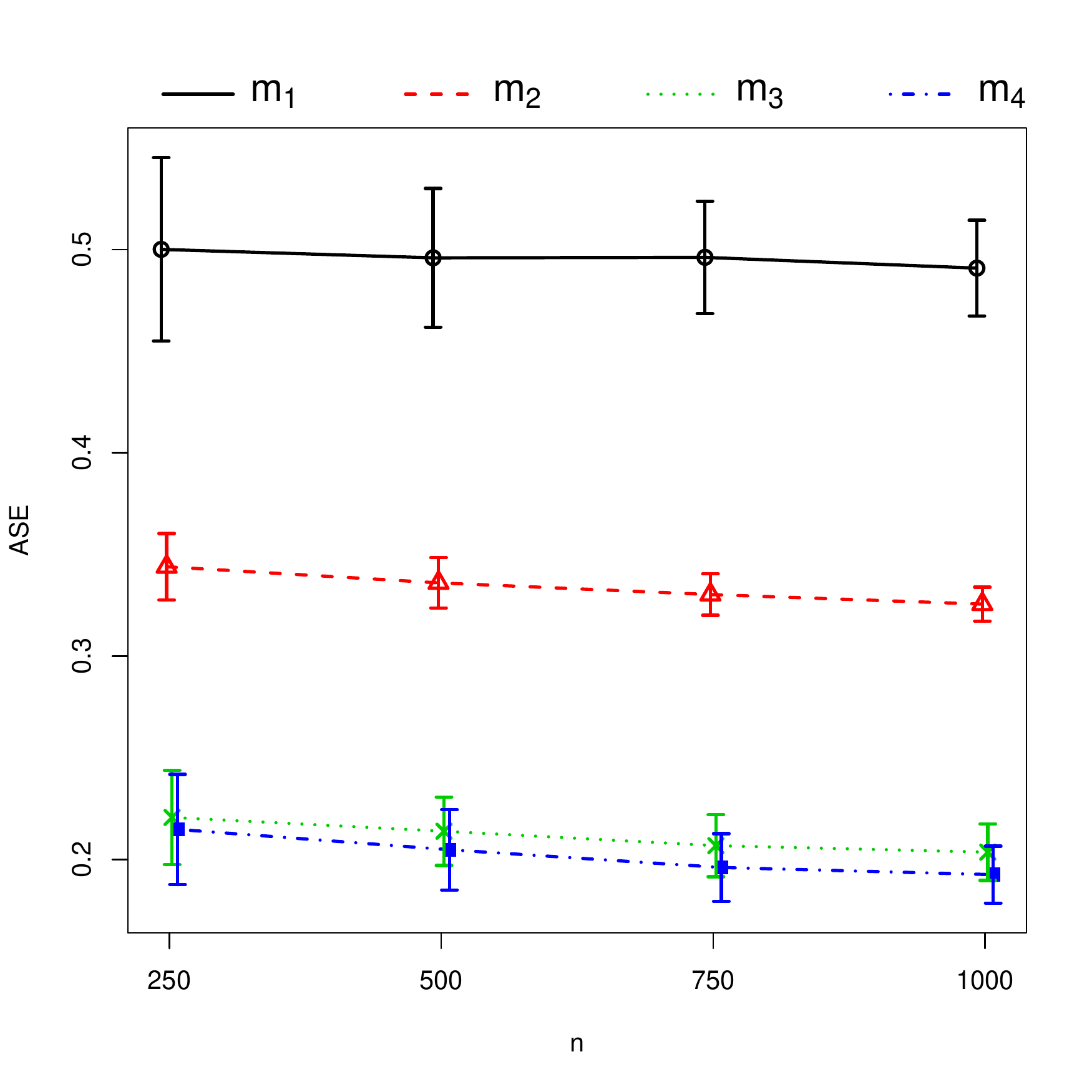}
\includegraphics[angle=0,width=0.24\linewidth]{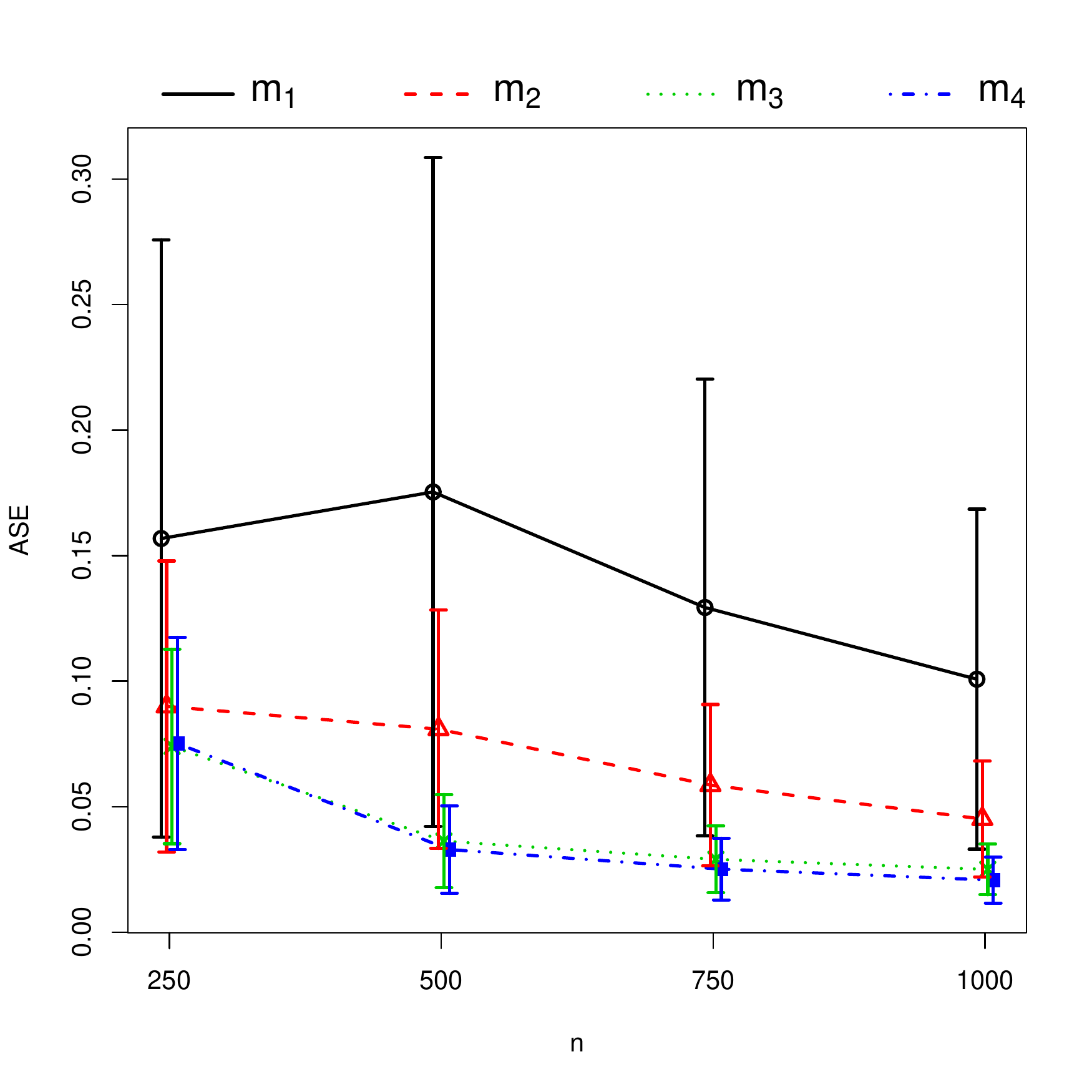}
\includegraphics[angle=0,width=0.24\linewidth]{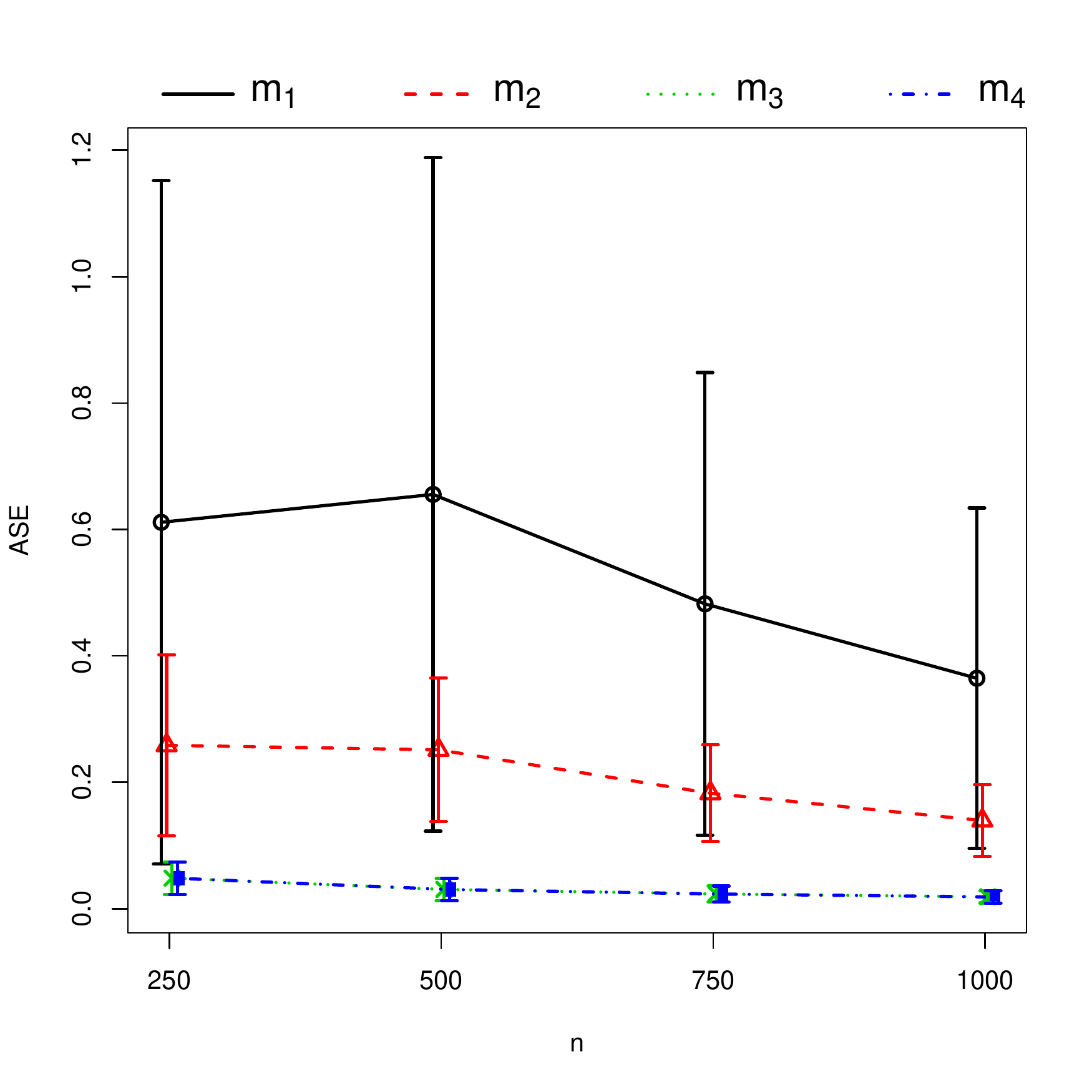}
\includegraphics[angle=0,width=0.24\linewidth]{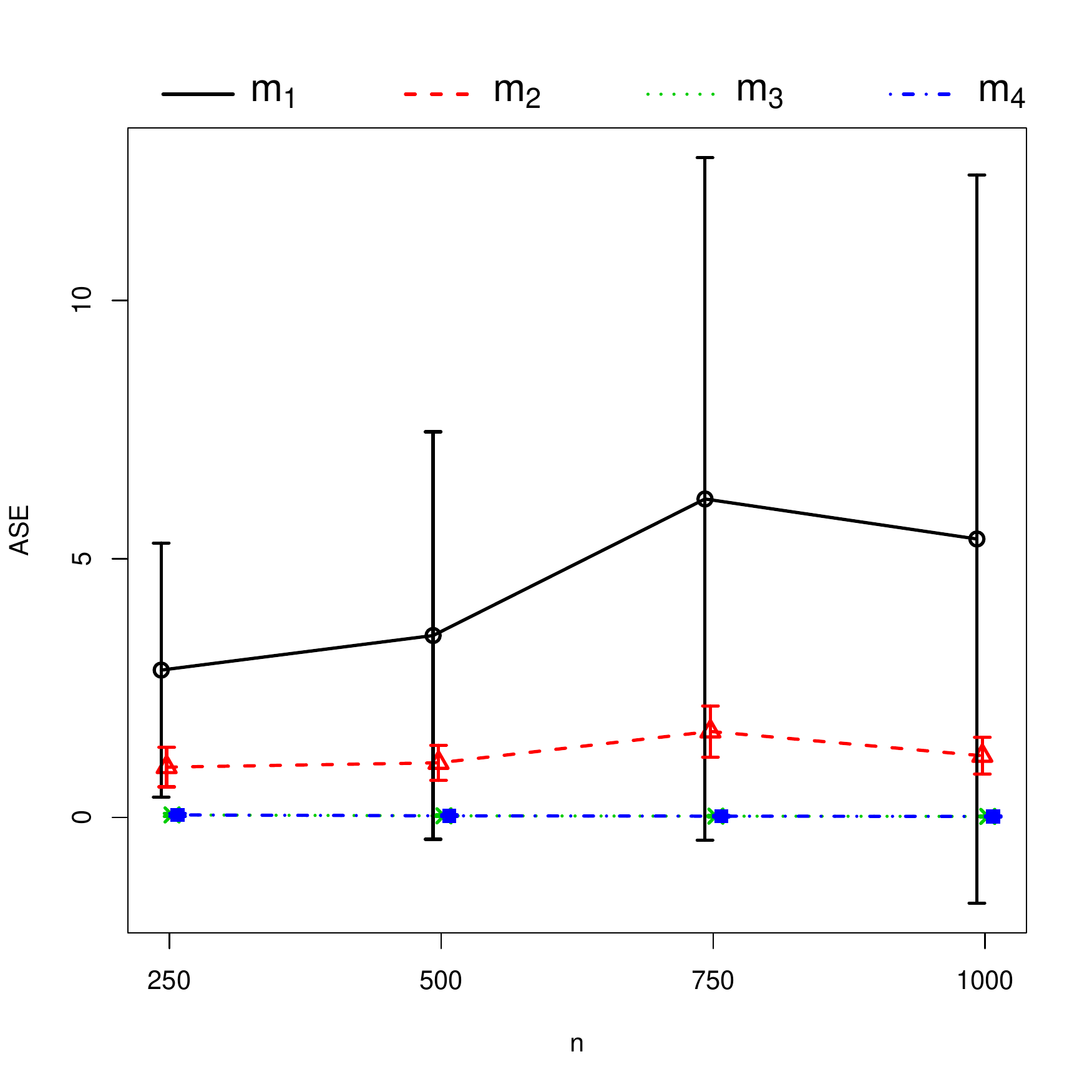}
\caption{Average plus/minus a standard deviation of the ASEs. Full lines, dashed lines, dotted lines and dot-dashed lines represent the ASEs obtained by the methods $ m_1 $--$ m_4 $, respectively. Columns 1--4 are related to the cases where we consider the finest resolution level $ J_1 = \ceil{p\log_2 n} $, $ p = 0.20, 0.45, 0.70, 0.95 $, respectively. The $ i $-th row corresponds to the example $ i $, $ i = 1, 2, 3 $.}
\label{fig:ncomp}
\end{figure}

Finally, another advantage of using warped wavelets can be observed in Figures \ref{fig:estimates-ex1}--\ref{fig:estimates-ex3}, which present pointwise estimates (averages) and 95\% confidence intervals (highest density intervals) based on the 1,000 replications. It becomes clear that estimates based on warped wavelets are more precise for regions in the density's support where the weighting function is close to zero.
%Furthermore, poor estimates are obtained for all methods $ m_k $, $ k = 1, 2, 3, 4 $.
When we employ $ J_1 = \ceil{0.95 \log_2 n} $, the traditional method $ m_2 $ provides negative estimates.

\begin{figure}[!ht]
\centering
\includegraphics[angle=0,width=0.24\linewidth]{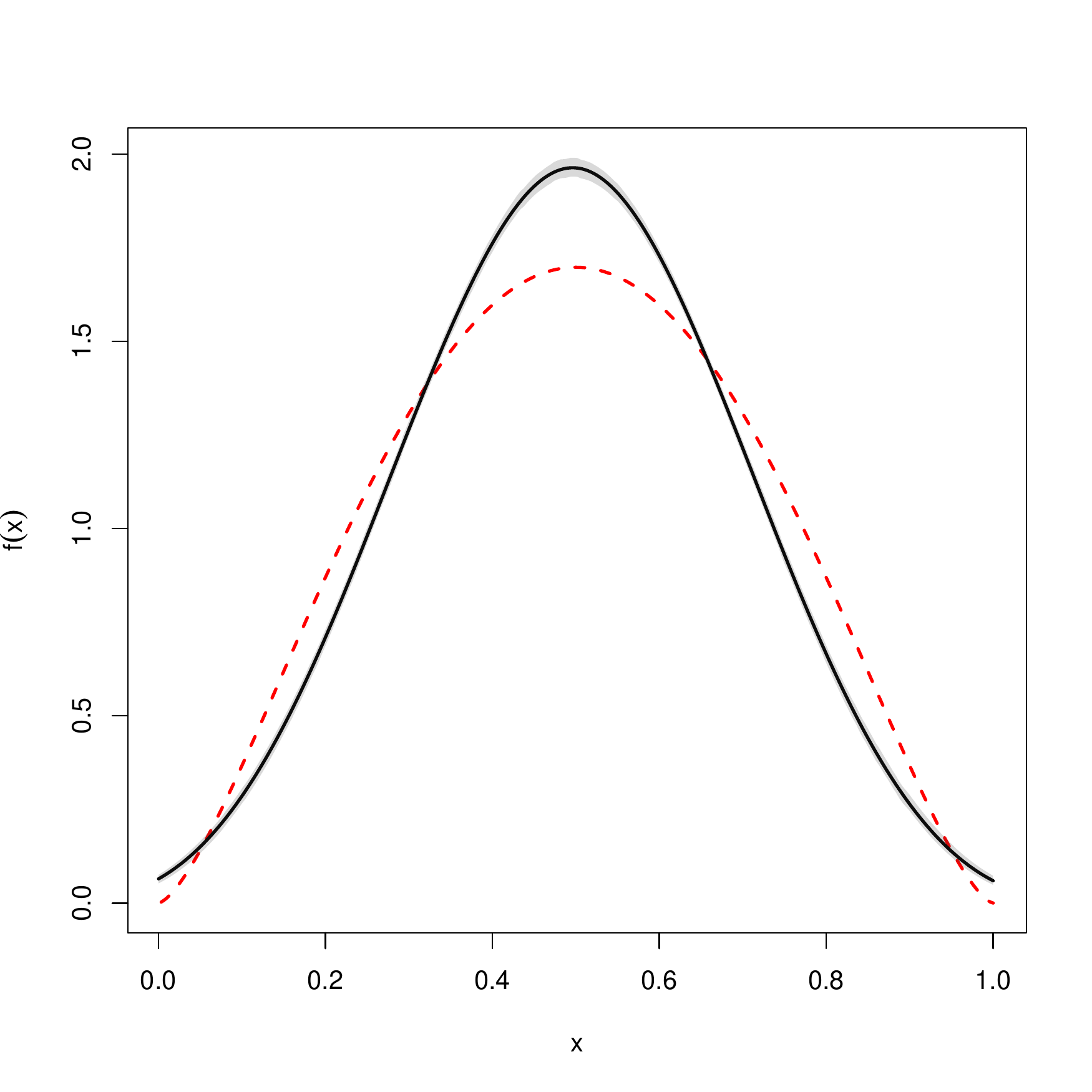}
\includegraphics[angle=0,width=0.24\linewidth]{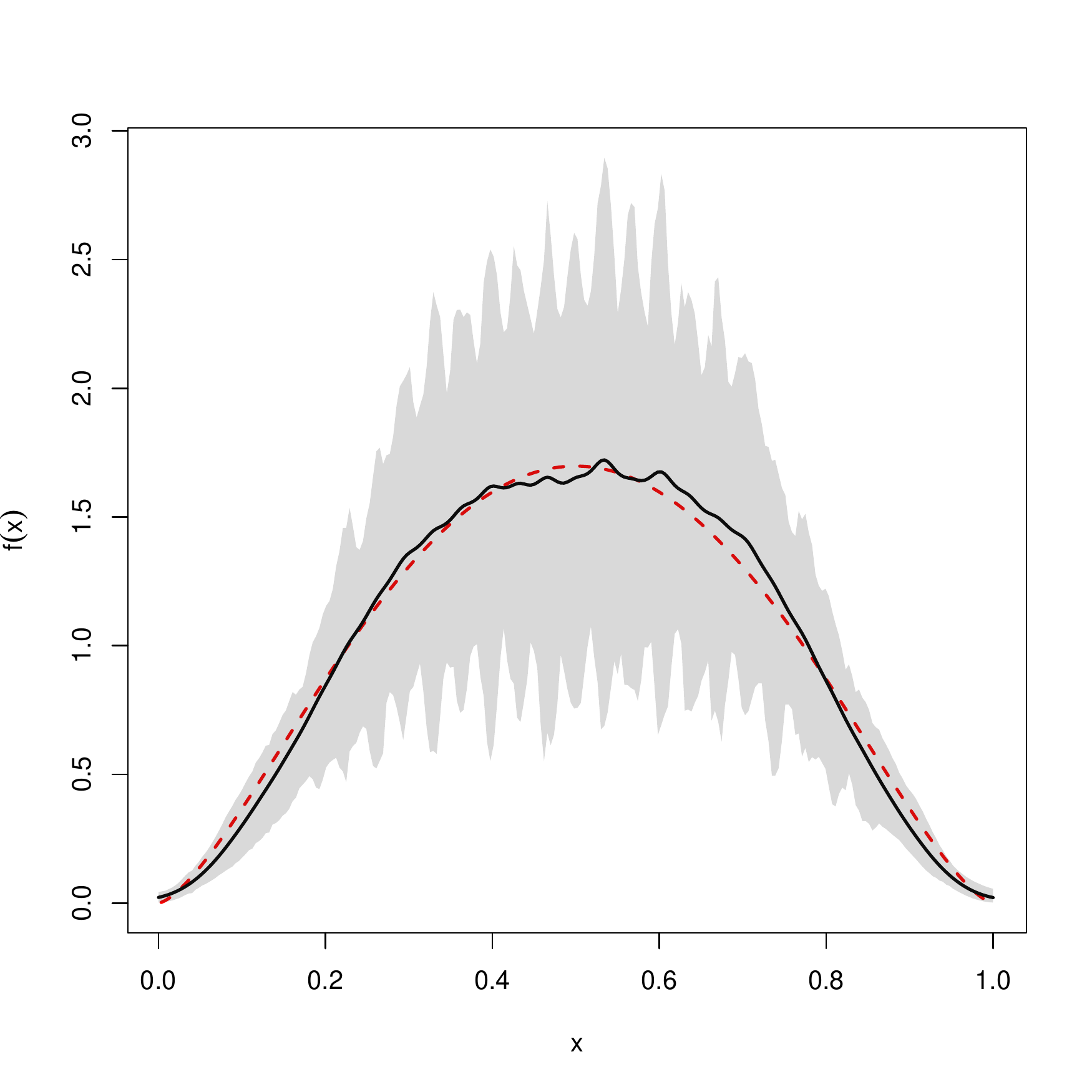}
\includegraphics[angle=0,width=0.24\linewidth]{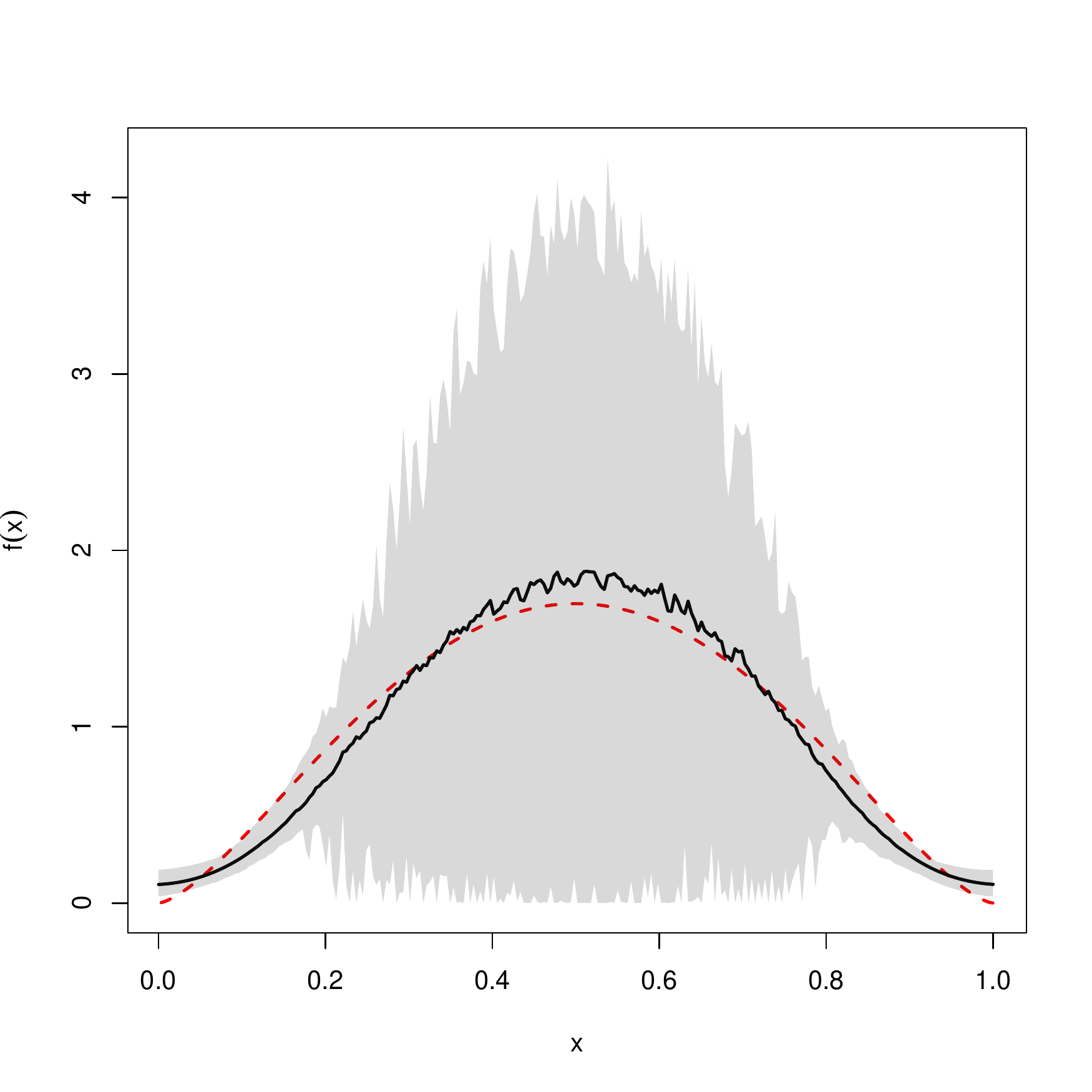}
\includegraphics[angle=0,width=0.24\linewidth]{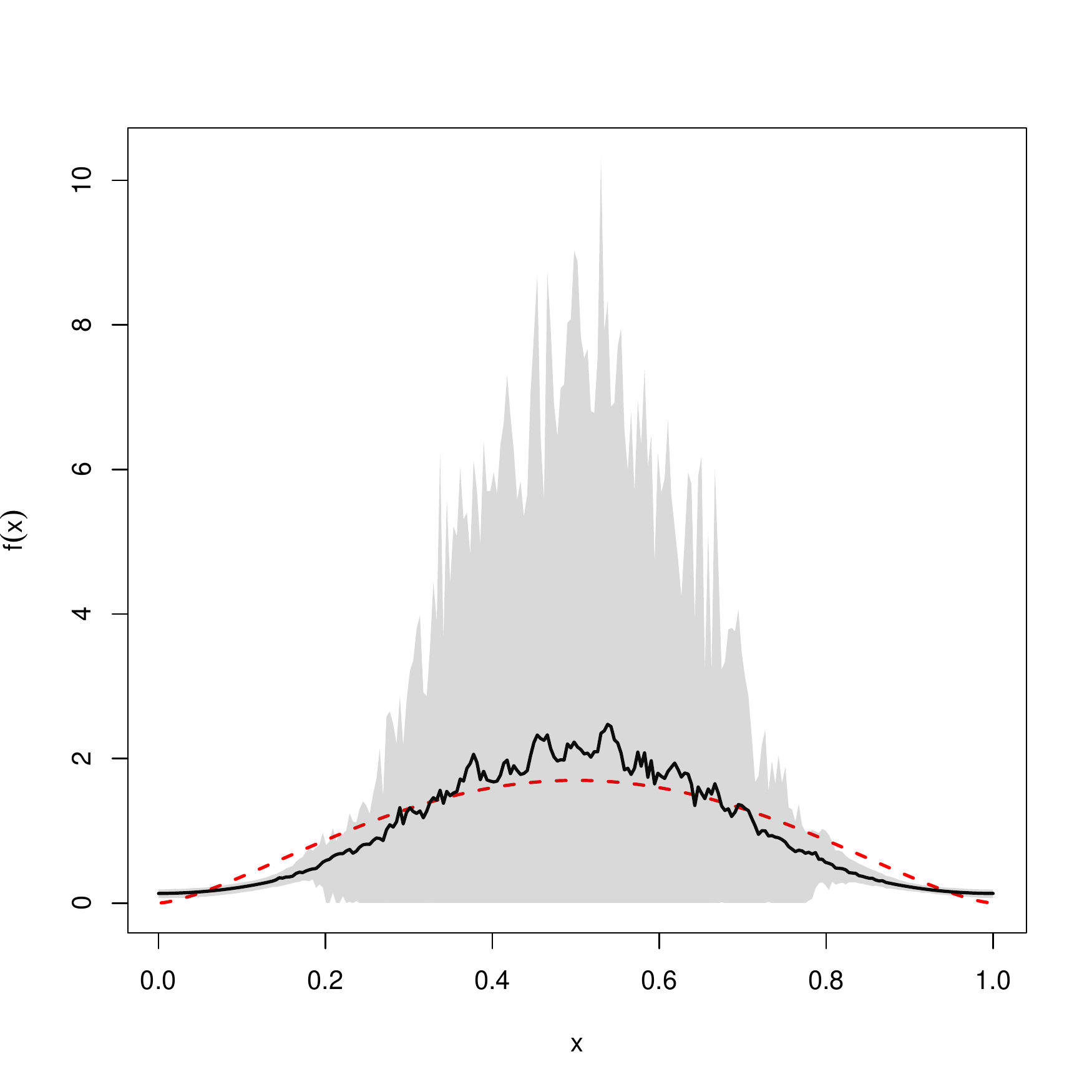} \\
\includegraphics[angle=0,width=0.24\linewidth]{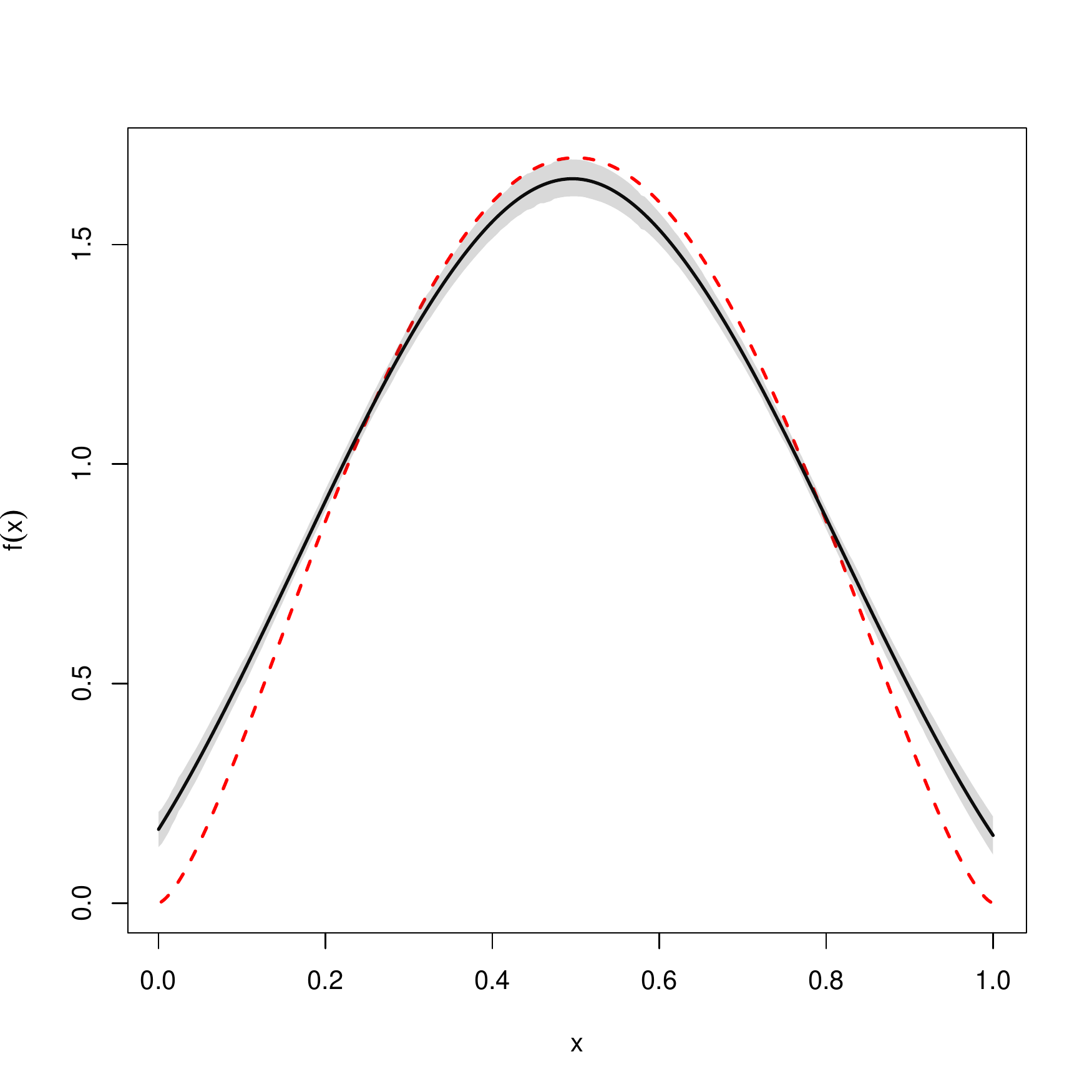}
\includegraphics[angle=0,width=0.24\linewidth]{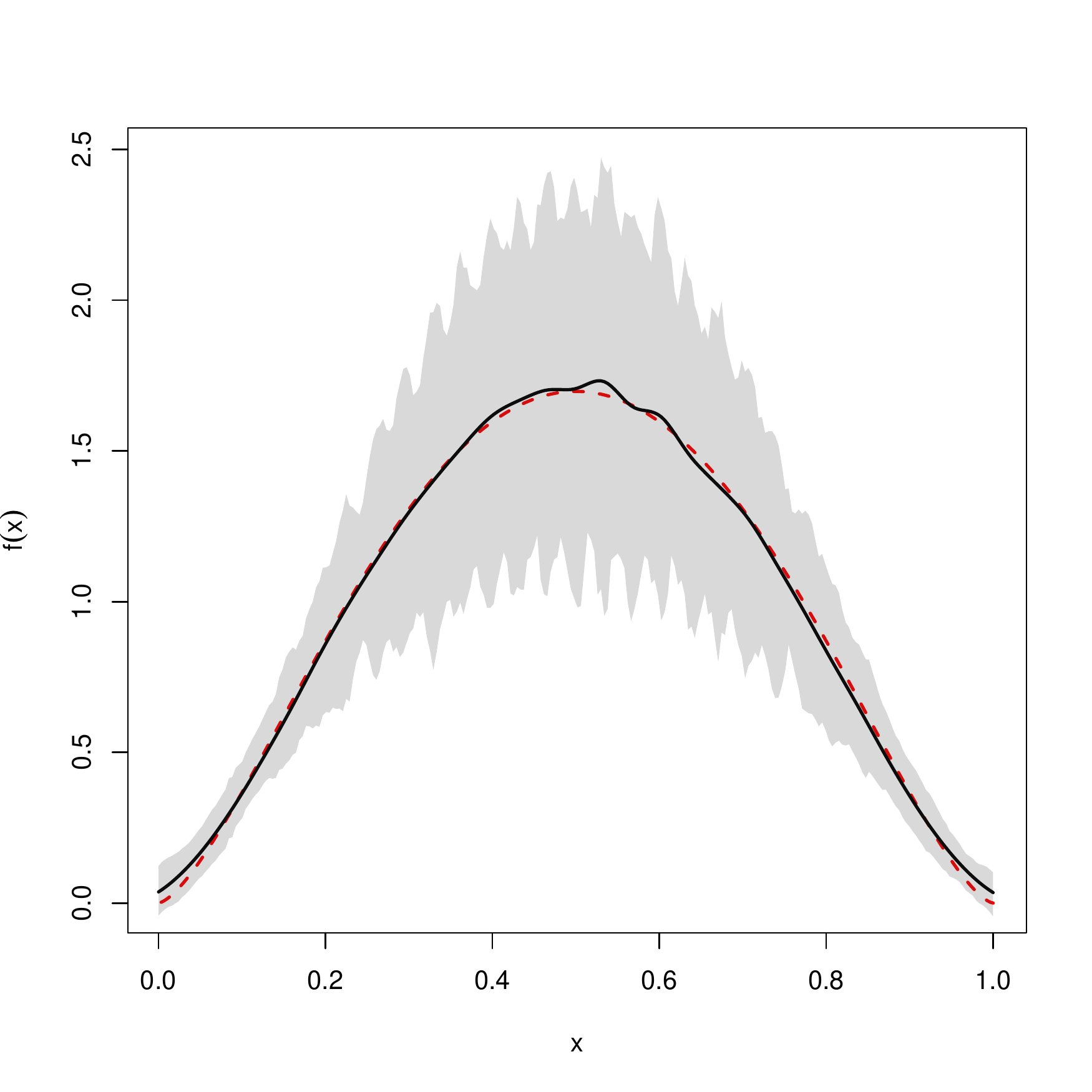}
\includegraphics[angle=0,width=0.24\linewidth]{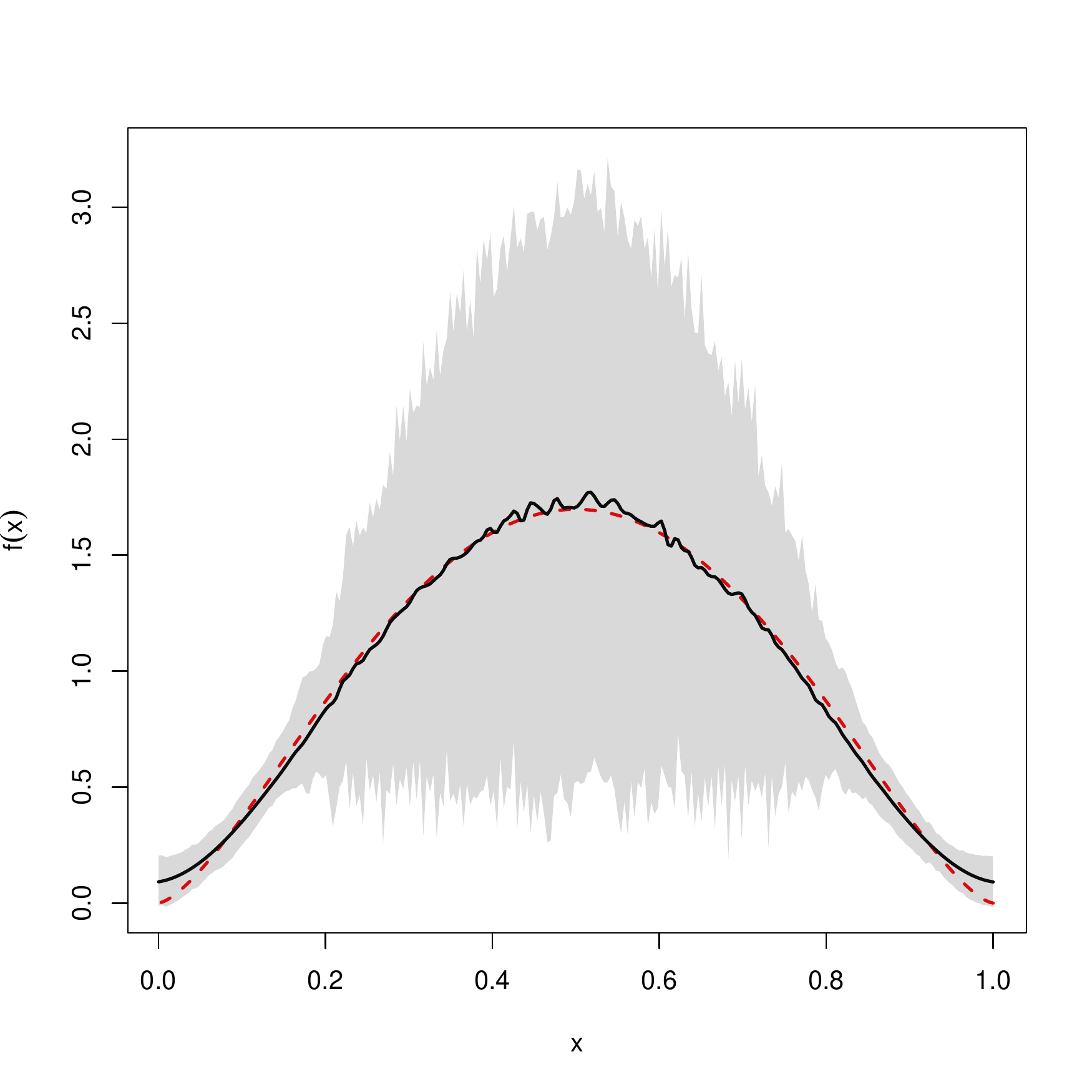}
\includegraphics[angle=0,width=0.24\linewidth]{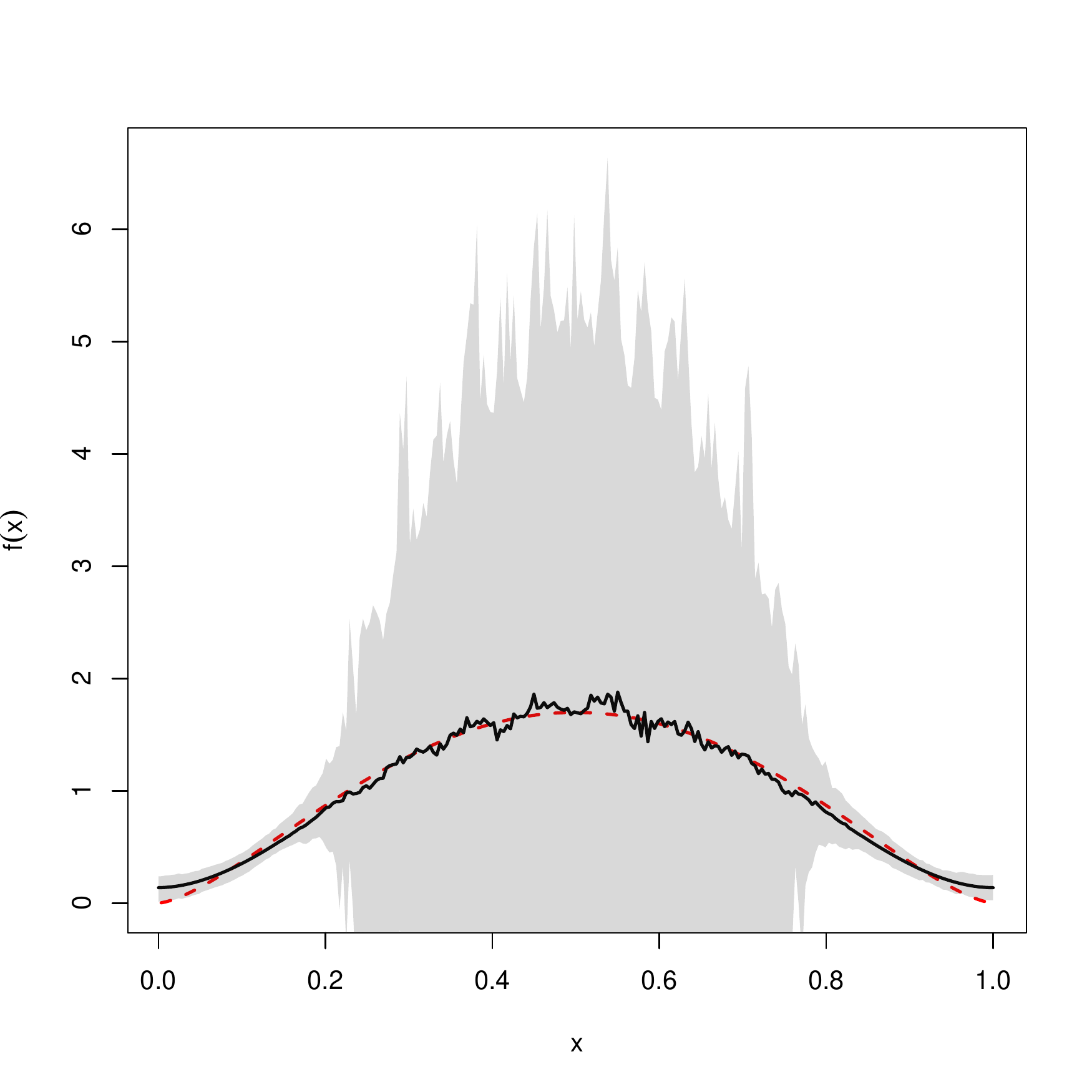} \\
\includegraphics[angle=0,width=0.24\linewidth]{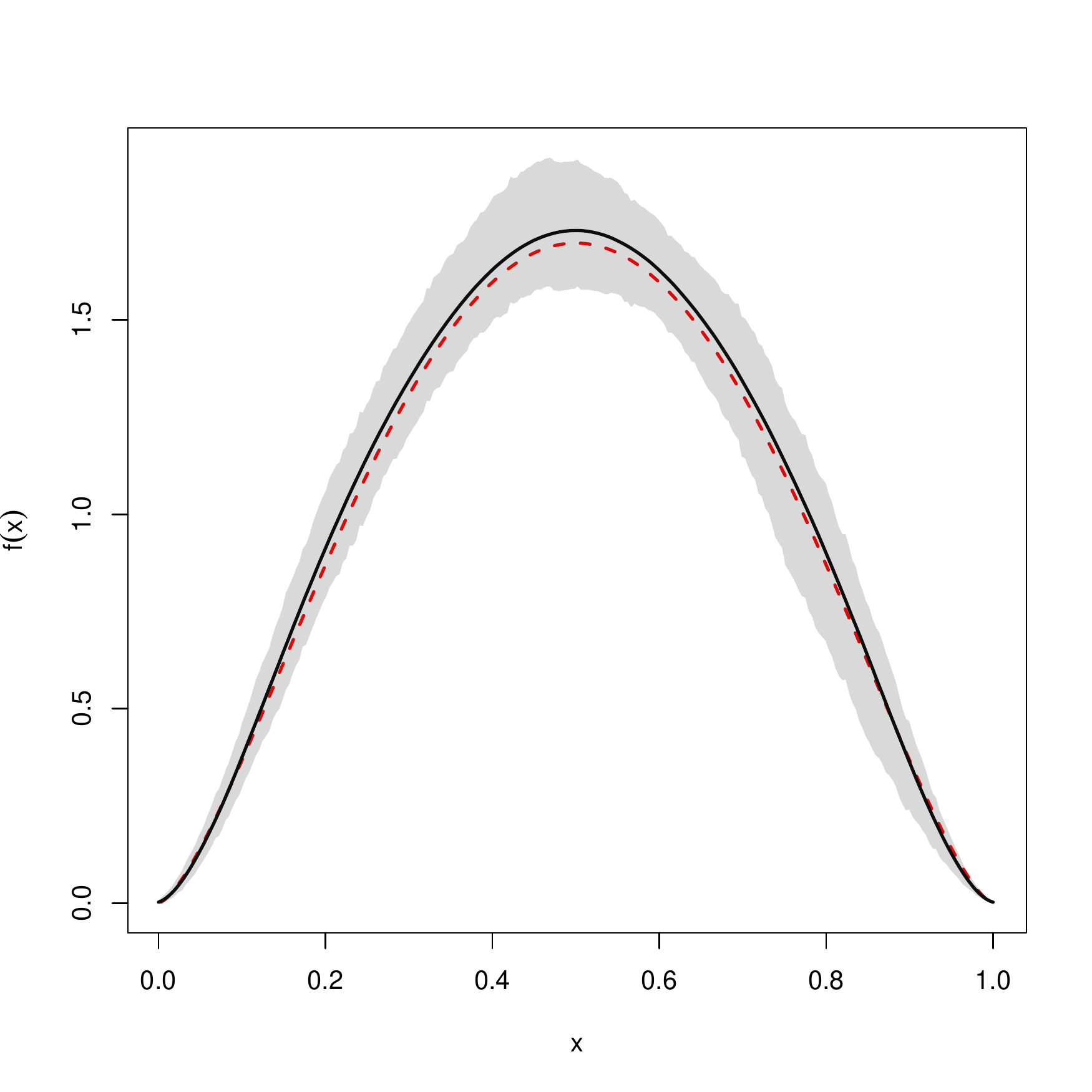}
\includegraphics[angle=0,width=0.24\linewidth]{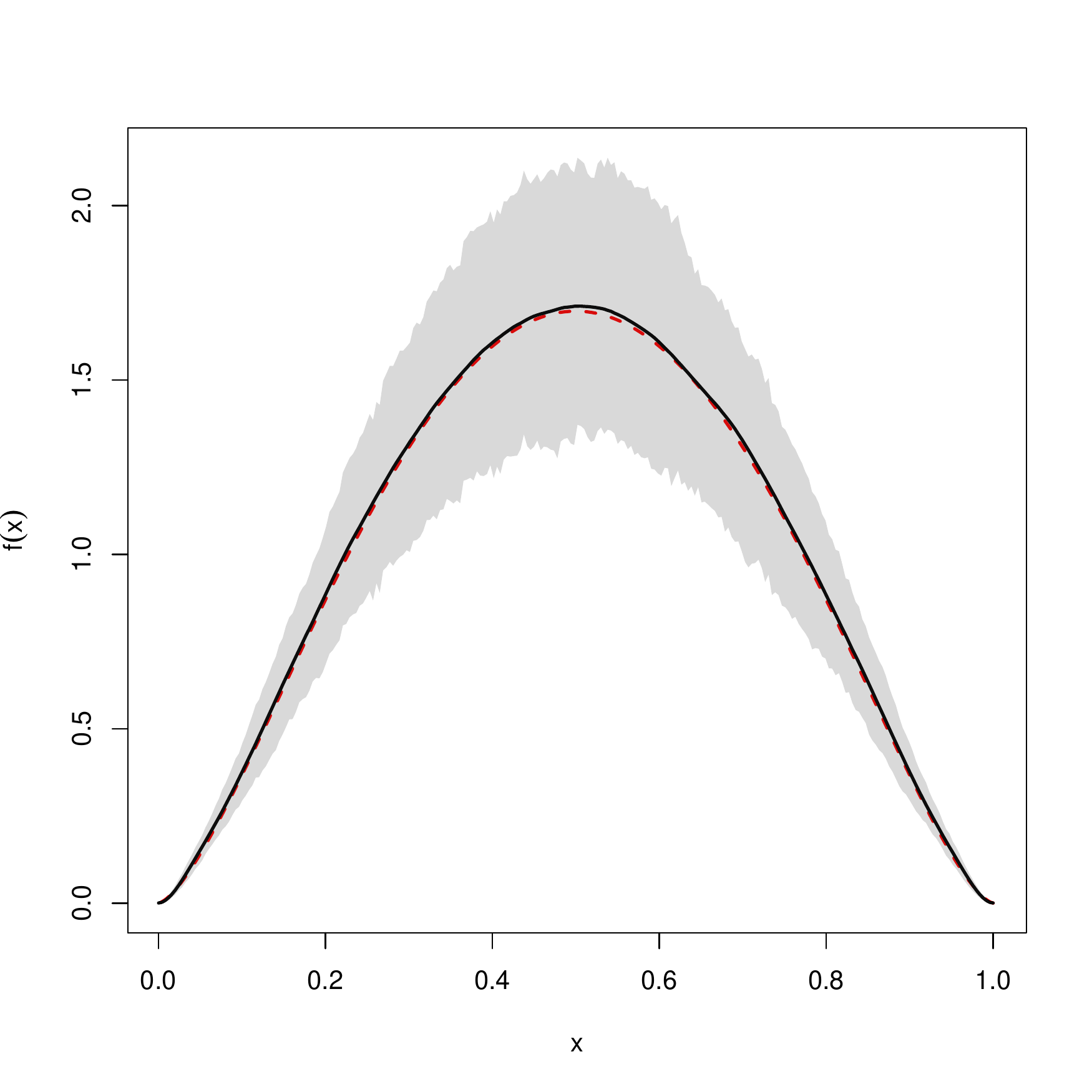}
\includegraphics[angle=0,width=0.24\linewidth]{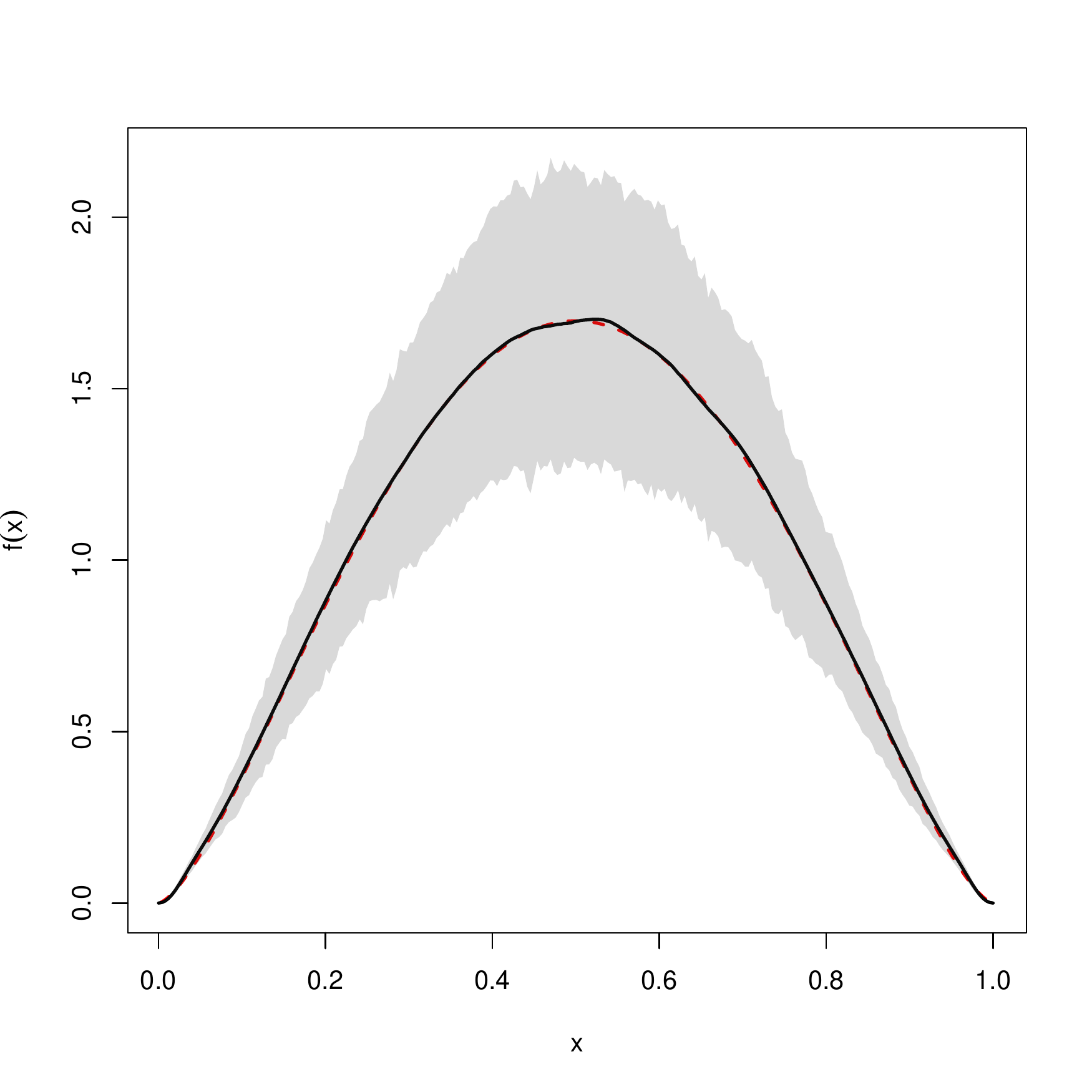}
\includegraphics[angle=0,width=0.24\linewidth]{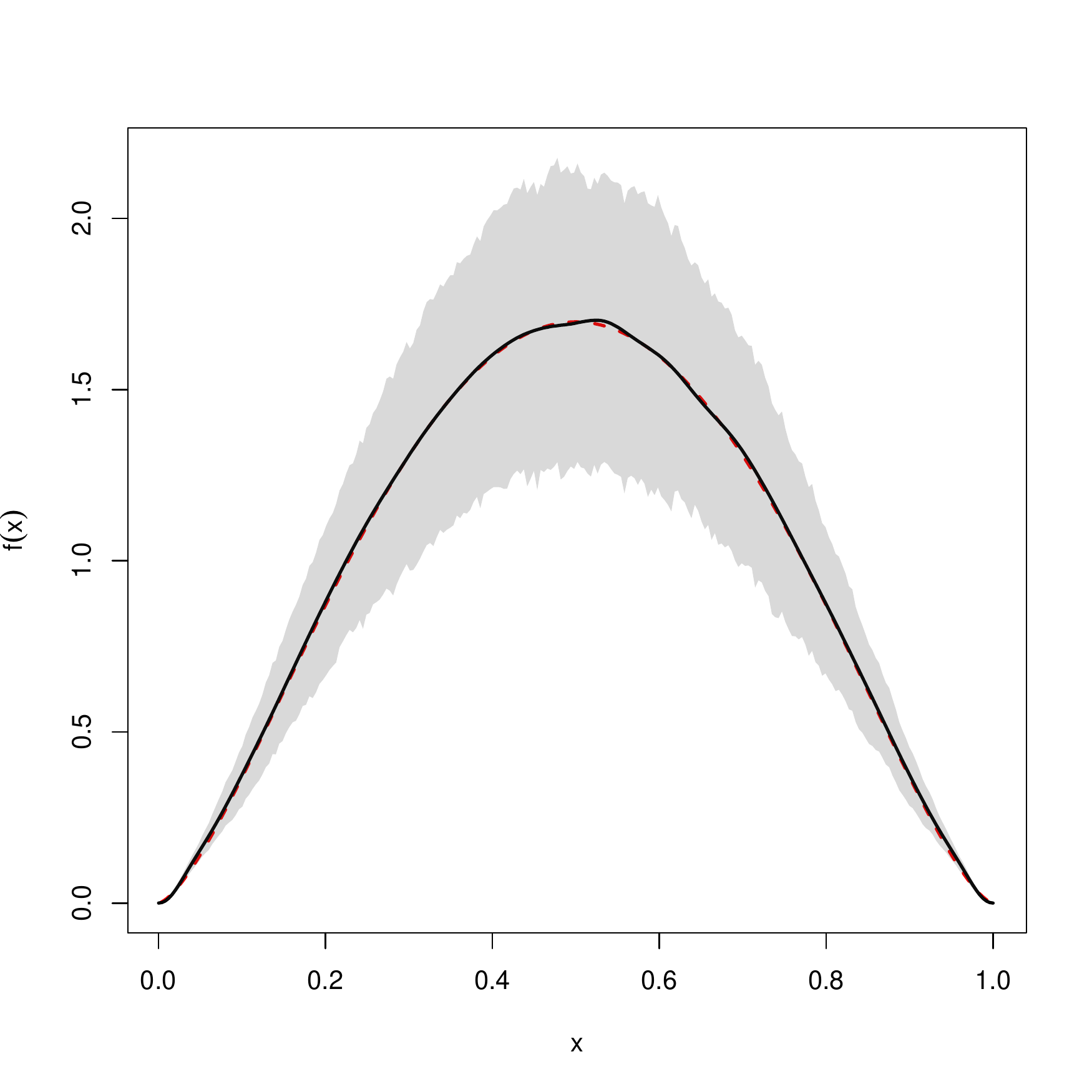} \\
\includegraphics[angle=0,width=0.24\linewidth]{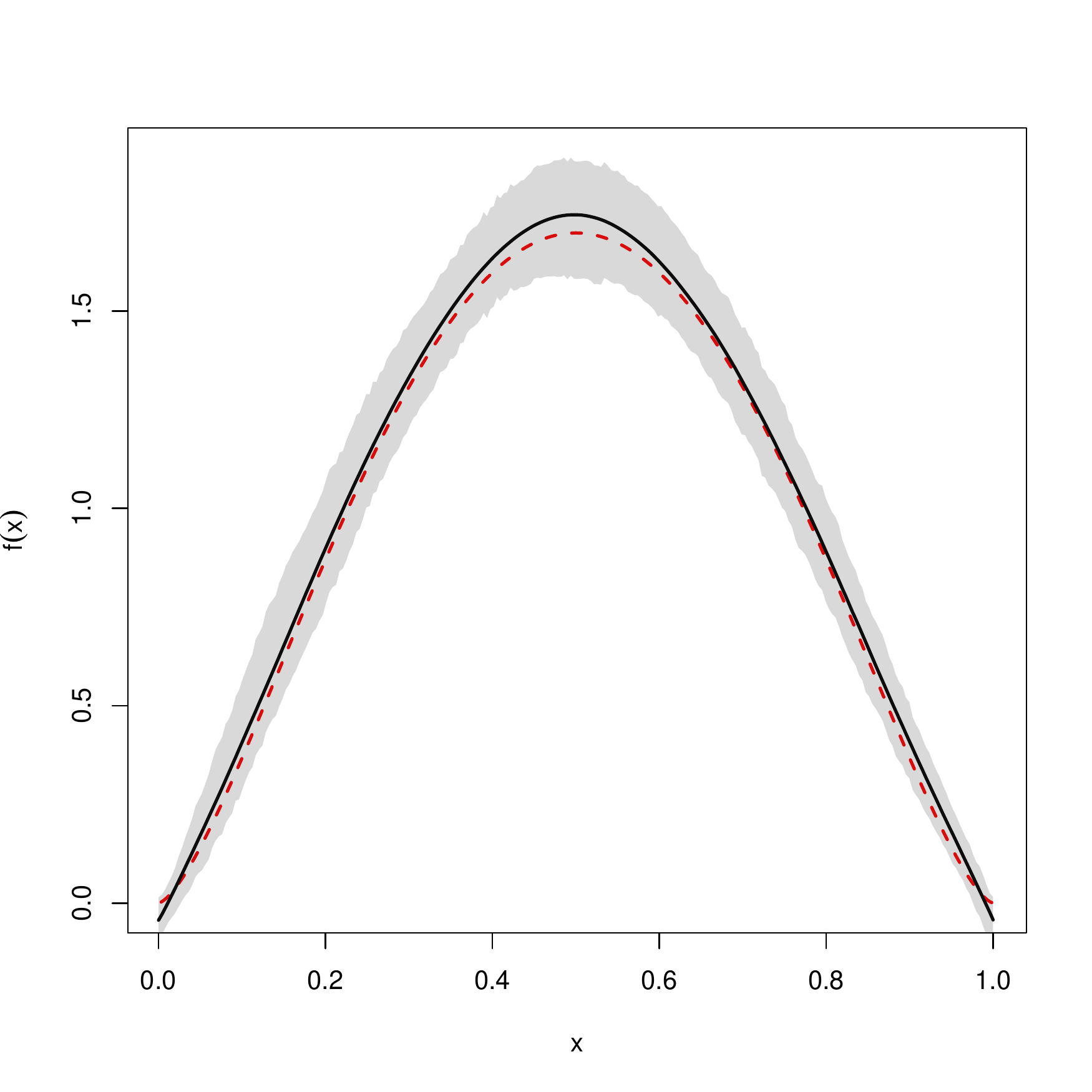}
\includegraphics[angle=0,width=0.24\linewidth]{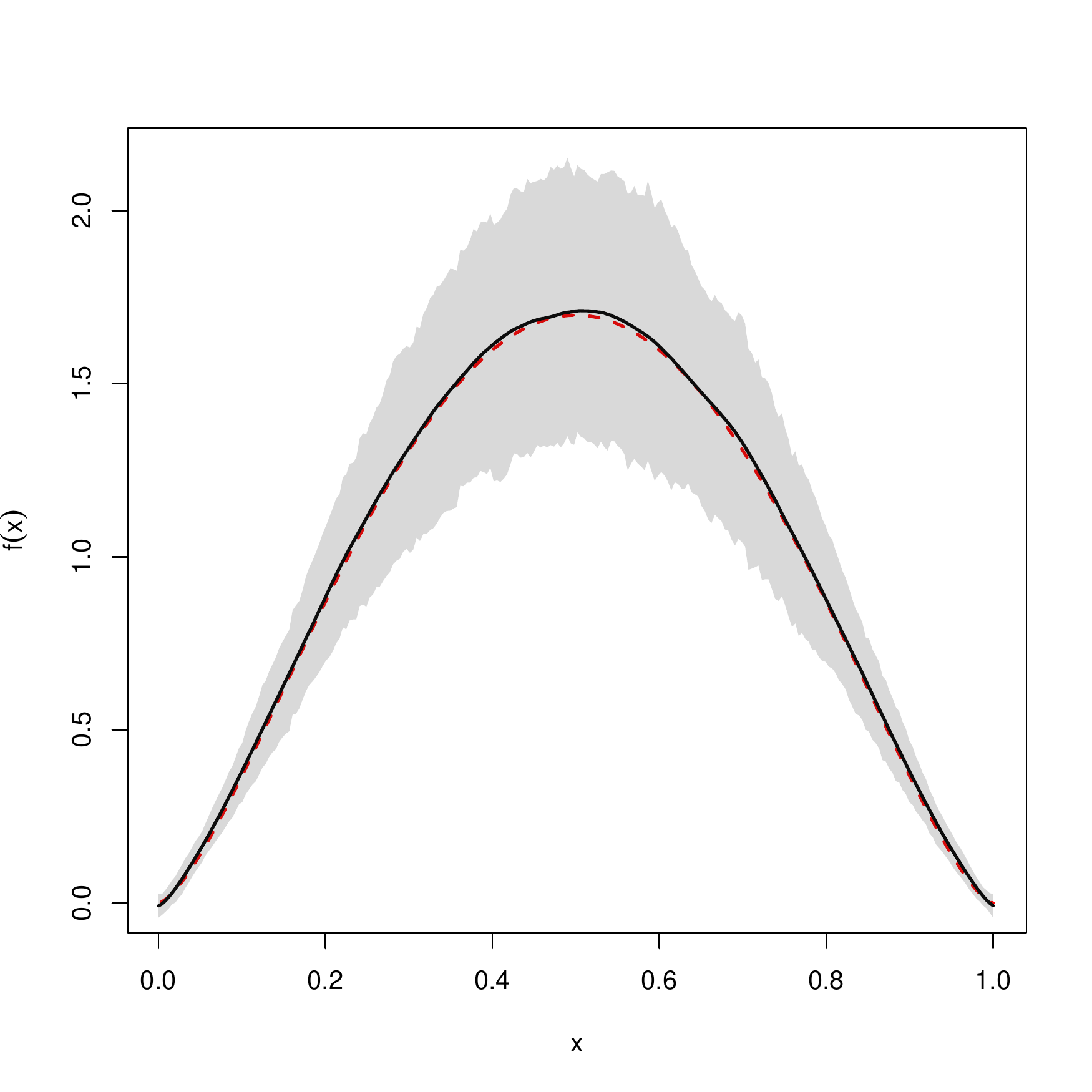}
\includegraphics[angle=0,width=0.24\linewidth]{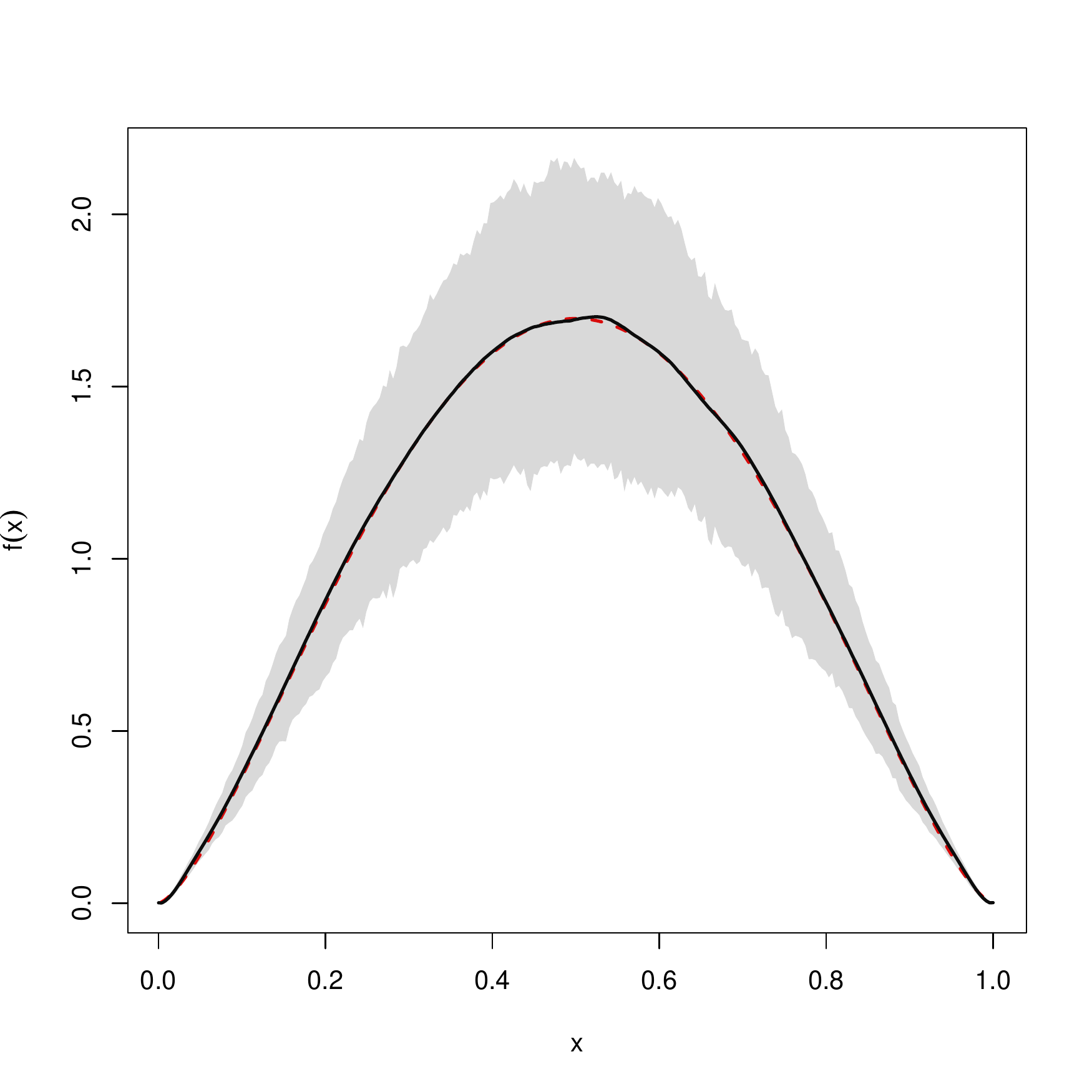}
\includegraphics[angle=0,width=0.24\linewidth]{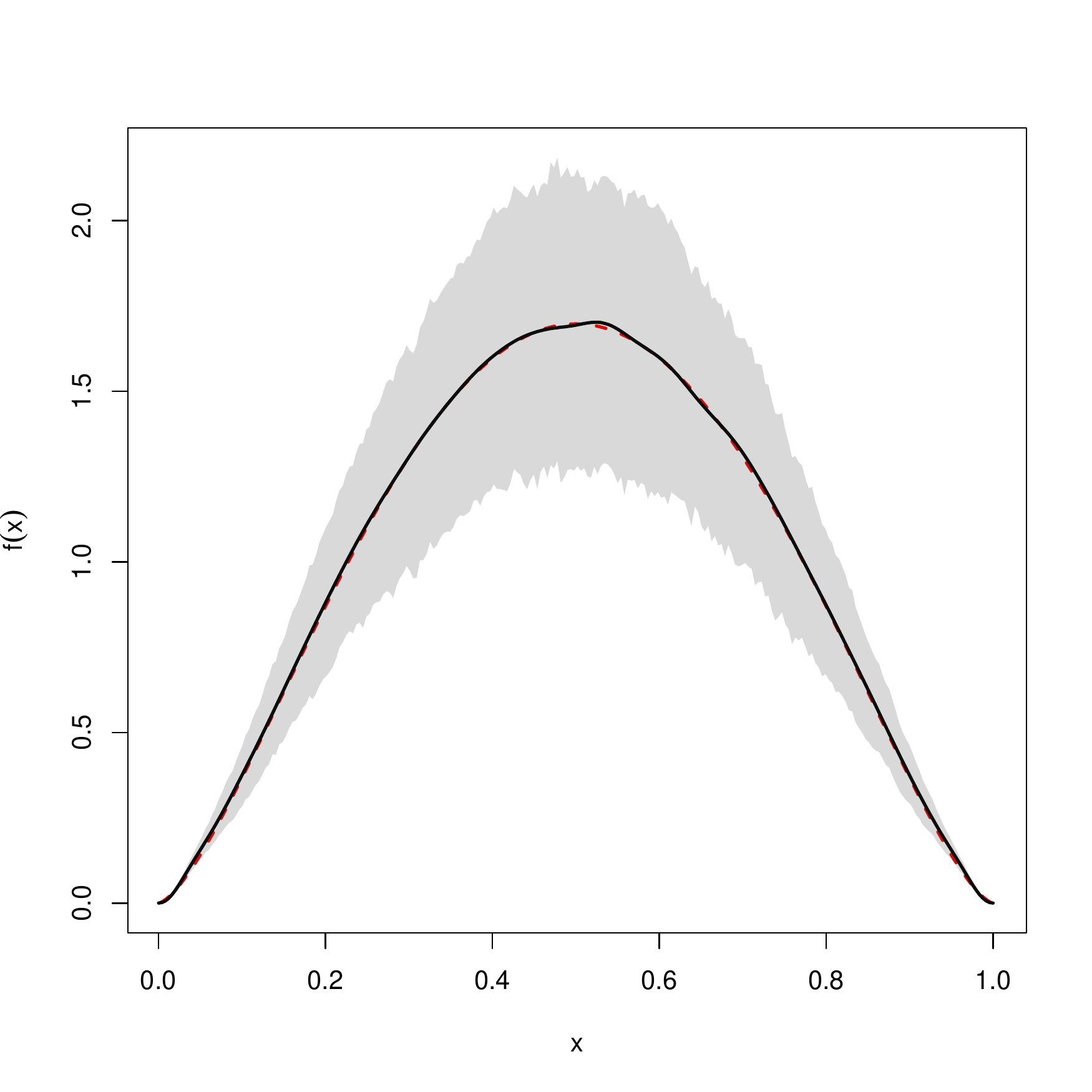} \\
\caption{Pointwise average estimates (full lines), with 95\% confidence intervals (shaded regions), for the density in $ \ex_1 $ (dashed lines), with $ n = 1,000 $. The columns 1--4 are related to the cases where we consider the finest resolution level $ J_1 = \ceil{p\log_2 n} $, $ p = 0.20, 0.45, 0.70, 0.95 $, respectively. The $ k $-th row is related to the estimation method $ m_k $, $ k = 1, 2, 3, 4 $.}
\label{fig:estimates-ex1}
\end{figure}

\begin{figure}[!ht]
\centering
\includegraphics[angle=0,width=0.24\linewidth]{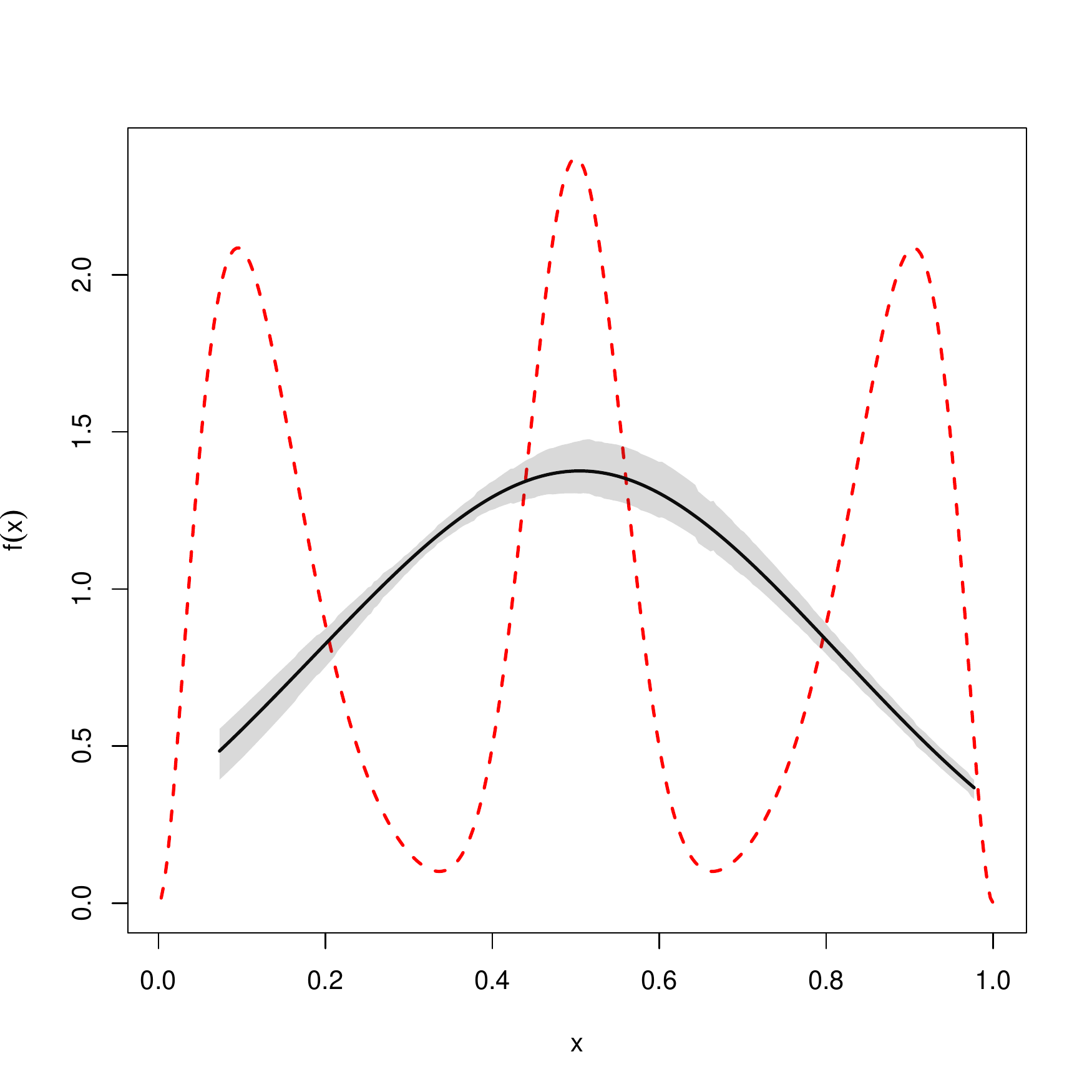}
\includegraphics[angle=0,width=0.24\linewidth]{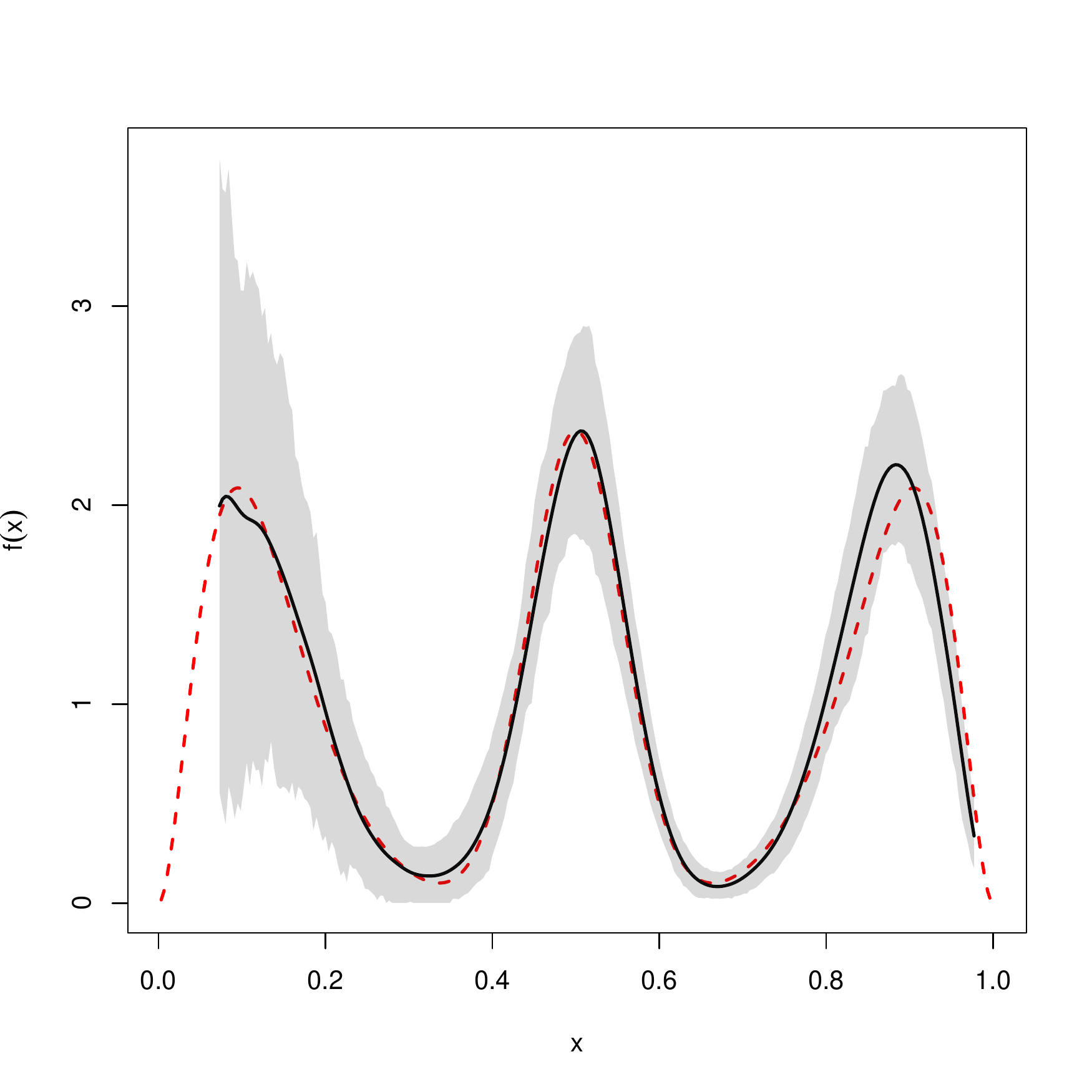}
\includegraphics[angle=0,width=0.24\linewidth]{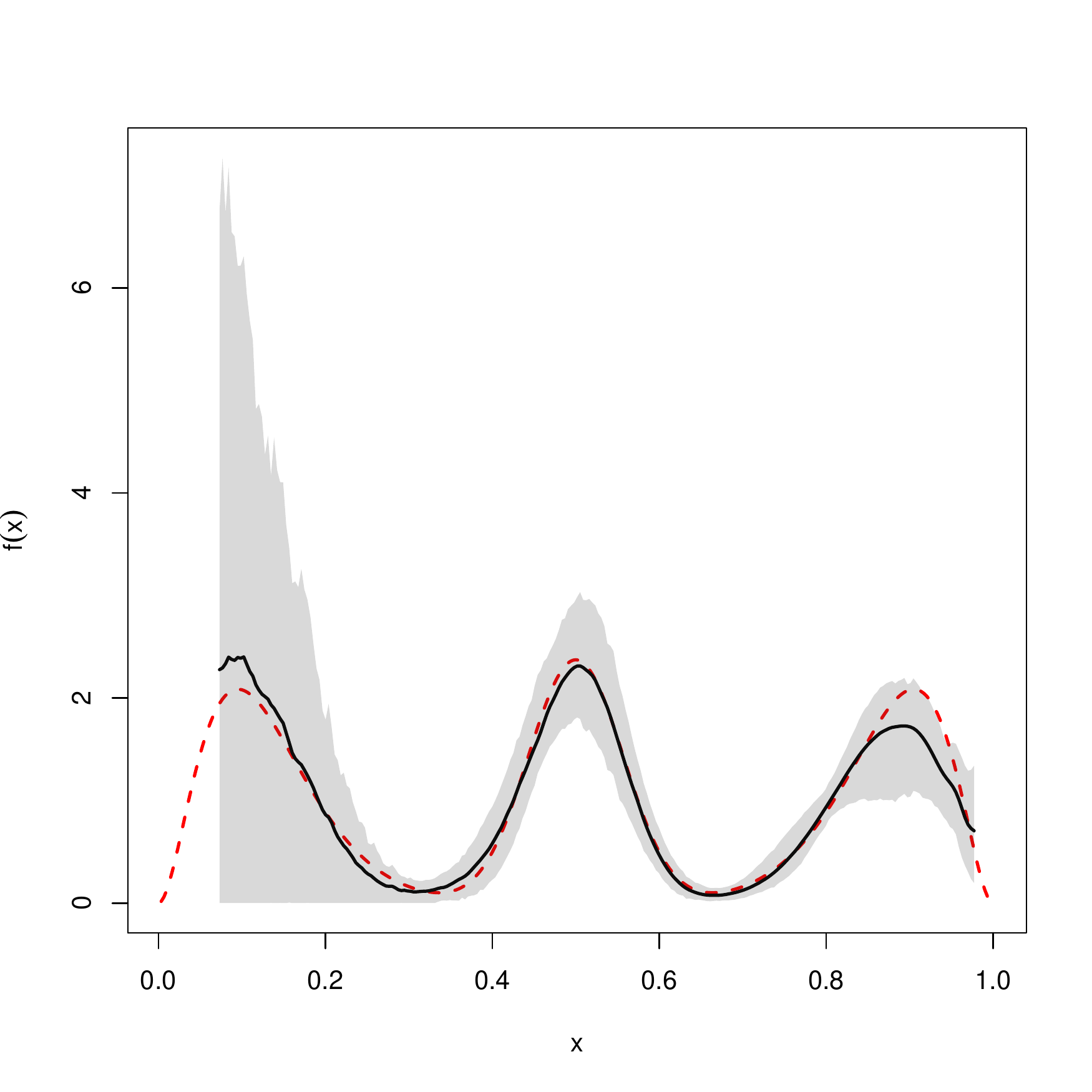}
\includegraphics[angle=0,width=0.24\linewidth]{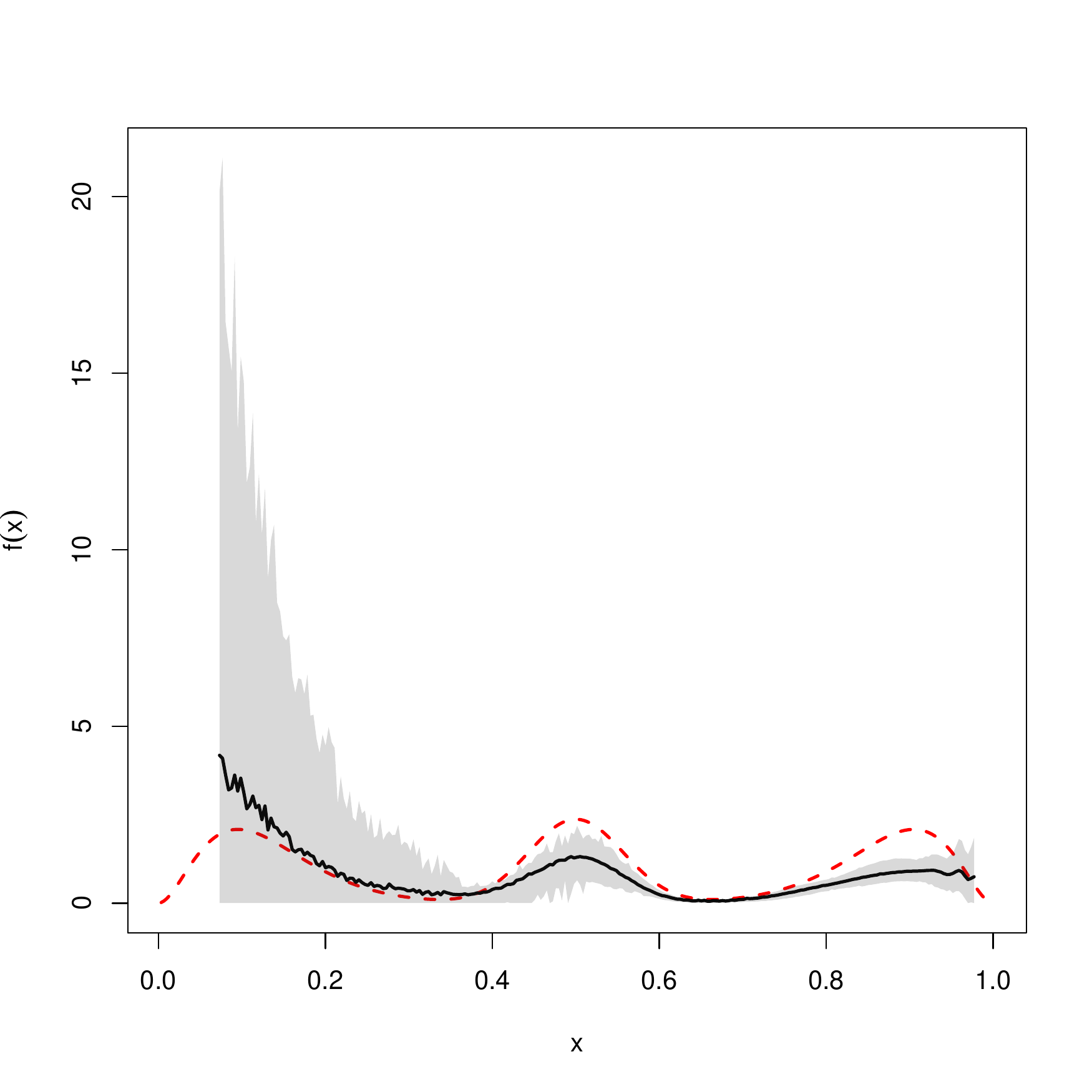} \\
\includegraphics[angle=0,width=0.24\linewidth]{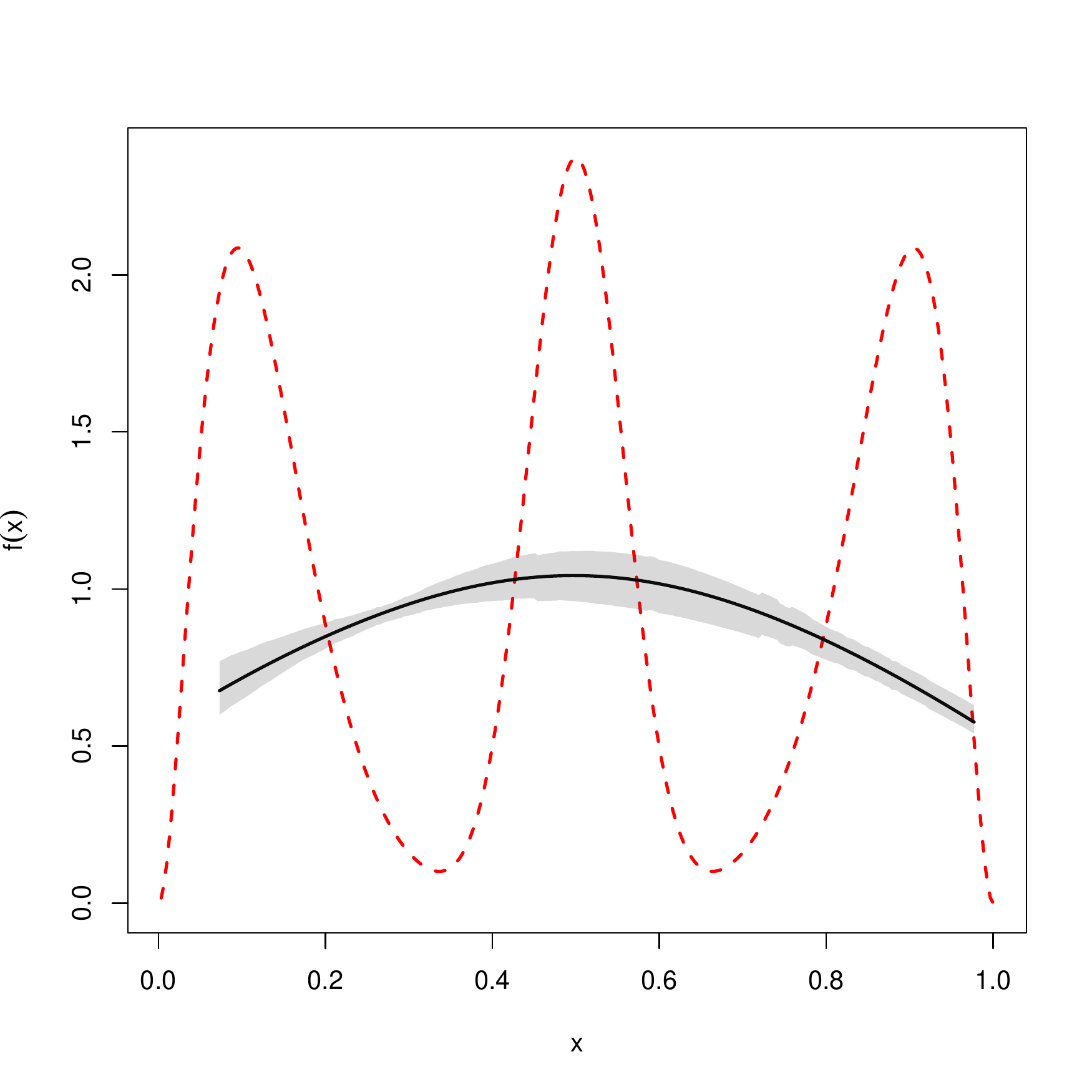}
\includegraphics[angle=0,width=0.24\linewidth]{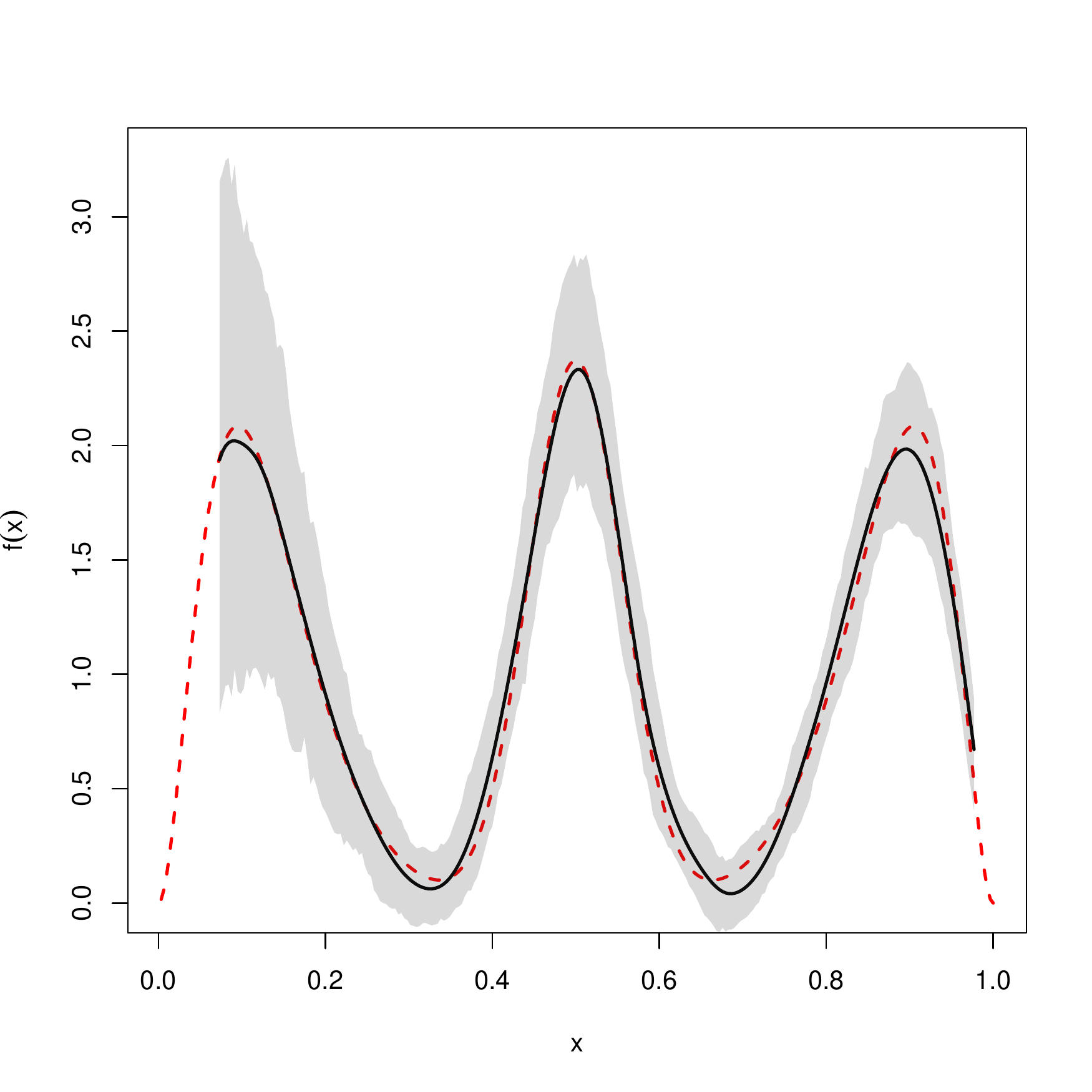}
\includegraphics[angle=0,width=0.24\linewidth]{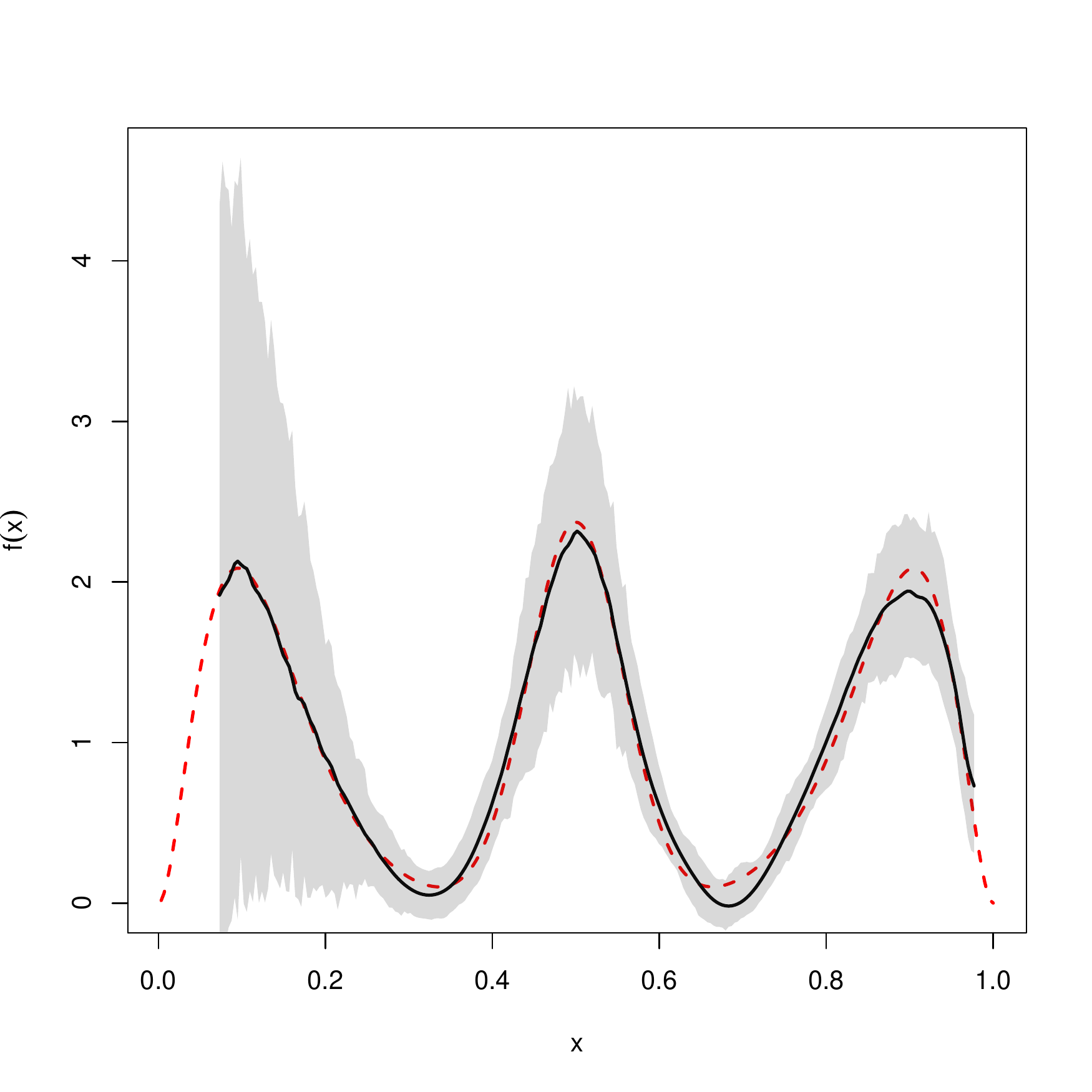}
\includegraphics[angle=0,width=0.24\linewidth]{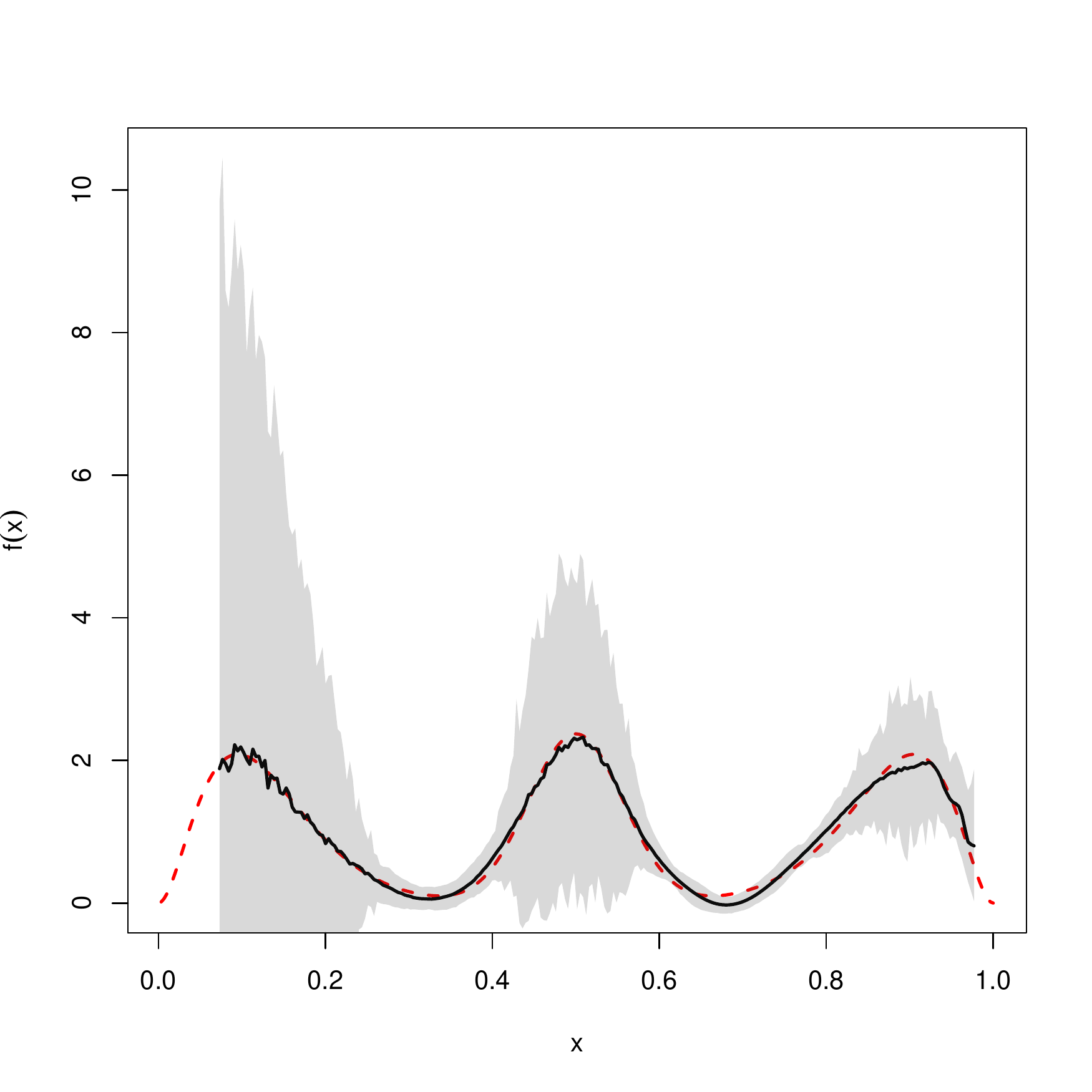} \\
\includegraphics[angle=0,width=0.24\linewidth]{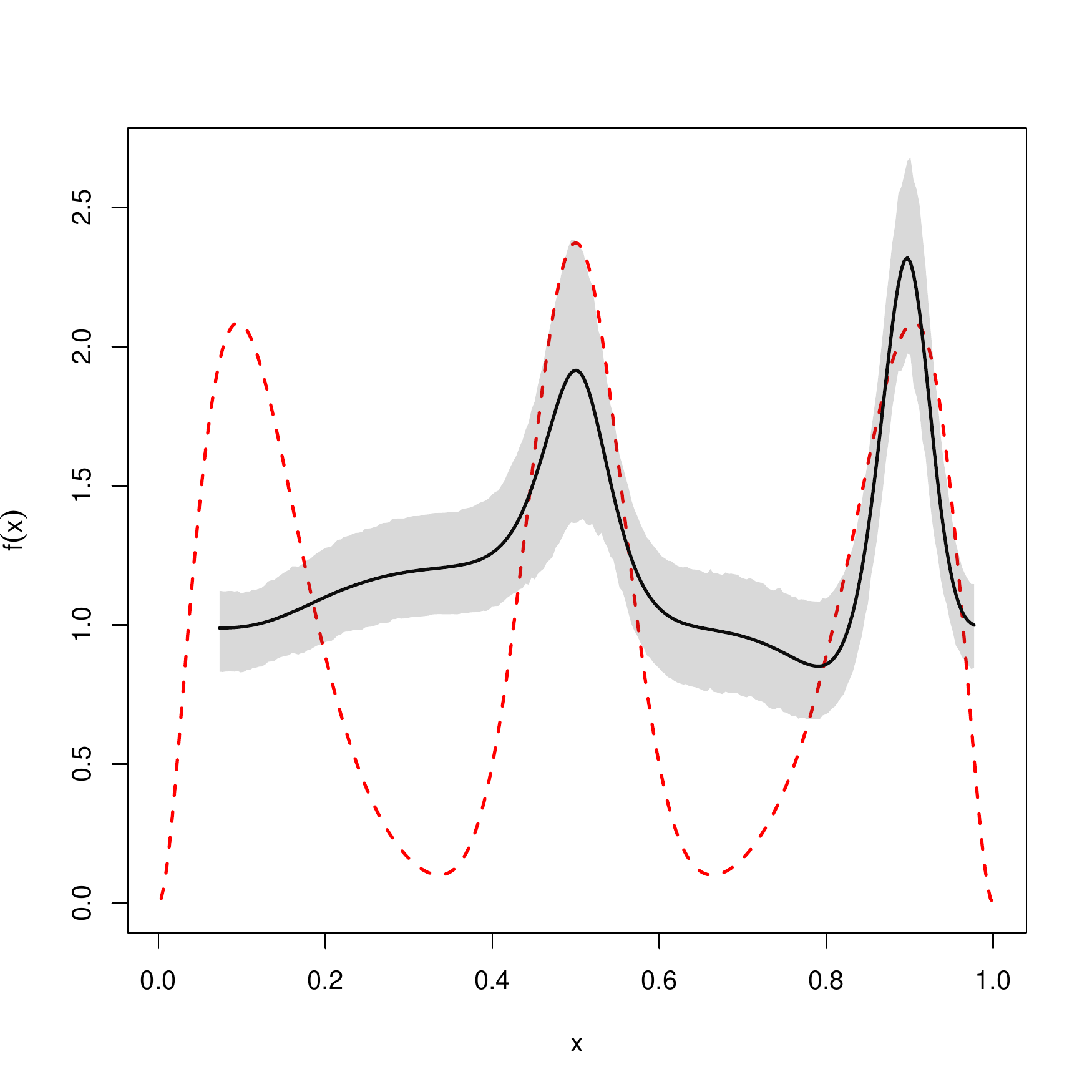}
\includegraphics[angle=0,width=0.24\linewidth]{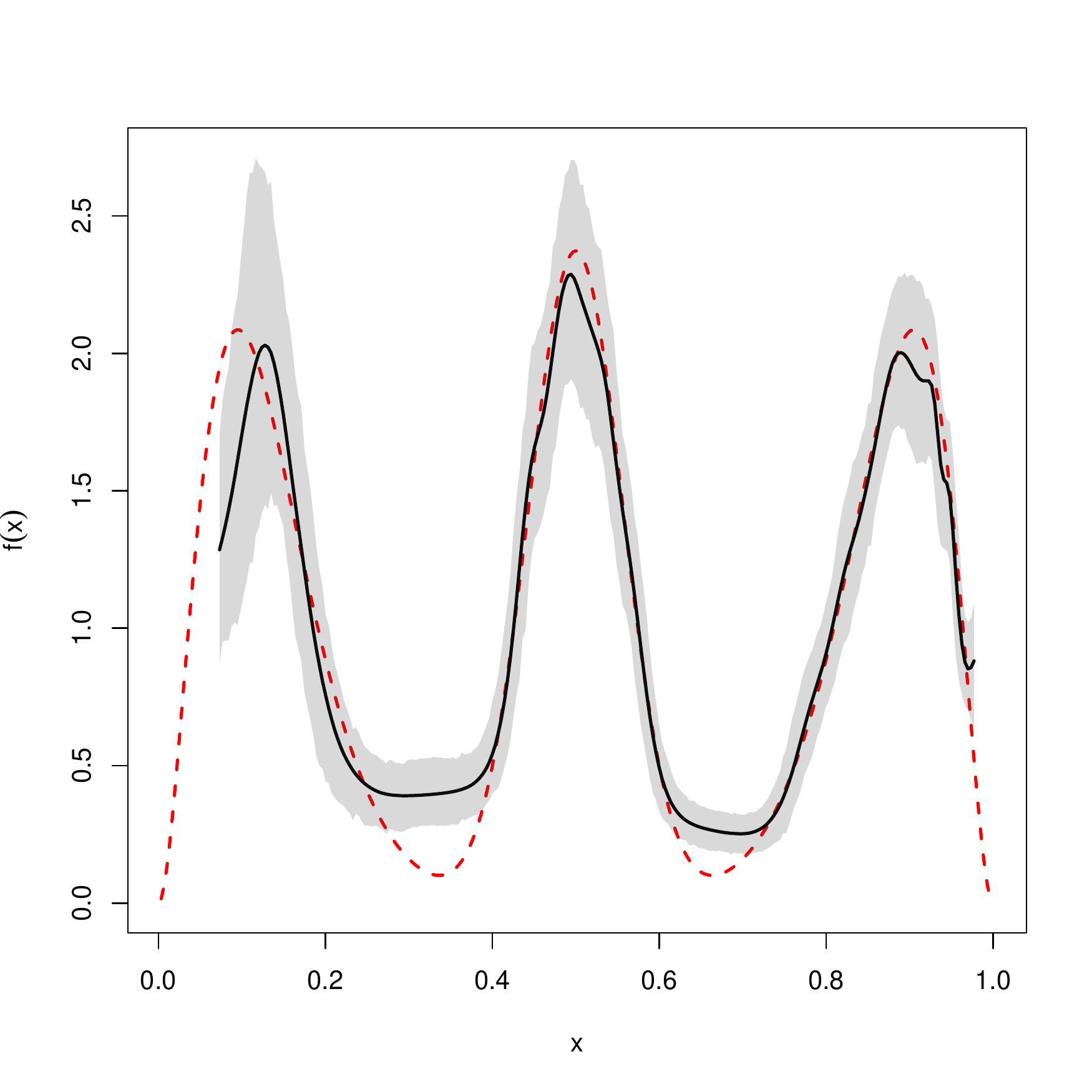}
\includegraphics[angle=0,width=0.24\linewidth]{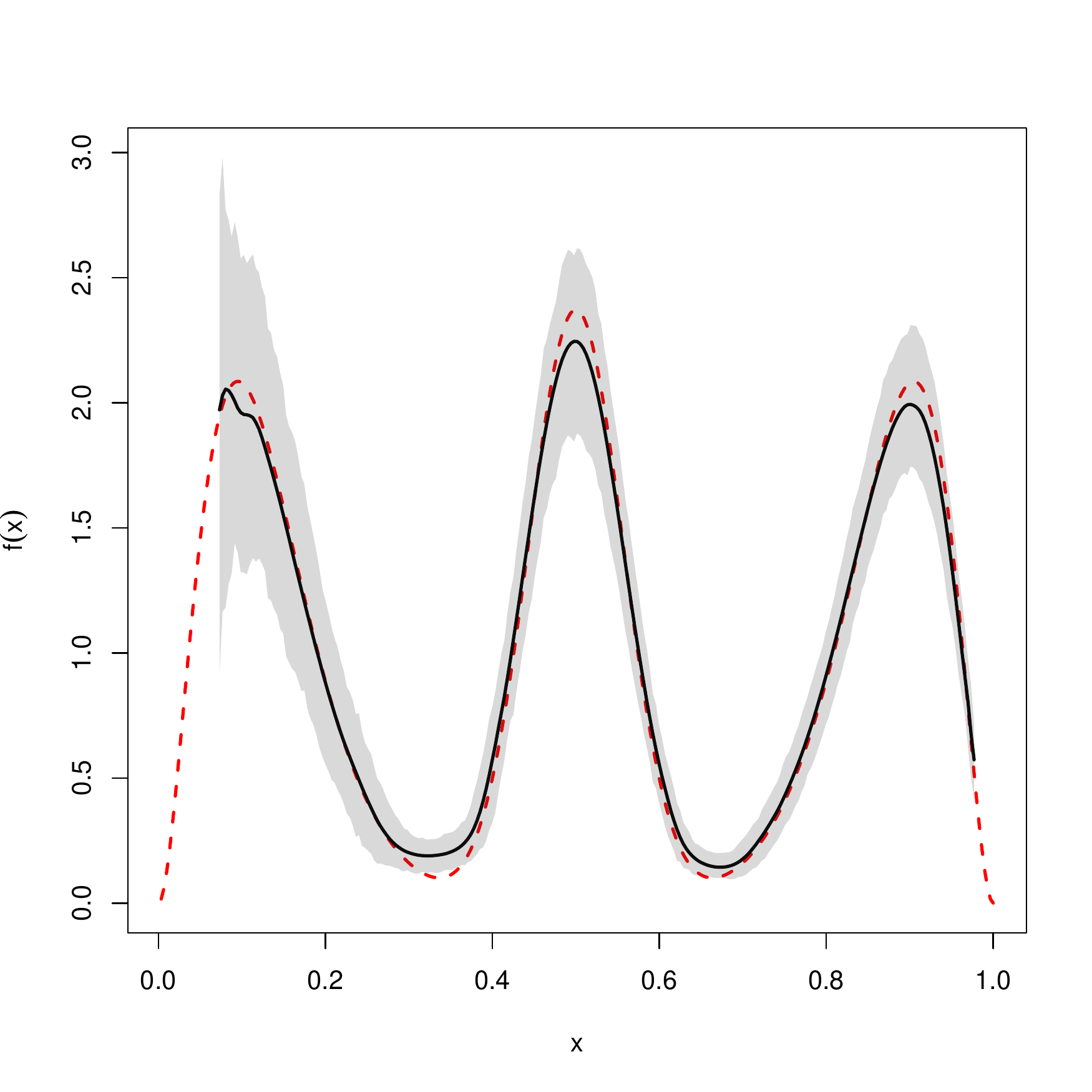}
\includegraphics[angle=0,width=0.24\linewidth]{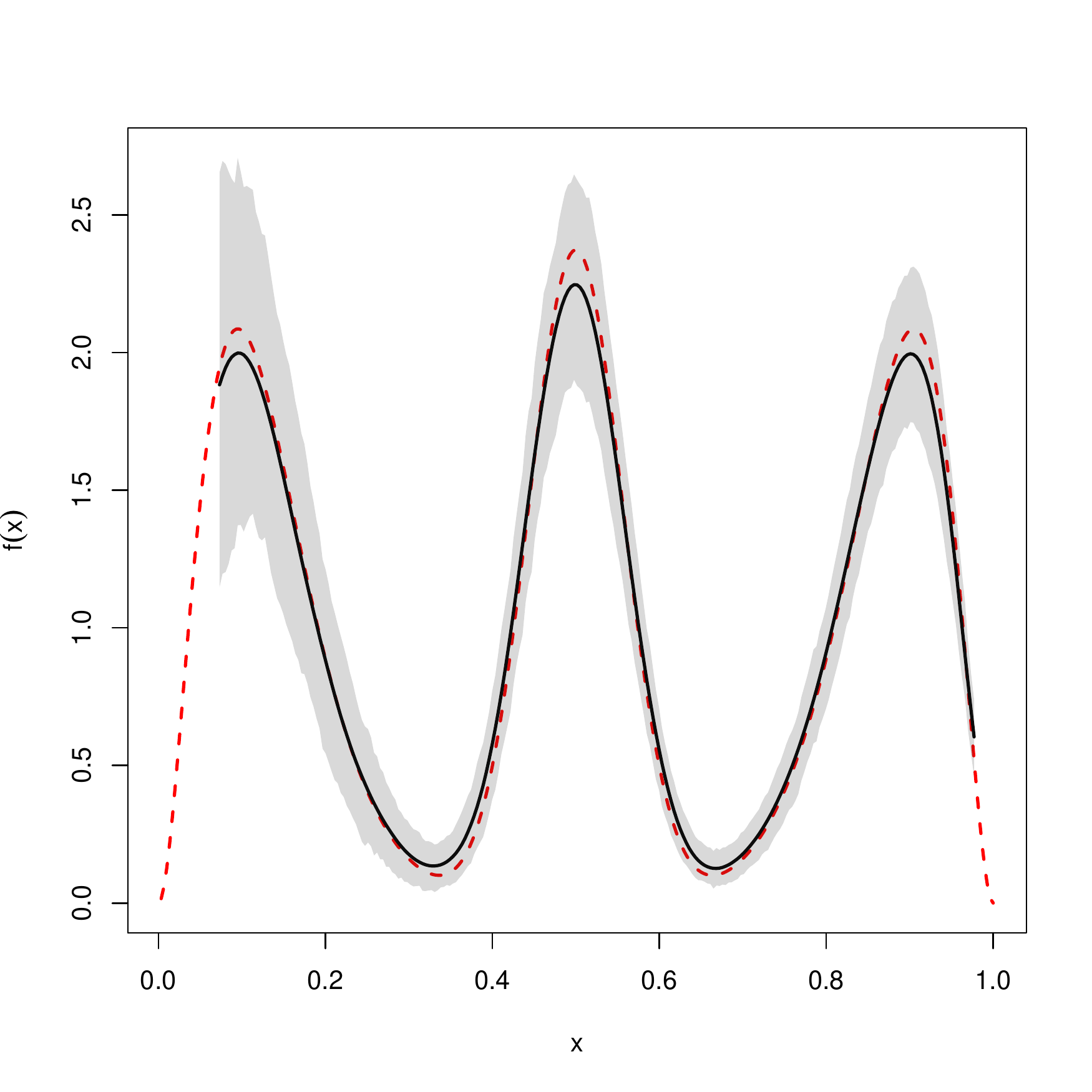} \\
\includegraphics[angle=0,width=0.24\linewidth]{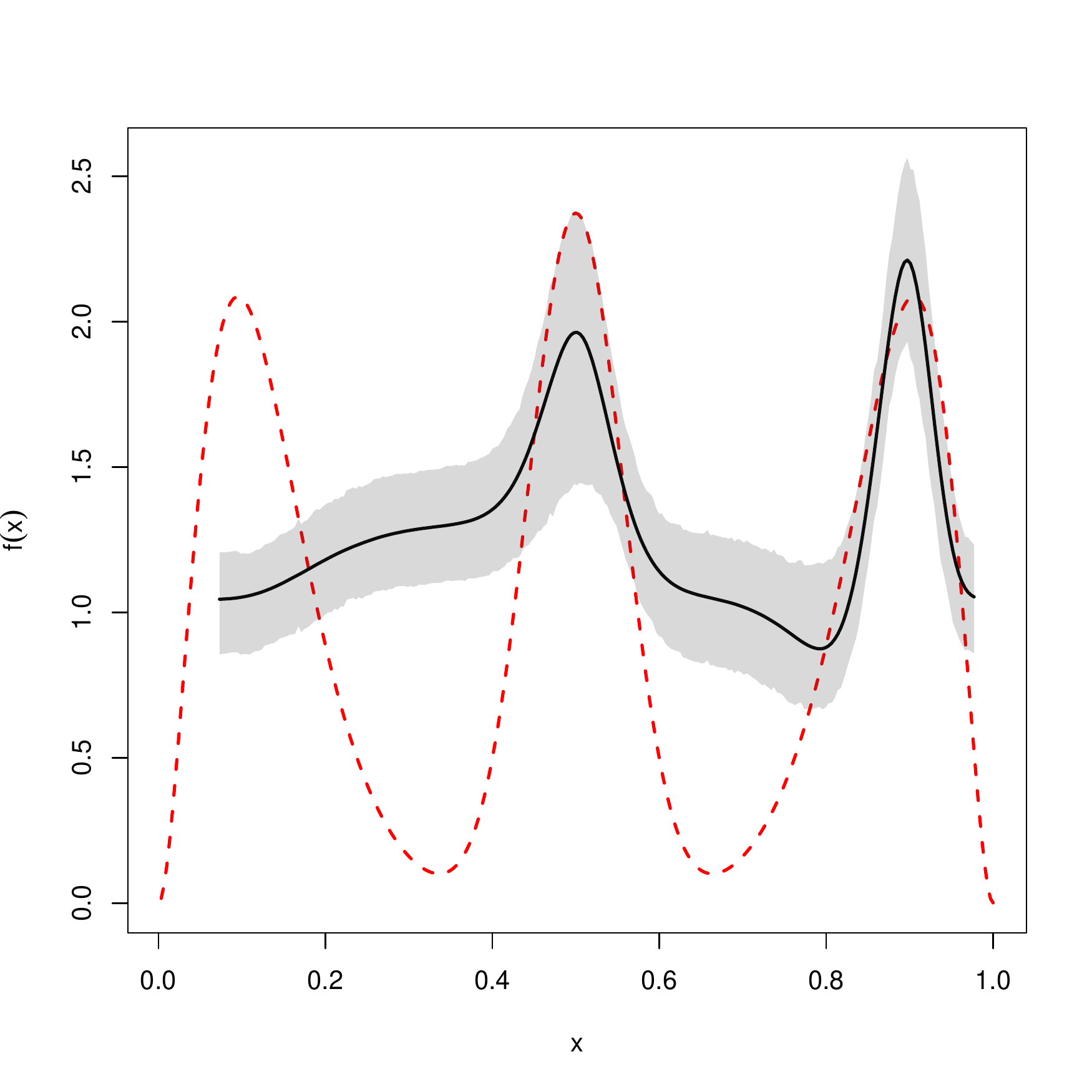}
\includegraphics[angle=0,width=0.24\linewidth]{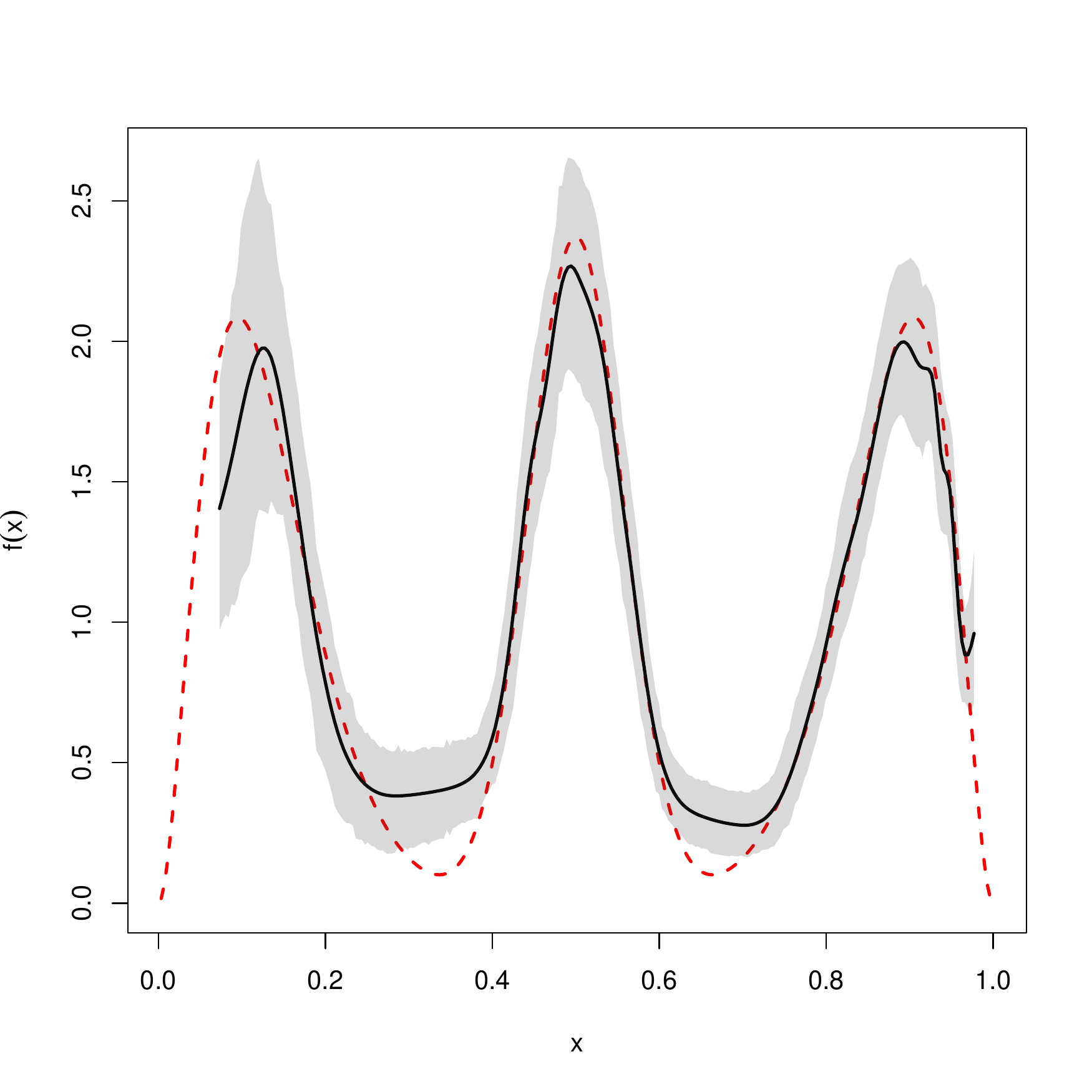}
\includegraphics[angle=0,width=0.24\linewidth]{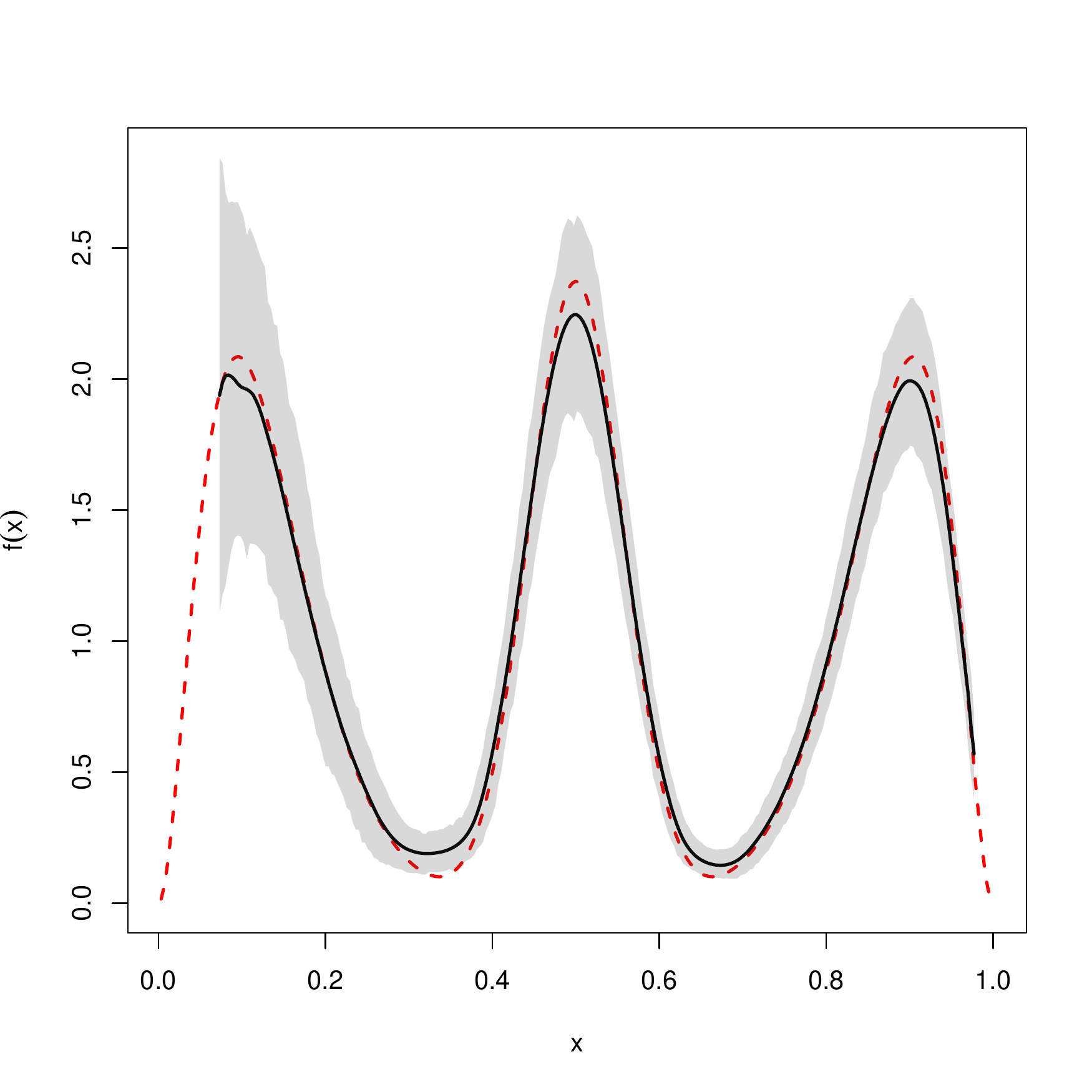}
\includegraphics[angle=0,width=0.24\linewidth]{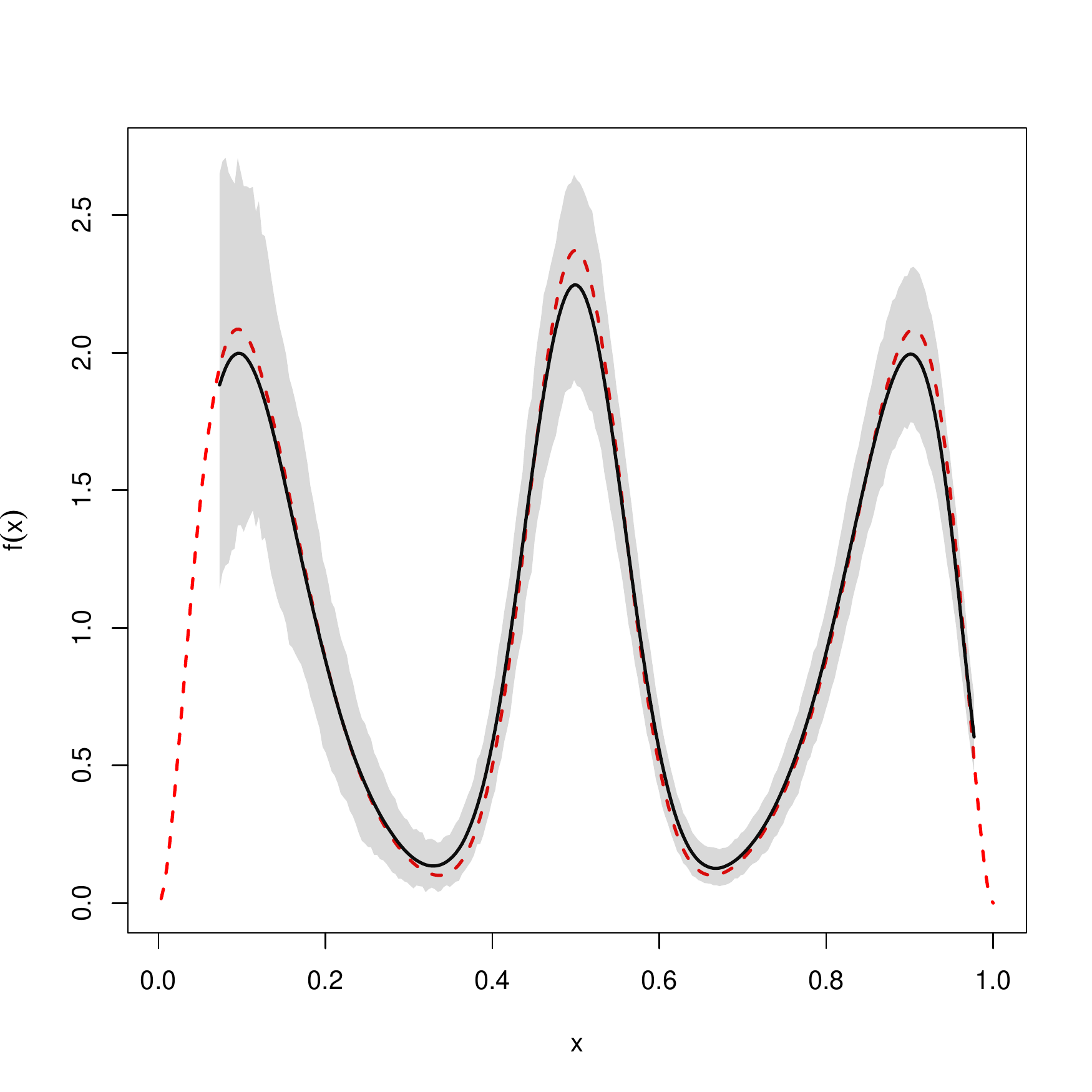} \\
\caption{Pointwise average estimates (full lines), with 95\% confidence intervals (shaded regions), for the density in $ \ex_2 $ (dashed lines), with $ n = 1,000 $. The columns 1--4 are related to the cases where we consider the finest resolution level $ J_1 = \ceil{p\log_2 n} $, $ p = 0.20, 0.45, 0.70, 0.95 $, respectively. The $ k $-th row is related to the estimation method $ m_k $, $ k = 1, 2, 3, 4 $.}
\label{fig:estimates-ex2}
\end{figure}

\begin{figure}[!ht]
\centering
\includegraphics[angle=0,width=0.24\linewidth]{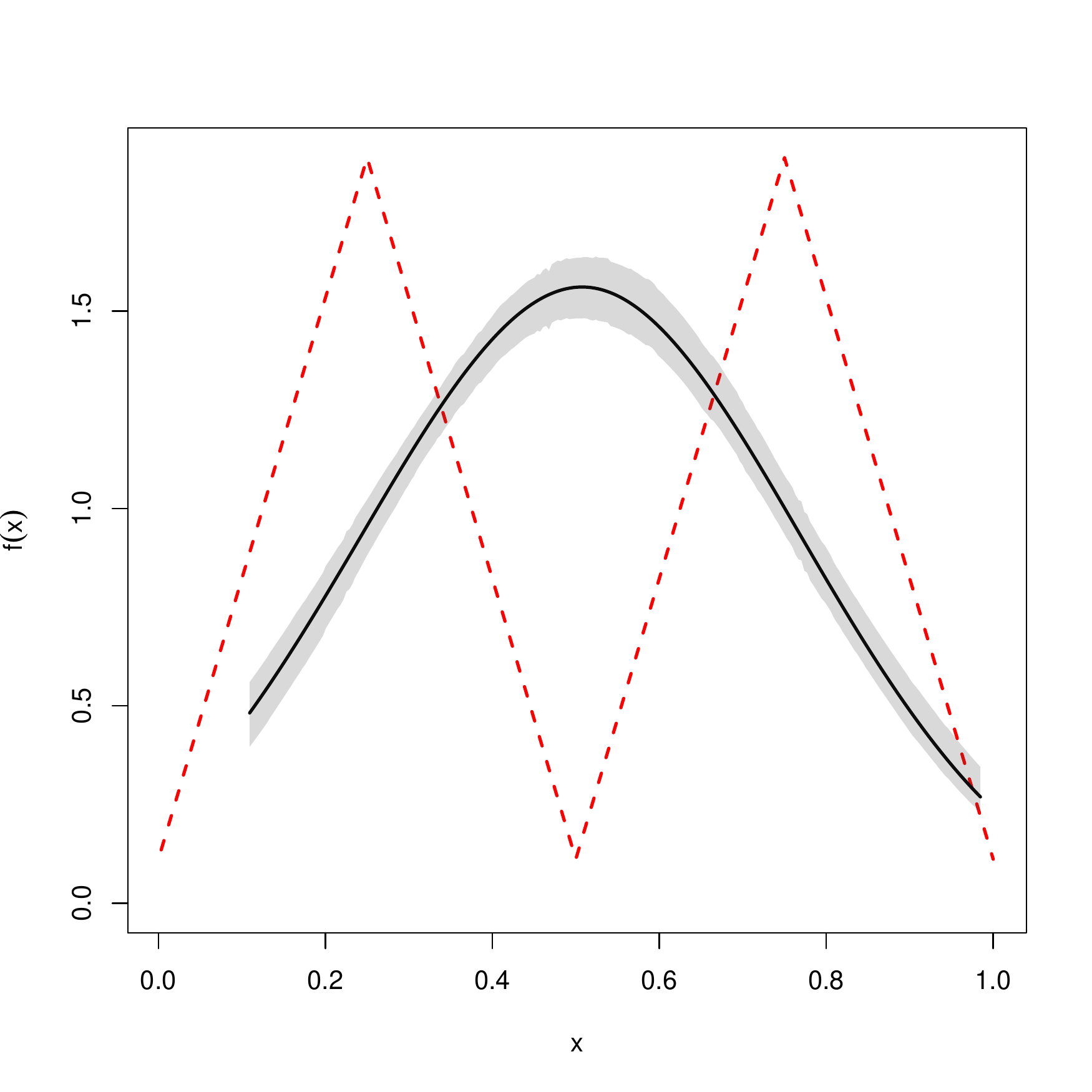}
\includegraphics[angle=0,width=0.24\linewidth]{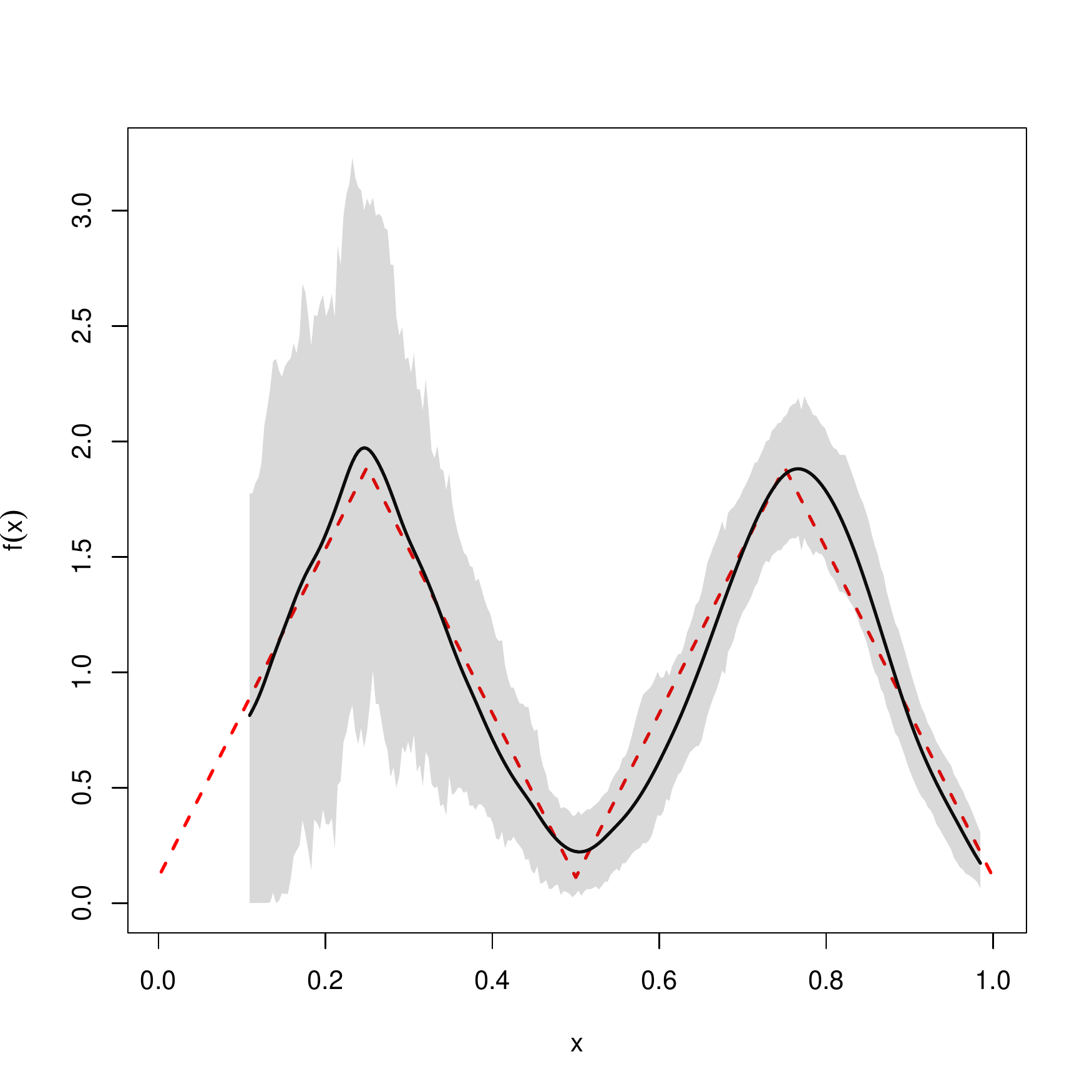}
\includegraphics[angle=0,width=0.24\linewidth]{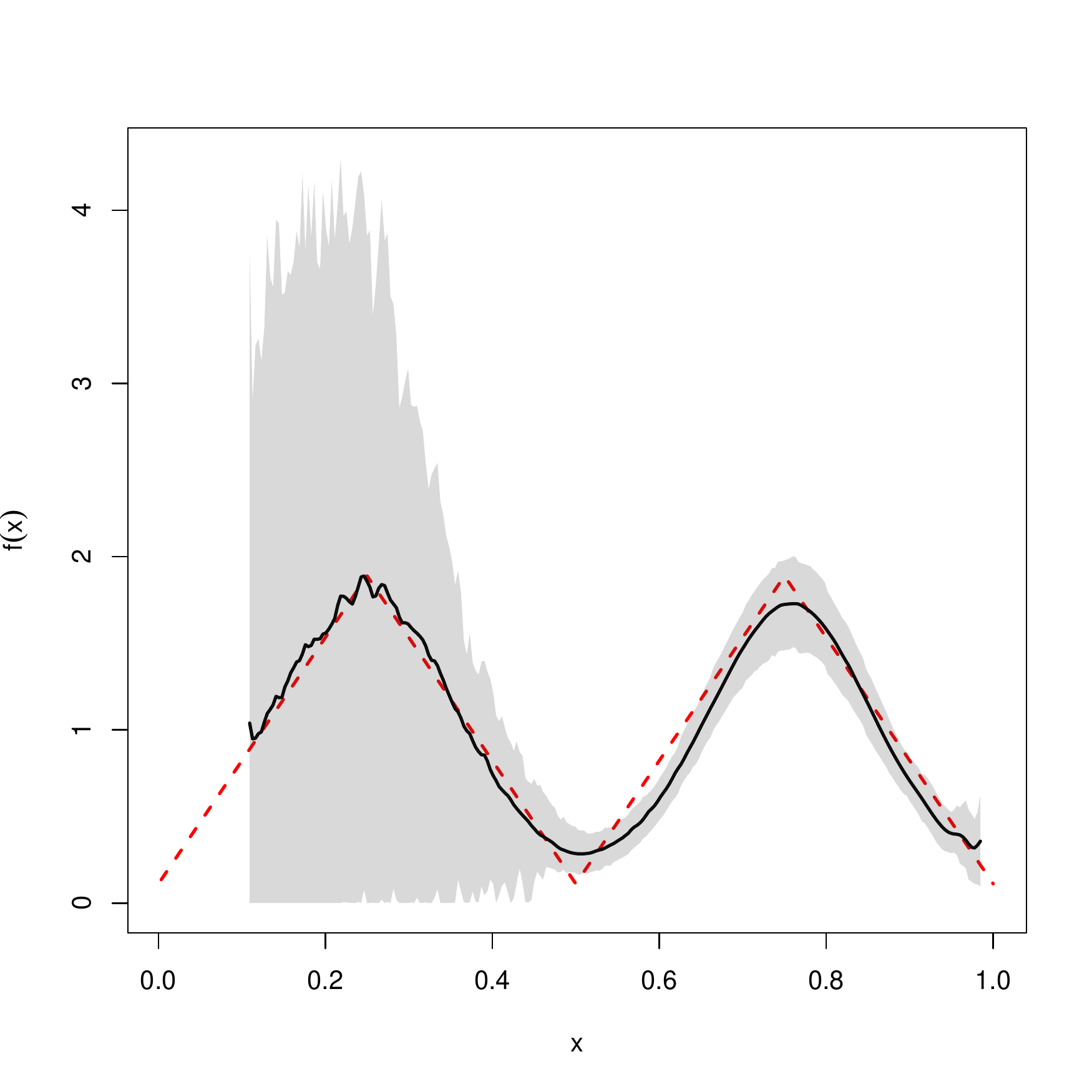}
\includegraphics[angle=0,width=0.24\linewidth]{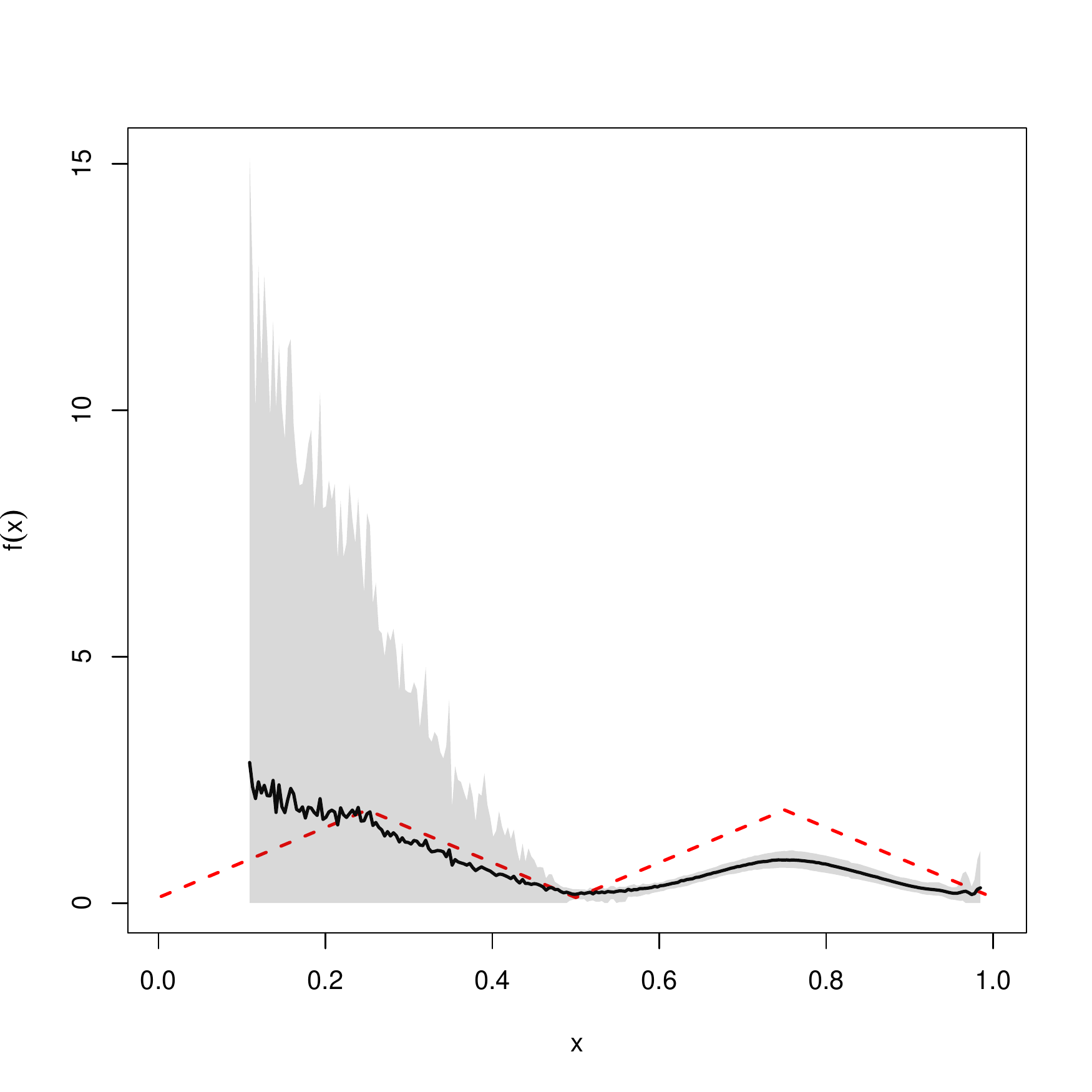} \\
\includegraphics[angle=0,width=0.24\linewidth]{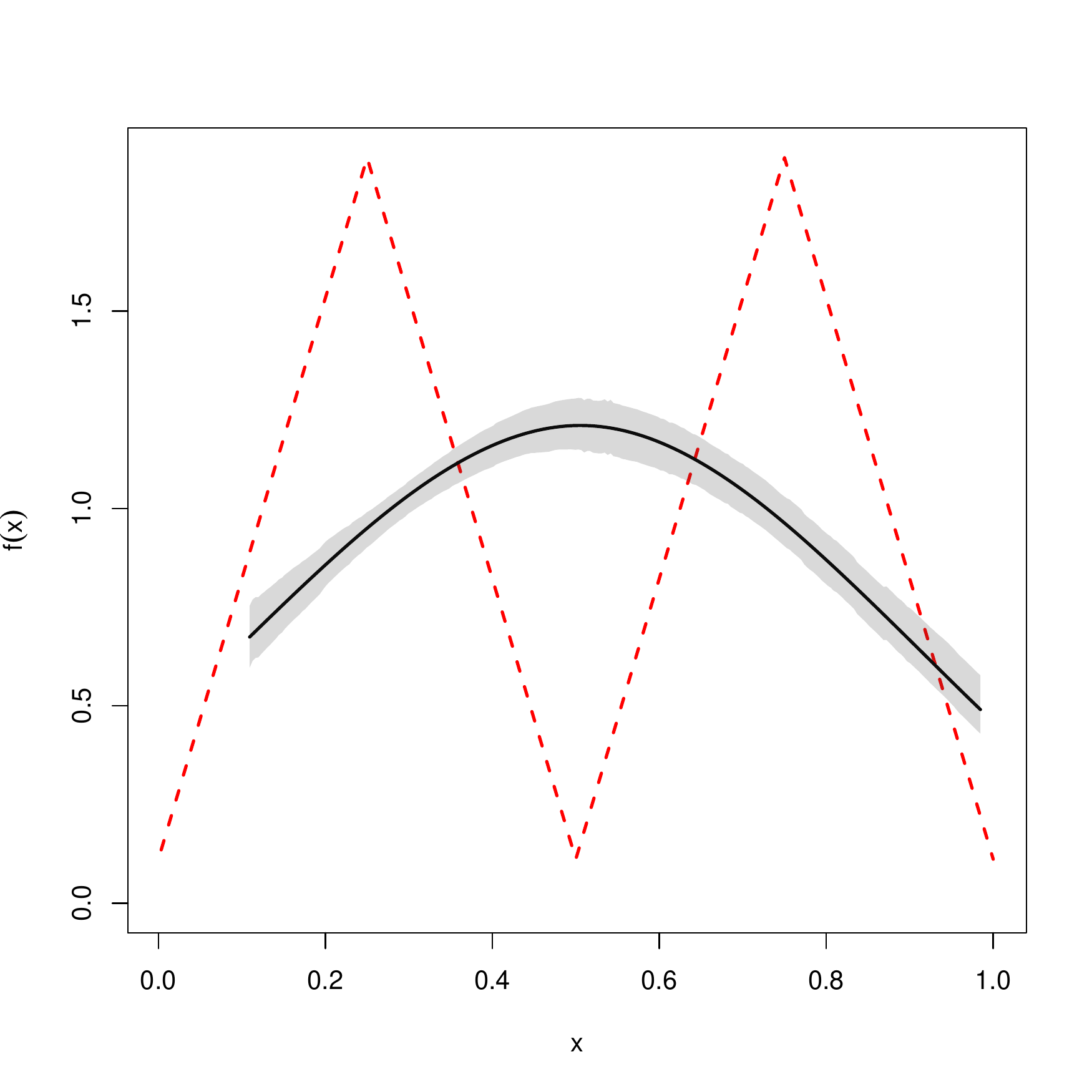}
\includegraphics[angle=0,width=0.24\linewidth]{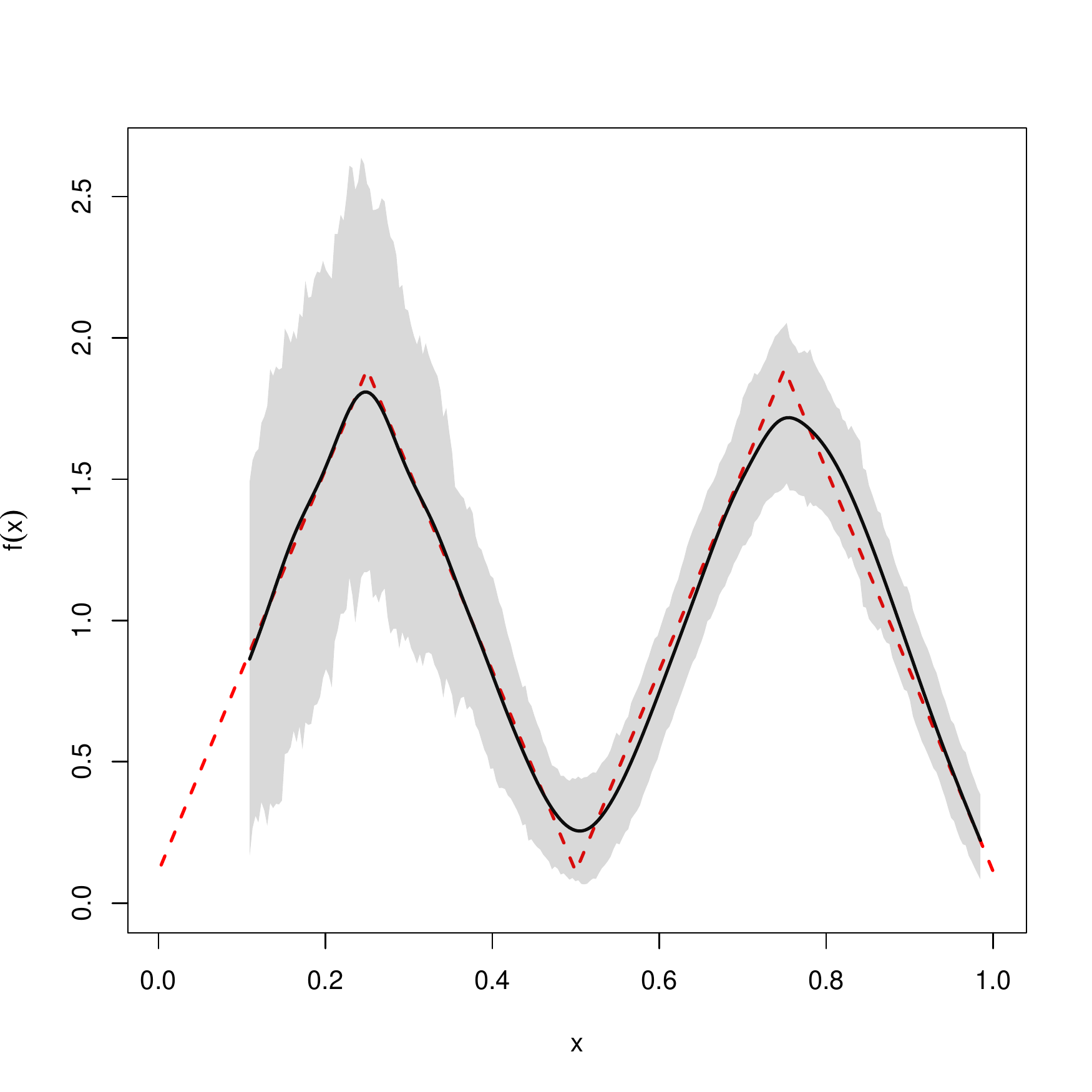}
\includegraphics[angle=0,width=0.24\linewidth]{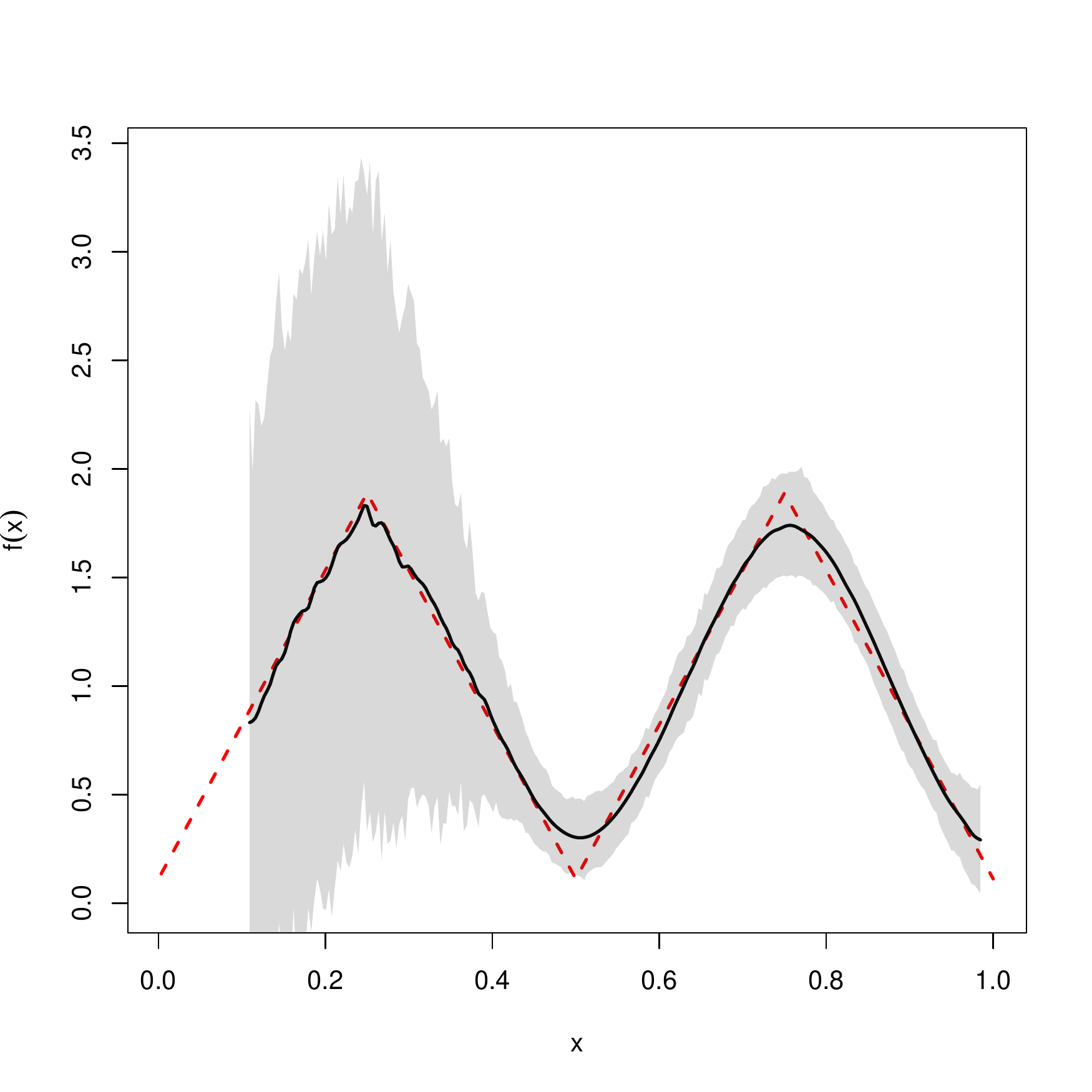}
\includegraphics[angle=0,width=0.24\linewidth]{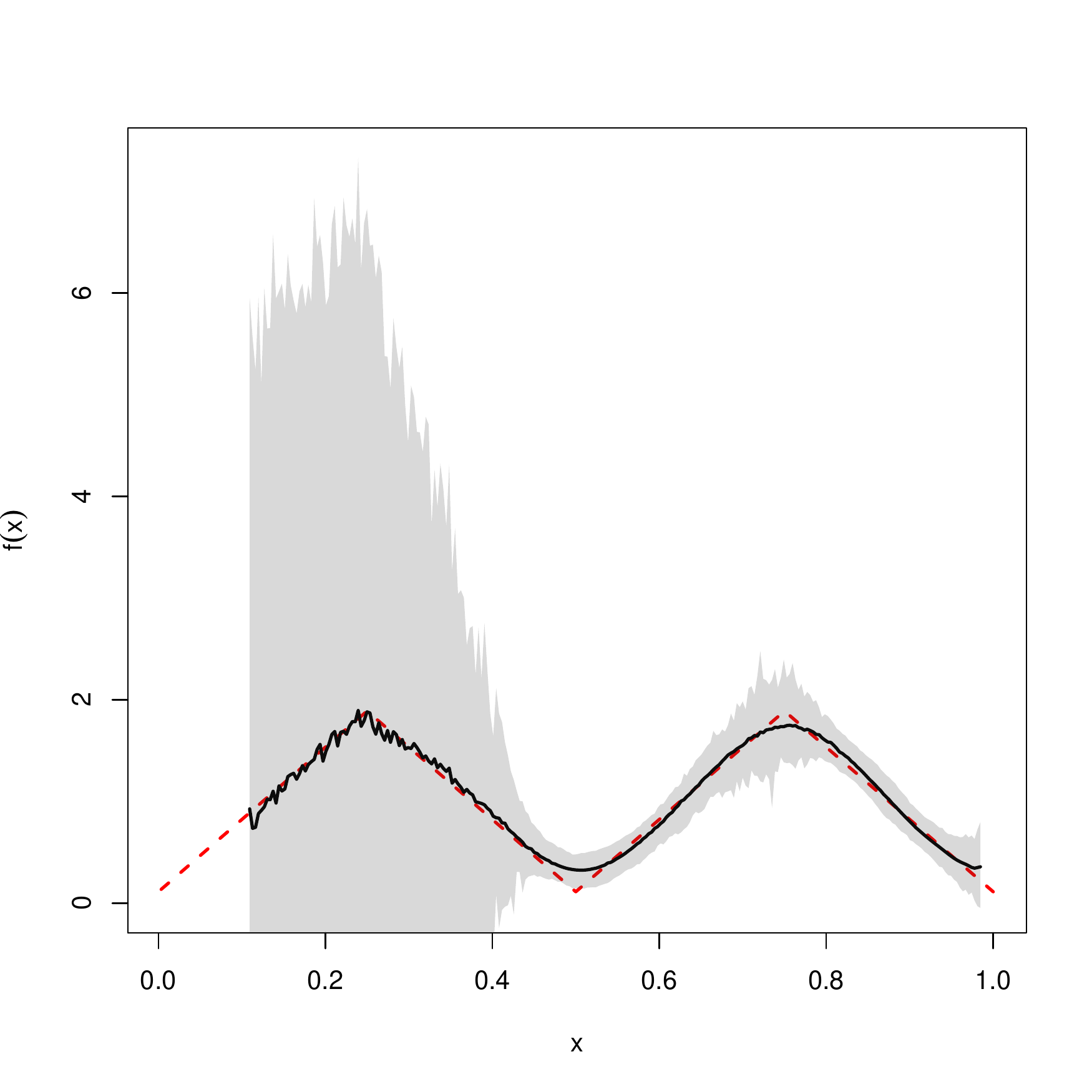} \\
\includegraphics[angle=0,width=0.24\linewidth]{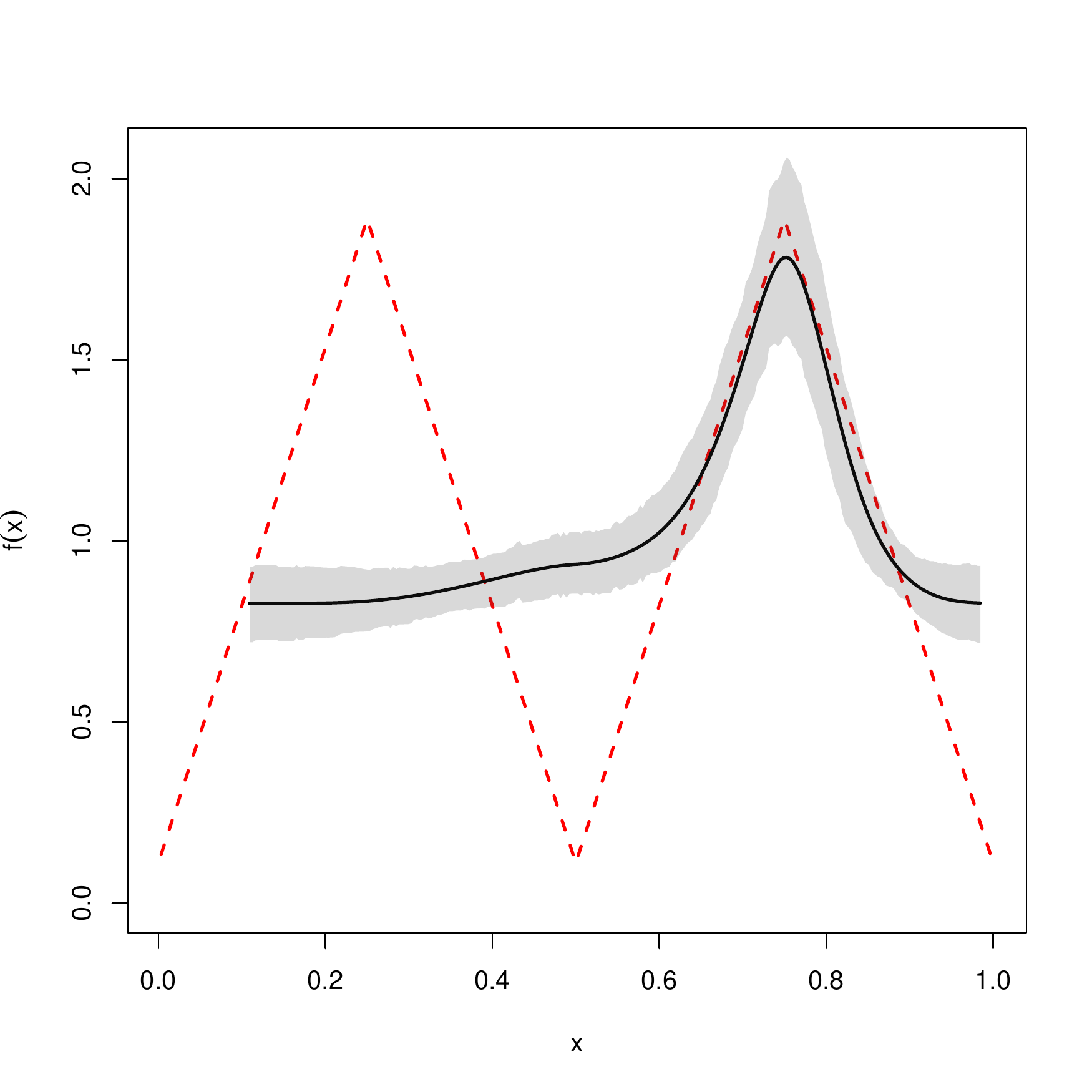}
\includegraphics[angle=0,width=0.24\linewidth]{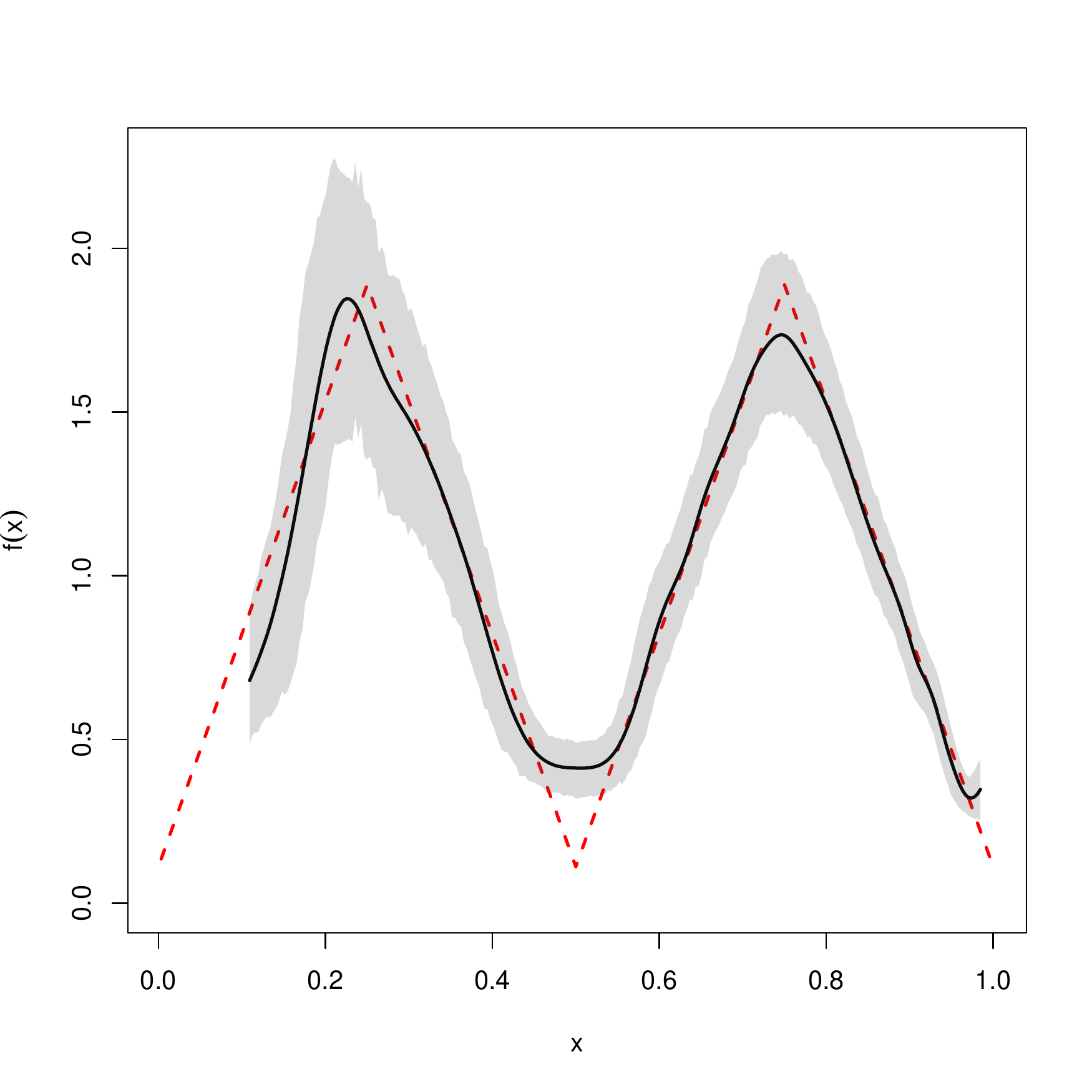}
\includegraphics[angle=0,width=0.24\linewidth]{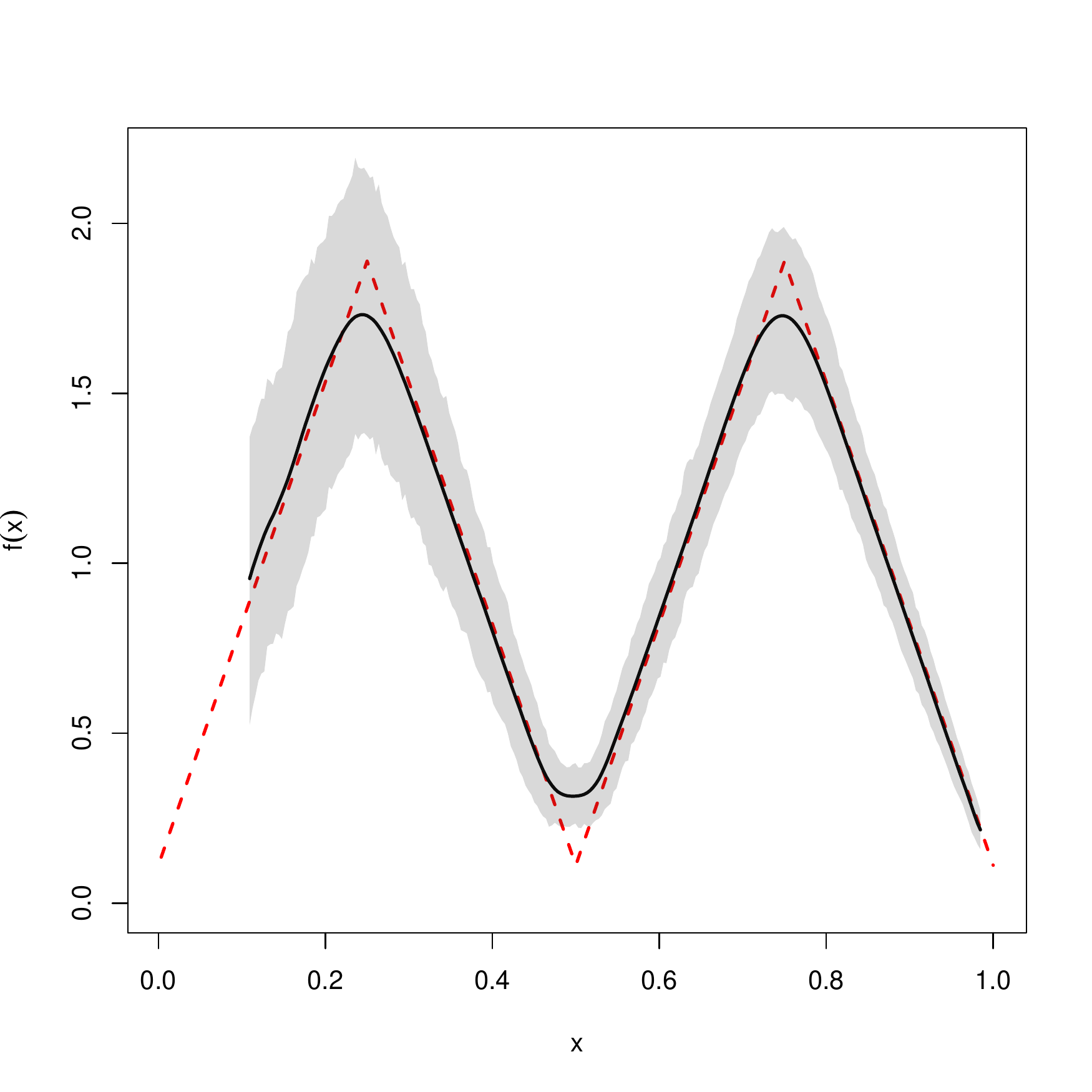}
\includegraphics[angle=0,width=0.24\linewidth]{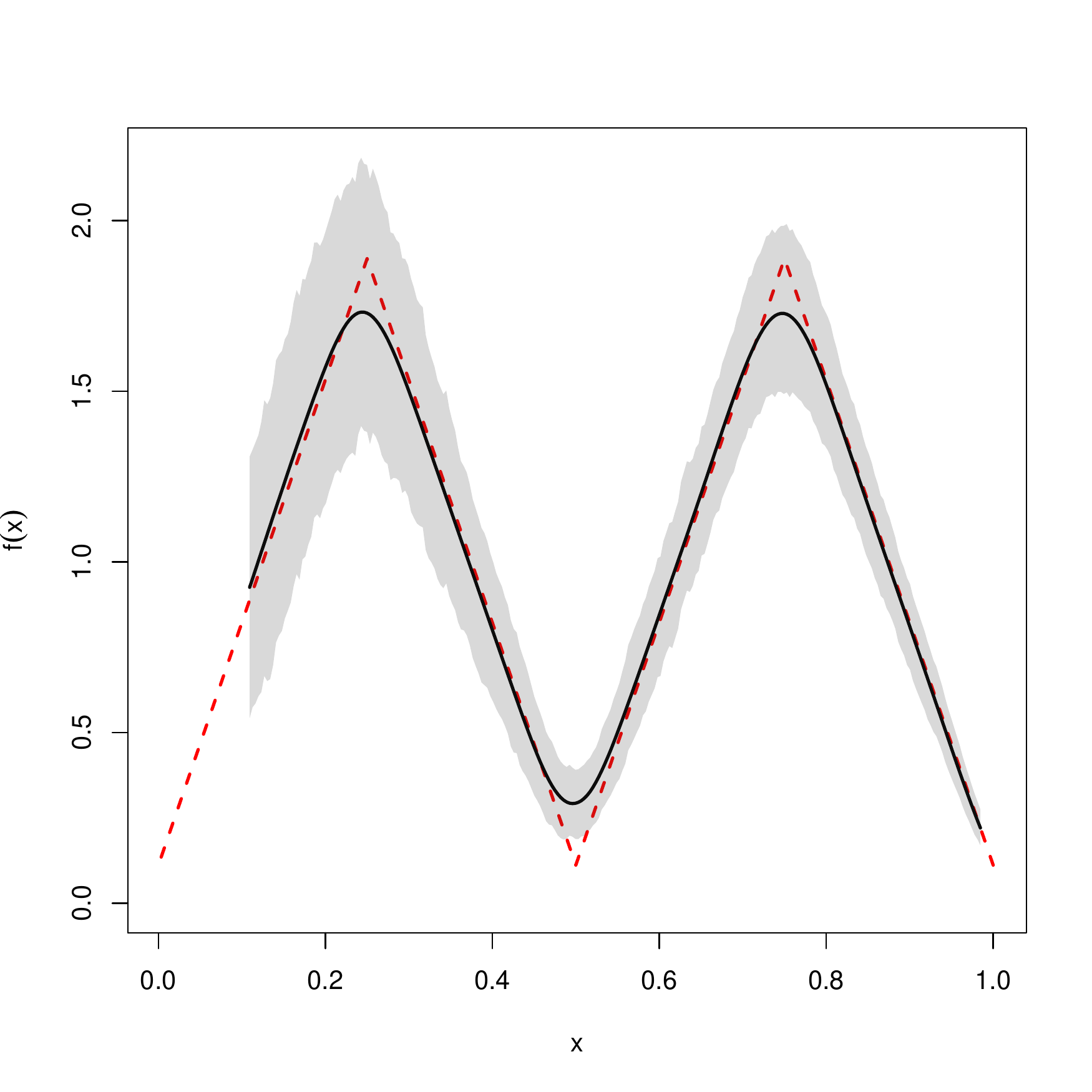} \\
\includegraphics[angle=0,width=0.24\linewidth]{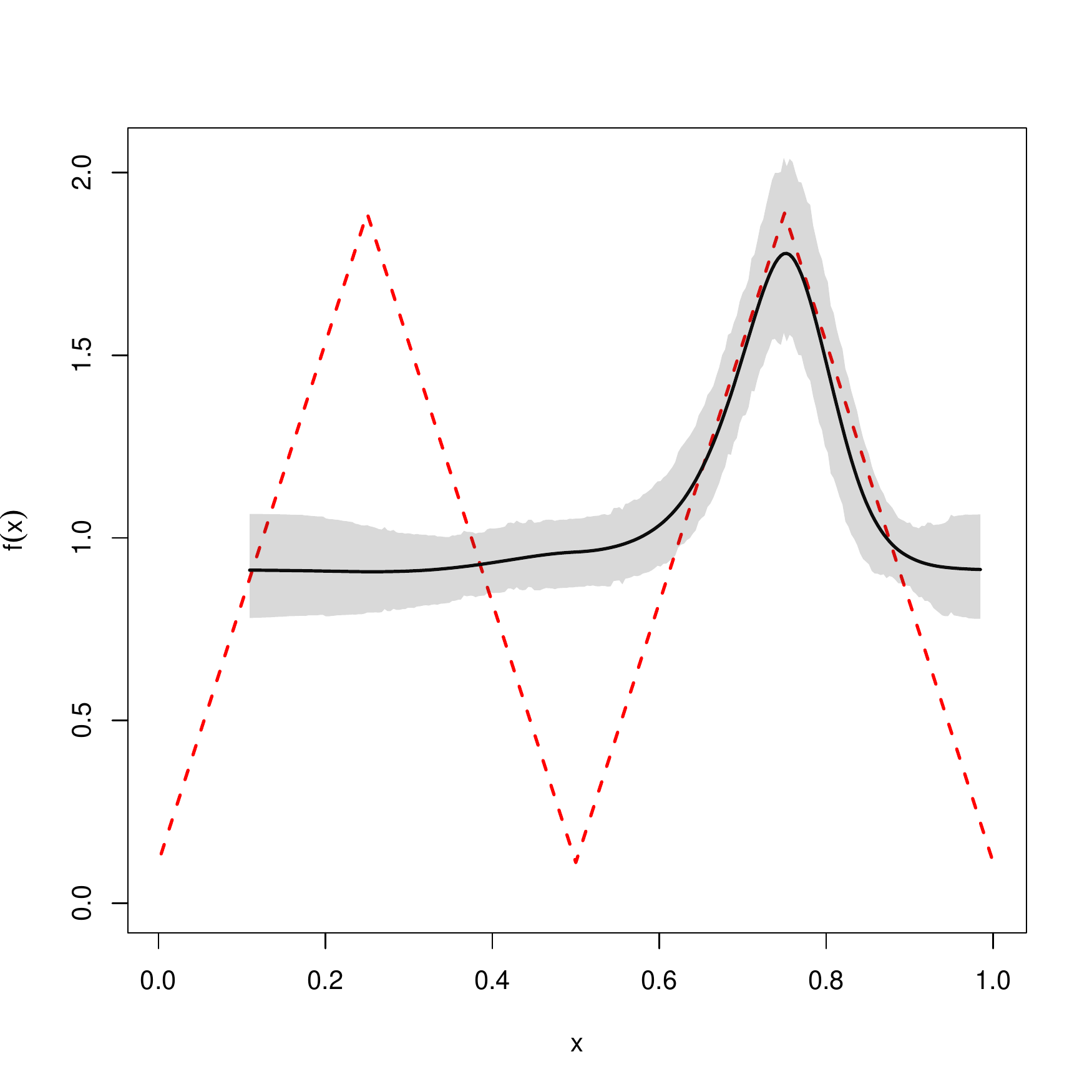}
\includegraphics[angle=0,width=0.24\linewidth]{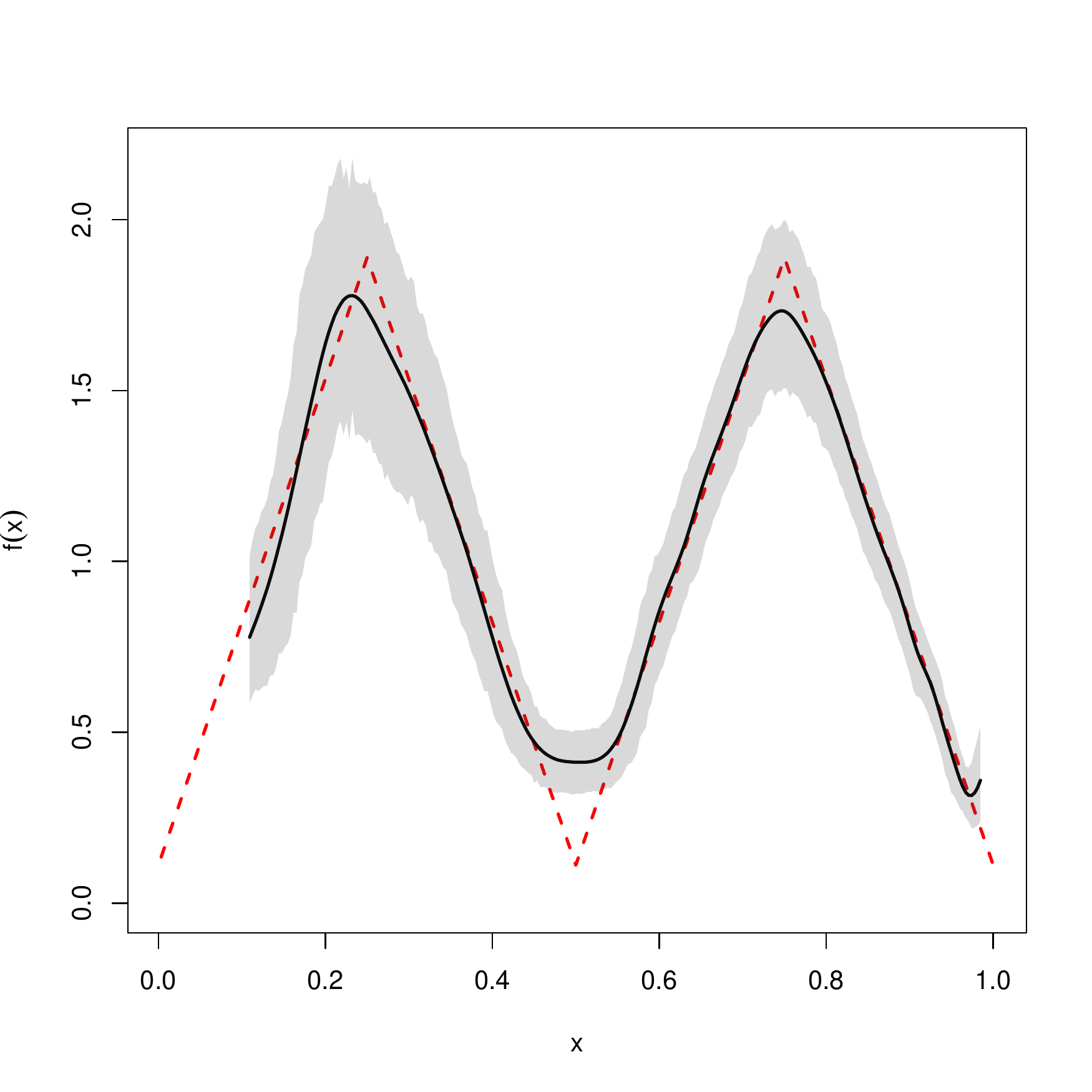}
\includegraphics[angle=0,width=0.24\linewidth]{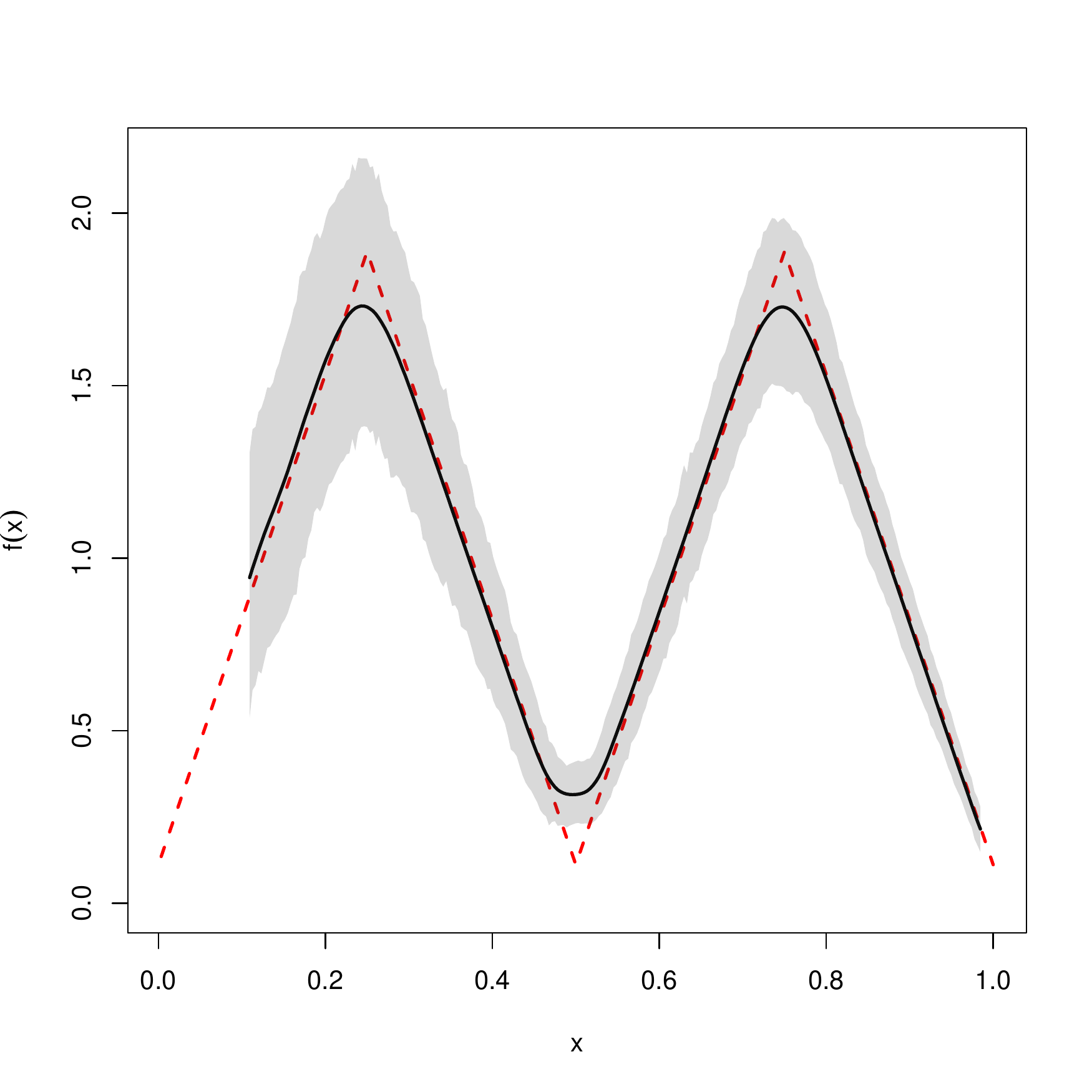}
\includegraphics[angle=0,width=0.24\linewidth]{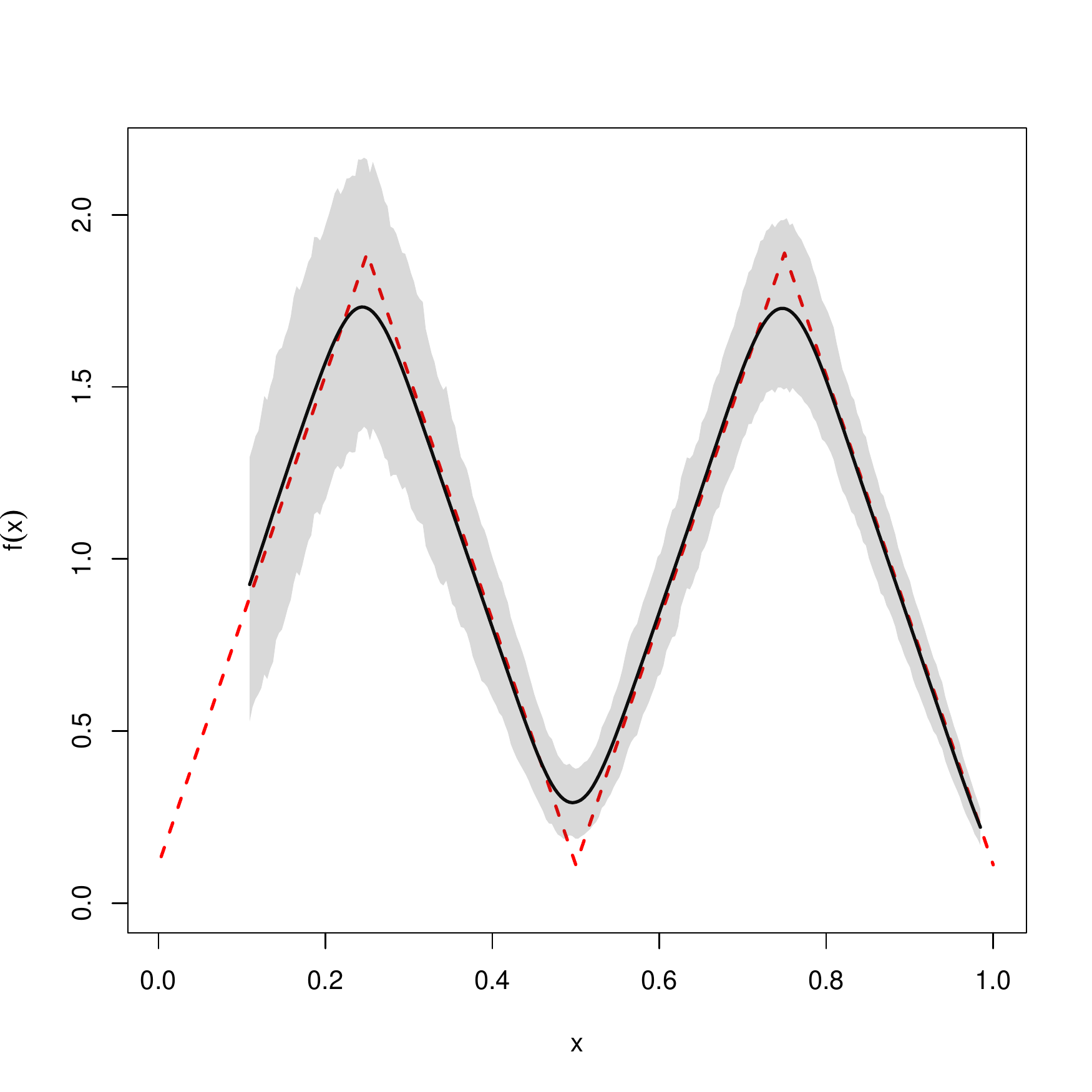} \\
\caption{Pointwise average estimates (full lines), with 95\% confidence intervals (shaded regions), for the density in $ \ex_3 $ (dashed lines), with $ n = 1,000 $. The columns 1--4 are related to the cases where we consider the finest resolution level $ J_1 = \ceil{p\log_2 n} $, $ p = 0.20, 0.45, 0.70, 0.95 $, respectively. The $ k $-th row is related to the estimation method $ m_k $, $ k = 1, 2, 3, 4 $.}
\label{fig:estimates-ex3}
\end{figure}

\subsection{Application}\label{sec:application}

Let us consider the dataset of 2,495 blood alcohol concentrations (BAC) of drivers involved in fatal accidents that occurred in the USA, during the year of 2019. The data was collected from the National Highway Traffic Safety Administration Department of Transportation (\url{www.nhtsa.dot.gov}). It is part of The Fatality Analysis Reporting System (FARS), from where we get  the brief description of the data (more details in this \href{https://www.nhtsa.gov/crash-data-systems/fatality-analysis-reporting-system}{link}).

\begin{quote}
The Fatality Analysis Reporting System (FARS) became operational in 1975, and contains data of fatal traffic crashes within the 50 States, the District of Columbia, and Puerto Rico. To be included in FARS, a crash must involve a motor vehicle traveling on a traffic way customarily open to the public, and must result in the death of a vehicle occupant or a nonoccupant within 30 days of the crash.
\end{quote}

BAC here is expressed in grams/100 ml. According to the 2019\linebreak FARS/CRSS Coding and Validation Manual (available \href{https://static.nhtsa.gov/nhtsa/downloads/FARS/FARS-DOC/Coding\%20and\%20Validation\%20Manual/2019\%20FARS\% 20CRSS\%20Coding\%20and\%20Validation\%20Manual\%20-\%20DOT\%20HS\%20813\%20010.pdf} {here}), we consider only fatal accidents where alcohol is involved (according to the police report). Moreover, crashes that are not included in the state highway inventory, not reported or unknown were discarded. Finally, we considered vehicles classified as automobiles, automobiles derivatives, utility vehicles and two-wheel motorcycles within the 50 states, the District of Columbia and Puerto Rico during 1975. %It was obtained from the National Highway Traffic Safety Administration Department of Transportation (\url{www.nhtsa.dot.gov}) and is the same used by \cite{Ramirez.Vidakovic-2010-JSPI}, where the reader can find more details. The BAC dataset is represented by realizations of continuous \rv's expressed in grams/100ml.

\begin{figure}[!htb]
\centering
\includegraphics[angle=0,width=0.9\linewidth]{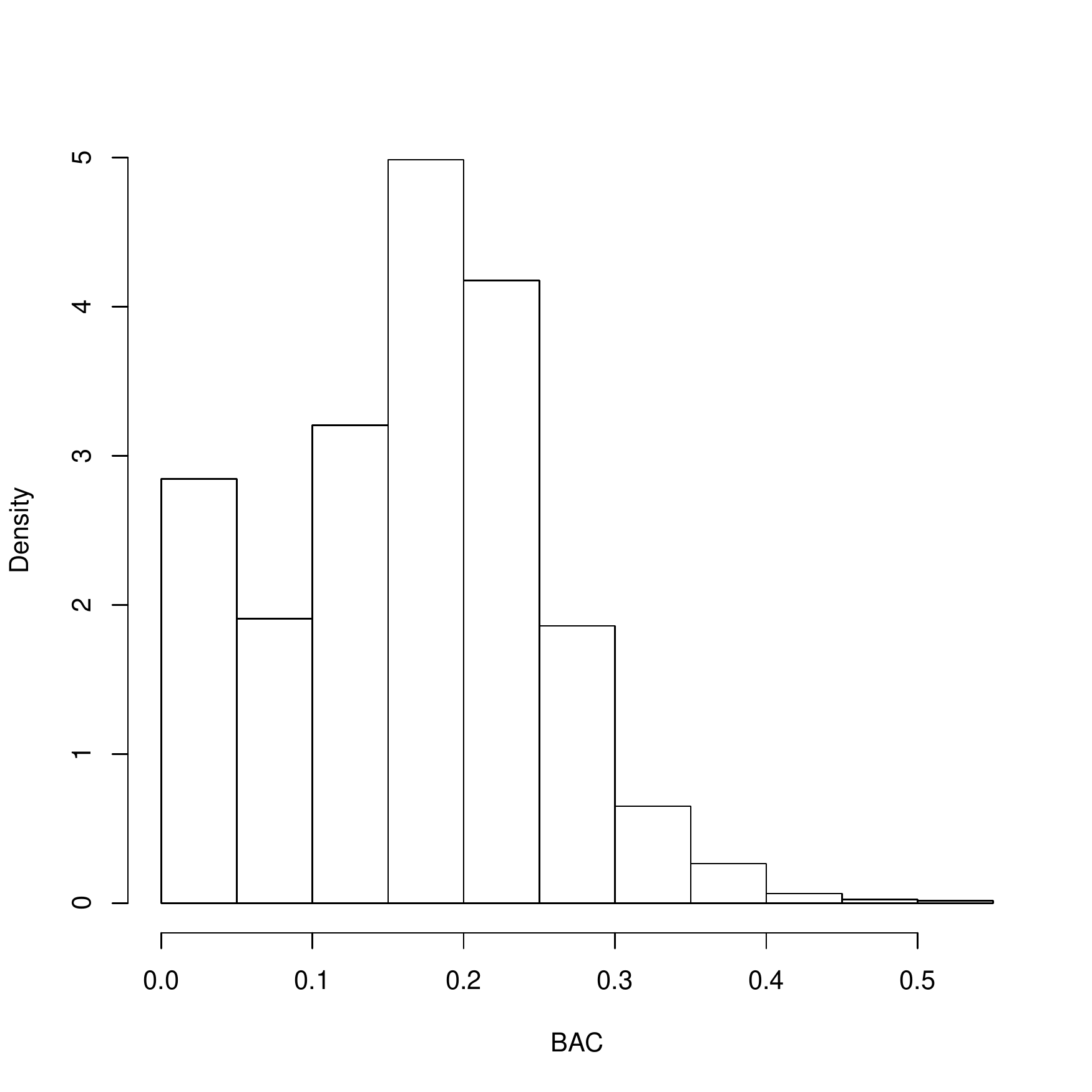}
\caption{Histogram of the BAC values of drivers involved in fatal accidents in the USA in 2019.}
\label{fig:bac-hist}
\end{figure}

This is a typical example of a size-biased data. Indeed, as discussed by \cite{Efromovich-1999}, drunk drivers are more likely be involved in fatal accidents. A histogram is presented in Figure \ref{fig:bac-hist}. The data is mainly concentrated around 0.10 -- 0.25 grams/100 ml. Moreover, the range of observations belongs to the unity interval, with maximum value smaller than 0.55 grams/100 ml, which is not close to one.

The arguments above suggest that the density of interest is biased by an increasing biasing function. Such a function is unknown in practice, and its choice is usually related to historic data, nature of phenomenon and/or common sense \cite{Ramirez.Vidakovic-2010-JSPI}. In general, the biasing function should be studied by additional experiments, but in many cases a linear behavior is recommended \cite{Efromovich-1999}. Therefore, in this analysis, we assume that
\[ w(x) = 0.1 + 0.9x. \]

As the data belongs to the unit interval, with its maximum ``far'' from 1 gram/100 ml, no transformation is needed (Section \ref{sec:com-asp}). Based on  Section \ref{sec:simul}, we consider $ J_1 = \ceil{0.45 \log_2 2495} = 6 $ for $ m_1 $ and $ m_2 $, and $ J_1 = \ceil{0.95 \log_2 2495} = 11 $ for $ m_3 $ and $ m_4 $.

\begin{figure}[!htb]
\centering
\includegraphics[angle=0,width=0.49\linewidth]{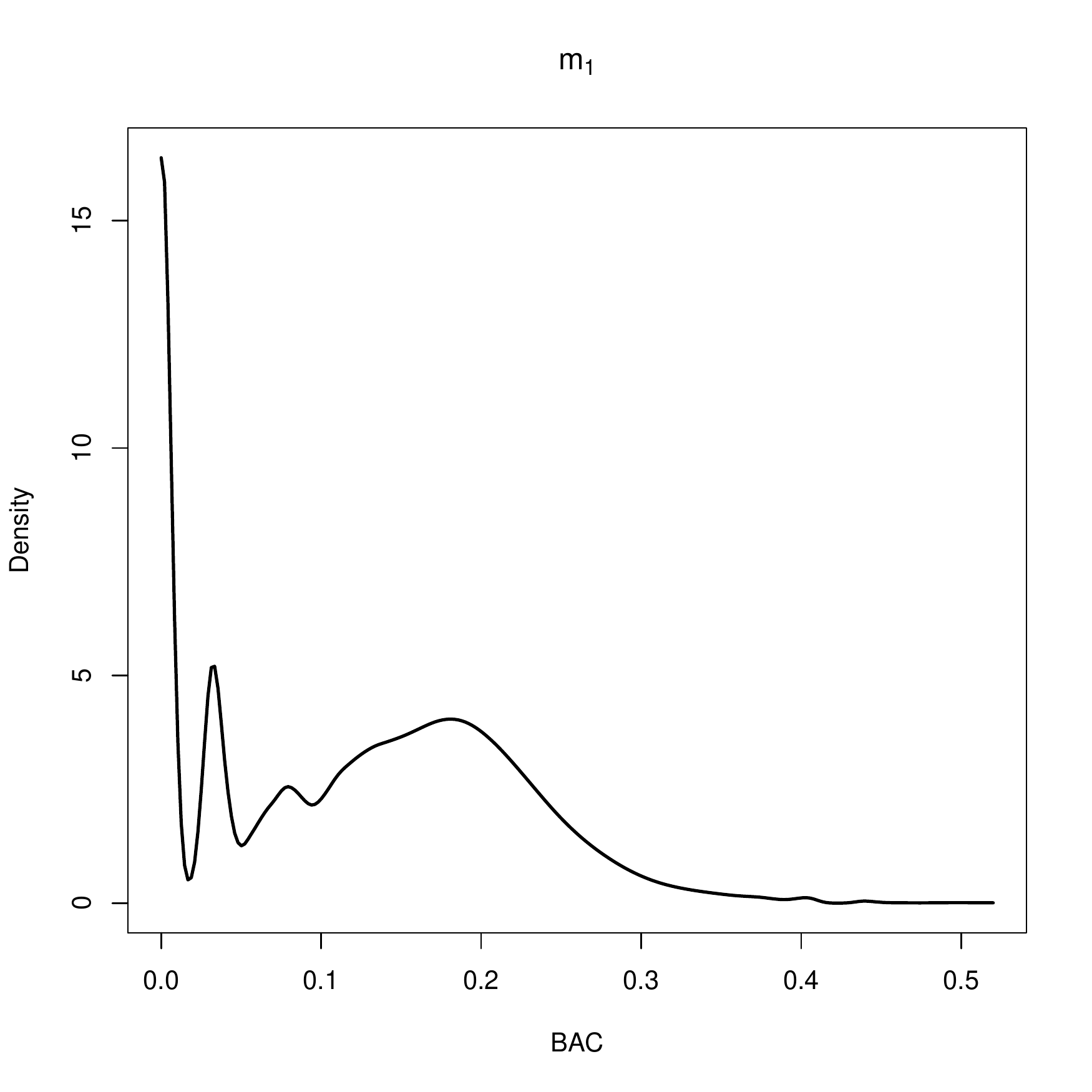}
\includegraphics[angle=0,width=0.49\linewidth]{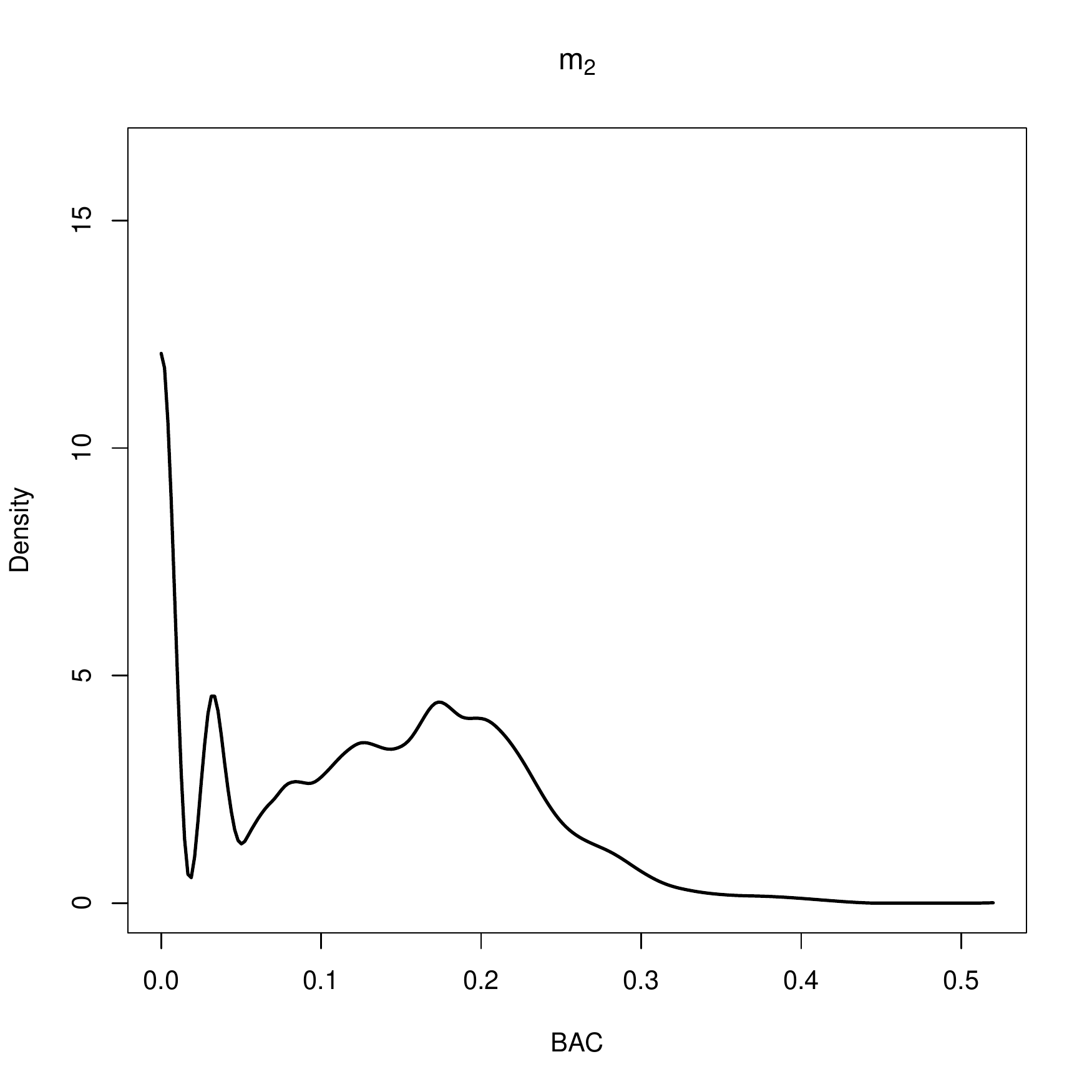} \\
\includegraphics[angle=0,width=0.49\linewidth]{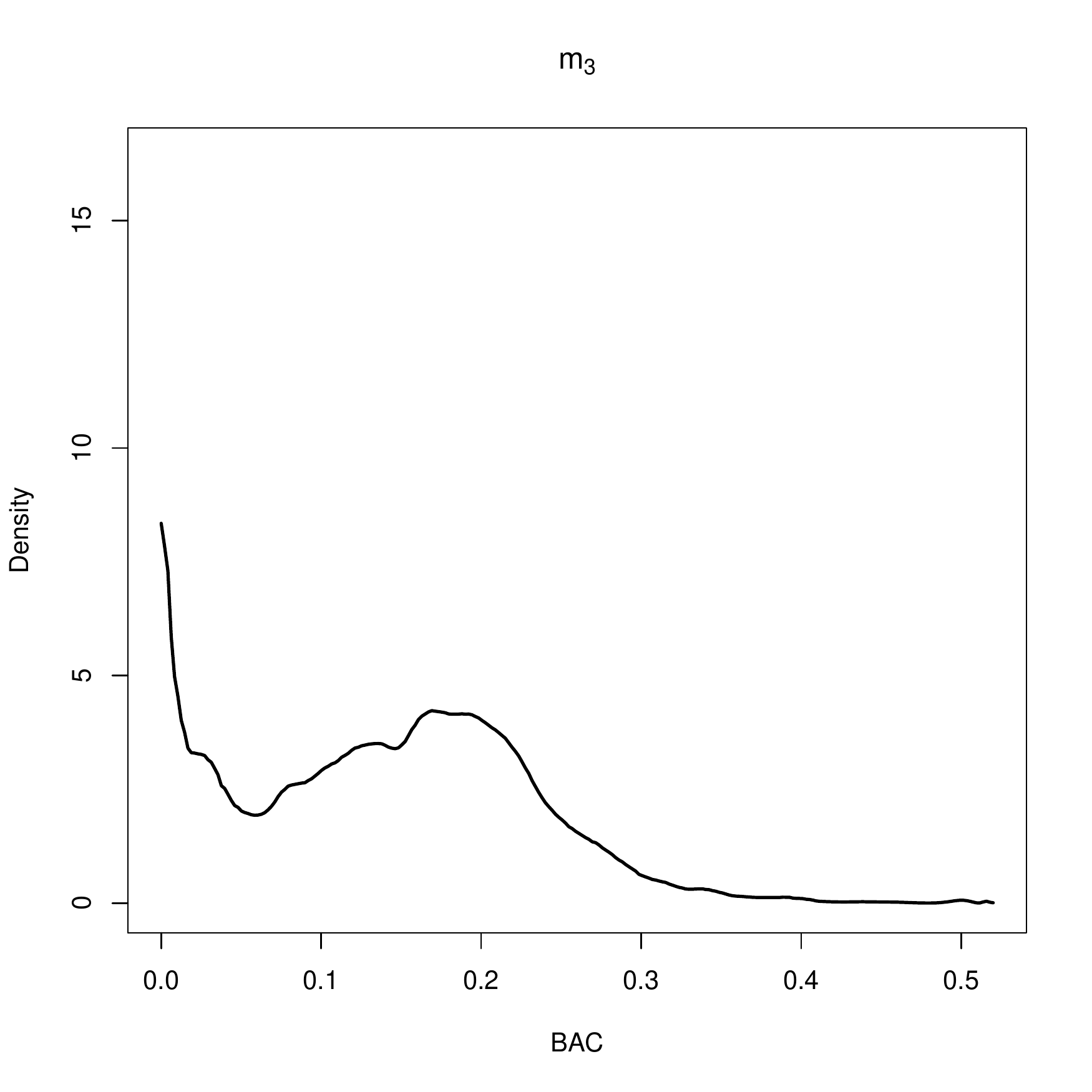}
\includegraphics[angle=0,width=0.49\linewidth]{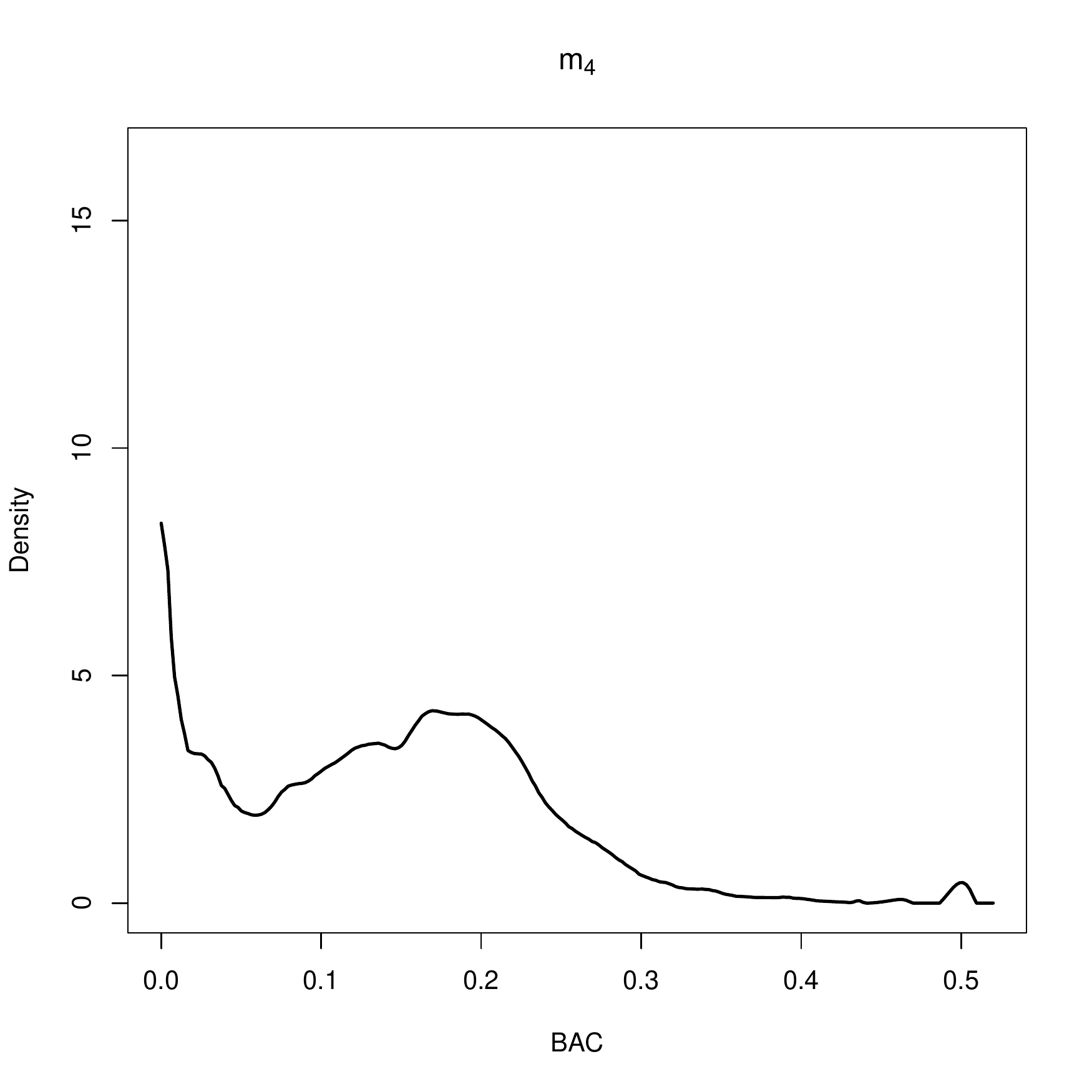}
\caption{BAC density estimates (top to bottom, left to right): $ m_1 $ - $a=1/2$ and $H(x)=x$;  $ m_2 $ - $a=1$ and $H(x)=x$;  $ m_3 $ - $a=1/2$ and $H(x)=\hat{G}(x)$; and $ m_4 $ - $a=1$ and $H(x)=\hat{G}(x)$.}
\label{fig:bac-dens}
\end{figure}

Figure \ref{fig:bac-dens} shows the four estimates. Methods $ m_1 $ and $ m_2 $ (orthonormal wavelets) indicate trimodal behavior, with a higher first peak for small amounts of BAC, whilst $ m_3 $ and $ m_4 $ (warped wavelets), suggest bimodal density, albeit for a tiny bump around 0.5 gram/100 ml for $ m_4 $. The data histogram (Figure \ref{fig:bac-hist}) and $m_1$-$m_3$ lead us to disregard this bump as some unwarranted feature due to $m_4$. Moreover, although not shown here, when $ m_1 $ and $ m_2 $ are employed with $ J_1 = 5 $, the second mode seen for $ J_1 = 6 $ vanishes, bringing all four estimates to a bimodal behavior. Finally, we can see in Figure  \ref{fig:bac-dens}  that some aliasing effect is present: for $m=2$ vis-a-vis $m=1$; and for either $m=2$ or $m=1$ vis-a-vis $m=4$ or $m=3$. Summarizing, we see that warping  and/or square-root estimation improves regularization by eliminating aliasing and most residual bumps. Thence, we conclude that $ m_3 $ provides the best regularized estimate for the true density in this application.
%When we compare the methods of estimation $ m_1 $--$ m_4 $, we can think of $ m_3 $ and $ m_4 $ providing a more realistic result, if we believe that most of the drivers did not ingest any alcohol before driving. Indeed, observe that the first peak of these methods is present for BAC 0 grams/100 ml, which is in accordance with the aforementioned assumption.

\section{Conclusions and further remarks}\label{sec:conclusions}

We propose a novel density estimation method in the context of size-biased data. We consider a wavelet-based method to estimate the power of a density of interest in a general framework, where the wavelet basis is allowed to be warped by some cumulative distribution function.

We show that both linear and regularized wavelet estimators are asymptotically consistent and that they attain optimal or near-optimal rates. In numerical studies, we considered four methods of estimation (particular cases of the proposed methodology), which include powers $ a = 1/2 $ \citep{Pinheiro.Vidakovic-1997-CSDA} and the usual $ a = 1 $, as well as orthonormal and warped wavelet bases. The results indicated that coarser resolution levels are better for ordinary orthonormal wavelet bases, whilst finer resolution levels are better for warped wavelets. They also indicate that warped wavelet estimators outperform orthonormal estimators, especially in the case of $ a = 1/2 $.

An issue not pursued here, which will be left as a topic for further research, regards a sharper data-driven estimate for the finest resolution level $ J_1 $.

\section*{Acknowledgments}
The first author acknowledges FAPESP (Funda\c{c}\~{a}o de Amparo \`{a}
Pesquisa do Estado de S\~{a}o Paulo) Grants 2018/04654-9 and 2020/00646-1. The second author acknowledges FAPESP Grant 2018/04654-9, and CNPq (Conselho Nacional de Desenvolvimento Cient\'{\i}fico e Tecnol\'ogico) Grant 310991/2020-0.
%% if your bibliography is in bibtex format, uncomment commands:
%\bibliographystyle{imsart-number} % Style BST file (imsart-number.bst or imsart-nameyear.bst)

\newpage
%%% The macro below will remove the \hrule's that should contain the abstract
\makeatletter
  \long\def\pprintMaketitle{\clearpage
  \iflongmktitle\if@twocolumn\let\columnwidth=\textwidth\fi\fi
  \resetTitleCounters
  \def\baselinestretch{1}%
  \printFirstPageNotes
  \begin{center}%
 \thispagestyle{pprintTitle}%
   \def\baselinestretch{1}%
    \Large\@title\par\vskip18pt
    \normalsize\elsauthors\par\vskip10pt
    \footnotesize\itshape\elsaddress\par\vskip36pt
    % \hrule\vskip12pt
    % \ifvoid\absbox\else\unvbox\absbox\par\vskip10pt\fi
    % \ifvoid\keybox\else\unvbox\keybox\par\vskip10pt\fi
    % \hrule\vskip12pt
    \end{center}%
  \gdef\thefootnote{\arabic{footnote}}%
  }
\makeatother
%%% The macro above will remove the \hrule's that should contain the abstract

{\linespread{1}
\begin{center}
\Large Supplementary material for ``Wavelet-based estimation of power densities of size-biased data''
\end{center}
\begin{center}
Michel H. Montoril$ ^a $, Aluísio Pinheiro$ ^b $, Brani Vidakovic$ ^c $
\end{center}
\begin{center} \footnotesize
$ ^a $Department of Statistics, Federal University of São Carlos, Brazil \\
$ ^b $Department of Statistics, University of Campinas, Brazil \\
$ ^c $Department of Statistics, Texas A\&M University, USA
\end{center}
}

\bigskip

In this supplementary material we present the proofs of the theoretical results in Section \ref{sec:theory} of the main manuscript. For the sake of simplicity, the assumptions used in the paper are presented below.

\subsubsection*{Assumptions}
\begin{enumerate}
\item[(a1)] $ f $ in \eqref{eq:pdf} is bounded away from zero and infinity and $ f^a \in \tilde{W}_2^m(U) $, for $ a \geq 1/2 $, $ 0 < U < \infty $ and $ m = 1, 2, \ldots $.
\item[(a2)] $ w $ in \eqref{eq:pdf} is bounded away from zero and infinity.
%\item[(a3)] $ J_0 \equiv J_0(n) $ and $ J_1 \equiv J_1(n) $ are positive integers such that $ J_0 \leq J_1 $, $ 2^{J_0} \asymp 2^{J_1} $ and $ J_12^{J_1}/n \to 0 $ as $ n \to \infty $;
\item[(a3)] The \cdf\ $ H $ used to warp the wavelet basis is continuous
and strictly monotone. Its \pdf\ $ h $ is bounded away from zero and infinity uniformly on $ [0, 1] $.
\item[(a4)] The employed wavelet basis is a periodized version of some Daubechies compactly supported wavelet basis, with at least $ m $ vanishing moments.
\end{enumerate}

%To avoid confusion with equation labels of the paper, we present the proofs below as an appendix, which provides a modification the equation labels in this supplementary material.

\appendix
\section{Proofs of the theoretical results in Section \ref{sec:theory}}\label{sec:appendix}

We give here the proofs of the results from Section \ref{sec:theory}. By assumptions (a2) and (a3),
\begin{eqnarray}
w(y) \asymp 1, \label{asymp:w} \\
h(y) \asymp 1. \label{asymp:h},
\end{eqnarray}
respectively. Also, note that assumptions (a1)-(a2) guarantee
\begin{eqnarray}
\mu & \asymp & 1, \label{asymp:mu} \\
g(y) & \asymp & 1 \label{asymp:g},
\end{eqnarray}
From \eqref{asymp:w},
\begin{equation}\label{asymp:hmu}
\hat{\mu} = \dfrac{1}{n} \sum_{i = 1}^{n} w^{-1}(Y_i) \asymp 1.
\end{equation}

We need the following lemma.

\begin{lemma}\label{lemma:diff-c}
Under the assumptions (a1)--(a4), for $ k = 0, 1, \ldots, 2^{J_0} - 1 $,
\[ \E\abs{\hat{c}_{J_0k} - c_{J_0k}}^2 \lesssim D_n \1_{\{a \neq 1\}} + n^{-1}, \]
where $ D_n $ is a positive sequence such that $ \sup_{y \in [0,1]} \abs{\hat{g}(y) - g(y)}^2 \lesssim D_n $ \as
\end{lemma}

\begin{proof}[Proof of Theorem \ref{theo:lin-a1}]
Analogous to the proof of Theorem \ref{theo:lin-aa} below.
\end{proof}

\begin{proof}[Proof of Theorem \ref{theo:lin-aa}]
Initially, observe that the convergence of $ \hat{f}_{J_0}^a $ is equivalent to the convergence of $ \hat{r}_{J_0}^a $, where $ r^a $ is defined in \eqref{eq:w-analysis}. In fact, by assumption (a3), it is easy to see that
\begin{equation}\label{asymp:norm-rel}
\norm{f^a}^2 \asymp \norm{r^a}^2.
\end{equation}
Let us denote $\rho_{j} = \norm{r_{j}^a - r^a} $, for a positive integer $ j $. Since, by (a4), $ r^a $ is analyzed by an orthonormal basis, it is easy to see that %Therefore, based on \eqref{asymp:norm-rel}, the rate of convergence of $ \hat{f}_{J_0}^a $ to $ f^a $ is Observe that the wavelet analysis of $ r^a $ and $ f^a $ are related to each other by their wavelet coefficients
%\begin{eqnarray*}
%f^a(x) & = & f^a\pa{H^*\pa{H(x)}} = r^a(y) \\
%& = & \sum_k c_{J_0k} \phi_{J_0k}(y) + \sum_{j \geq j_0} \sum_{k} d_{jk} \psi_{jk}(y) \\
%& = & \sum_k c_{J_0k} \phi_{J_0k}(H(x)) + \sum_{j \geq j_0} \sum_{k} d_{jk} \psi_{jk}(H(x))
%\end{eqnarray*}
%It is easy to see that
%\begin{equation*}
%\E\norm{\hat{f}_{J_0}^{a} - f^a}^2 = \E\norm{\hat{f}_{J_0}^{a} - f_{J_0}^a}^2 + \rho_{J_0}^2. %\E\norm{f_{J_0}^{a} - f^a}^2.
%\end{equation*}
\begin{equation*}
\E\norm{\hat{r}_{J_0}^{a} - r^a}^2 = \E\norm{\hat{r}_{J_0}^{a} - r_{J_0}^a}^2 + \rho_{J_0}^2. %\E\norm{f_{J_0}^{a} - f^a}^2.
\end{equation*}

By (a1) we have \citep{Restrepo.Leaf-1997-IJNME-supp}
\begin{equation}\label{asymp:rho}
\rho_{J_0}^2 \lesssim 2^{-2mJ_0}.
\end{equation}

By (a4) the basis is orthonormal. Therefore, by Parseval's identity and Lemma \ref{lemma:diff-c},
\begin{eqnarray}\label{asymp:diff-f}
\begin{split}
\E\norm{\hat{r}_{J_0}^{a} - r_{J_0}^a}^2  & = \sum_{k = 0}^{2^{J_0}-1} \E\abs{\hat{c}_{J_0k} - c_{J_0k}}^2 \lesssim \sum_{k = 0}^{2^{J_0}-1} \pa{D_n + n^{-1}} \\
& = 2^{J_0} D_n + \dfrac{2^{J_0}}{n}.
\end{split}
\end{eqnarray}

The desired result follows from \eqref{asymp:norm-rel}--\eqref{asymp:diff-f}.
\end{proof}

\begin{proof}[Proof of Corollary \ref{corol:lin-aa}]
Observe that by Theorem \ref{theo:gine.et.al},
\begin{equation*}
D_n \lesssim  \dfrac{J_0 2^{J_0}}{n} + 2^{-2mJ_0}.
\end{equation*}
Hence,
\begin{eqnarray*} \label{asymp:corol}
\begin{split}
2^{J_0} D_n + \dfrac{2^{J_0}}{n} + 2^{-2mJ_0} & \lesssim  \dfrac{J_0 2^{2J_0}}{n} +2^{-(2m-1)J_0} +\dfrac{2^{J_0}}{n} + 2^{-2mJ_0} \\
& \lesssim \dfrac{J_0 2^{2J_0}}{n} +2^{-(2m-1)J_0},
\end{split}
\end{eqnarray*}
which provides the desired result.
\end{proof}

\begin{proof}[Proof of Theorem \ref{theo:nonlin-a1}]
Still using $ r^a $ as defined in \eqref{eq:w-analysis}, observe that $\rho_{j} = \norm{r_{j}^a - r^a} = \sum_{j = J_1}^{+\infty} \sum_{k = 0}^{2^j-1} d_{jk} $. By Parseval's identity, it is easy to see that
\begin{equation}\label{ineq:dist-r}
\E \norm{ \tilde{r}_{J_1}^a - r^a }^2 = \sum_{k = 0}^{2^{J_0} - 1} \E \pa{ \hat{c}_{J_0k} - c_{J_0k} }^2 + \sum_{j = J_0}^{J_1 - 1} \sum_{k = 0}^{2^{j} - 1} \E \pa{ \tilde{d}_{jk} - d_{jk} }^2 + \rho_{J_1}.
\end{equation}

Since $ 0 \leq \lambda_{jk} \leq 1 $, the second term of the right hand side of the inequality above is bounded by
\begin{eqnarray}\label{ineq:nlin-a1}
\begin{split}
\sum_{j = J_0}^{J_1 - 1} \sum_{k = 0}^{2^j-1} \E \pa{ \tilde{d}_{jk} - d_{jk} }^2 & = \sum_{j = J_0}^{J_1 - 1} \sum_{k = 0}^{2^j-1} \E \pa{ \lambda_{jk} \hat{d}_{jk}- d_{jk} }^2 \\
& = \sum_{j = J_0}^{J_1 - 1} \sum_{k = 0}^{2^j-1} \E \pa{ \lambda_{jk} \hat{d}_{jk} - \lambda_{jk} d_{jk} + \lambda_{jk} d_{jk} - d_{jk} }^2 \\
& \leq 2 \sum_{j = J_0}^{J_1 - 1} \sum_{k = 0}^{2^j-1} \E \co{\lambda_{jk}^2 (\hat{d}_{jk} - d_{jk})^2} \\ & \qquad + 2 \sum_{j = J_0}^{J_1 - 1} \sum_{k = 0}^{2^j-1} \E \co{ (\lambda_{jk} - 1)^2 d_{jk}^2 } \\
& \leq 2 \sum_{j = J_0}^{J_1 - 1} \sum_{k = 0}^{2^j-1} \E \pa{\hat{d}_{jk} - d_{jk}}^2 + 2 \sum_{j = J_0}^{J_1 - 1} \sum_{k = 0}^{2^j-1} d_{jk}^2 \\
& = 2 \co{ \sum_{j = J_0}^{J_1 - 1} \sum_{k = 0}^{2^j-1} \E \pa{\hat{d}_{jk} - d_{jk}}^2 + \rho_{J_0}^2 - \rho_{J_1}^2 },
\end{split}
\end{eqnarray}
because $ \sum_{j = J_0}^{J_1 - 1} \sum_{k = 0}^{2^j-1} d_{jk}^2 = \rho_{J_0}^2 - \rho_{J_1}^2 $.

Therefore, \eqref{ineq:dist-r}, \eqref{ineq:nlin-a1} and Theorem \ref{theo:lin-a1} ensure that
\begin{eqnarray}
\begin{split}
\E \norm{ \tilde{r}_{J_1}^a - r^a }^2 \lesssim \sum_{k = 0}^{2^{J_0} - 1} \E \pa{ \hat{c}_{J_0k} - c_{J_0k} }^2 + \sum_{j = J_0}^{J_1 - 1} \sum_{k = 0}^{2^{j} - 1} \E \pa{ \hat{d}_{jk} - d_{jk} }^2 + \rho_{J_1}.
\end{split}
\end{eqnarray}
The last term above comes from the fact that $ \rho_{J_0} \asymp \rho_{J_1} $, because $ 2^{J_0} \asymp 2^{J_1} $, as stated in the above-mentioned theorem.

The desired result is yielded by \eqref{asymp:norm-rel}.
\end{proof}

\begin{proof}[Proofs of Theorem \ref{theo:nonlin-aa} and Corollary \ref{corol:nonlin-aa}]
The proofs of these results are similar to the proof of Theorem \ref{theo:nonlin-a1} presented above.
\end{proof}

\begin{proof}[Proof of Lemma \ref{lemma:diff-c}]
Let us focus initially on the case where $ a \neq 1 $. Then the estimator of the wavelet coefficients in \eqref{est:linear} can be written as
\begin{eqnarray*}
\hat{c}_{J_0k} & = & \dfrac{\hat{\mu}^a}{n} \sum_{i = 1}^{n} \dfrac{\phi_{J_0k}[H(Y_i)] h(Y_i) \hat{g}^{a-1}(Y_i)}{w^a(Y_i)} \\
& = & \dfrac{\hat{\mu}^a}{n} \sum_{i = 1}^{n} \dfrac{\phi_{J_0k}[H(Y_i)] h(Y_i)}{w^a(Y_i)} \co{\hat{g}^{a-1}(Y_i) - g^{a-1}(Y_i)} \\
&& \qquad + \dfrac{\hat{\mu}^a}{n} \sum_{i = 1}^{n} \dfrac{\phi_{J_0k}[H(Y_i)] h(Y_i) g^{a-1}(Y_i)}{w^a(Y_i)}.
\end{eqnarray*}
Thus,
\begin{eqnarray*}
\hat{c}_{J_0k} - c_{J_0k} & = & \dfrac{\hat{\mu}^a}{n} \sum_{i = 1}^{n} \dfrac{\phi_{J_0k}[H(Y_i)] h(Y_i)}{w^a(Y_i)} \co{\hat{g}^{a-1}(Y_i) - g^{a-1}(Y_i)} \\
&& \qquad + \dfrac{\hat{\mu}^a}{n} \sum_{i = 1}^{n} \dfrac{\phi_{J_0k}[H(Y_i)] h(Y_i) g^{a-1}(Y_i)}{w^a(Y_i)} - c_{J_0k} \\
& = & \dfrac{\hat{\mu}^a}{n} \sum_{i = 1}^{n} \dfrac{\phi_{J_0k}[H(Y_i)] h(Y_i)}{w^a(Y_i)} \co{\hat{g}^{a-1}(Y_i) - g^{a-1}(Y_i)} \\
&& \qquad + \pa{\dfrac{\hat{\mu}}{\mu}}^a \co{\dfrac{\mu^a}{n} \sum_{i = 1}^{n} \dfrac{\phi_{J_0k}[H(Y_i)] h(Y_i) g^{a-1}(Y_i)}{w^a(Y_i)} - c_{J_0k}} \\
&& \qquad + c_{J_0k} \hat{\mu}^a \pa{\dfrac{1}{\mu^a} - \dfrac{1}{\hat{\mu}^a}},
\end{eqnarray*}
which implies
\begin{eqnarray}\label{asymp:esp-dif-c}
%\begin{split}
\E\abs{\hat{c}_{J_0k} - c_{J_0k}}^2 & \lesssim & \E\abs{ \dfrac{\hat{\mu}^a}{n} \sum_{i = 1}^{n} \dfrac{\phi_{J_0k}[H(Y_i)] h(Y_i)}{w^a(Y_i)} \co{\hat{g}^{a-1}(Y_i) - g^{a-1}(Y_i)} }^2 \nonumber\\
& & \qquad + \E\abs{ \pa{\dfrac{\hat{\mu}}{\mu}}^a \co{\dfrac{\mu^a}{n} \sum_{i = 1}^{n} \dfrac{\phi_{J_0k}[H(Y_i)] h(Y_i) g^{a-1}(Y_i)}{w^a(Y_i)} - c_{J_0k}} }^2 \nonumber\\
& & \qquad + \E\abs{ c_{J_0k} \hat{\mu}^a \pa{\dfrac{1}{\mu^a} - \dfrac{1}{\hat{\mu}^a}} }^2 \nonumber\\
& \equiv & I_1 + I_2 + I_3.
%\end{split}
\end{eqnarray}

Beginning with $ I_1 $, by \eqref{asymp:g} it is easy to see that $ g^a $ is Lipschitz, satisfying
\begin{equation}\label{asymp:gLipschitz}
\abs{\hat{g}^{a-1}(y) - g^{a-1}(y)} \lesssim \abs{\hat{g}(y) - g(y)} \quad \forall y\in[0,1] \qquad \as
\end{equation}
Hence, since $ Y_1, \ldots, Y_n $ are \iid, by \eqref{asymp:hmu} and due to \eqref{asymp:gLipschitz},
\begin{eqnarray*}
I_1 & = & \E\abs{ \hat{\mu}^a \dfrac{\phi_{J_0k}[H(Y)] h(Y)}{w^a(Y)} \co{\hat{g}^{a-1}(Y) - g^{a-1}(Y)} }^2 \\
& \asymp & \E\abs{ \dfrac{\phi_{J_0k}[H(Y)] h(Y)}{w^a(Y)} \co{\hat{g}^{a-1}(Y) - g^{a-1}(Y)} }^2 \\
& \lesssim & \E\abs{ \dfrac{\phi_{J_0k}[H(Y)] h(Y)}{w^a(Y)} \co{\hat{g}(Y) - g(Y)} }^2.
%& \lesssim & \E\abs{ \dfrac{\phi_{J_0k}[H(Y)] h(Y)}{w^a(Y)} \co{\hat{g}(y) - g(y)} }^2 \\
\end{eqnarray*}

If there is a positive sequence $ D_n $ such that
\[ \sup_{y \in [0,1]} \abs{\hat{g}(y) - g(y)}^2 \lesssim D_n \quad \as, \]
then,
\begin{eqnarray}\label{asymp:I1}
\begin{split}
I_1  \lesssim D_n \E\abs{ \dfrac{\phi_{J_0k}[H(Y)] h(Y)}{w^a(Y)} }^2
 \asymp D_n \int_{0}^{1} \phi_{J_0k}^2(y) dy
& = D_n,
\end{split}
\end{eqnarray}
where the second inequality comes from \eqref{asymp:w}-\eqref{asymp:h}, and the third inequality is due to the fact that $ \int_{0}^{1} \phi_{J_0k}^2(y) dy = 1 $ \citep{Restrepo.Leaf-1997-IJNME-supp}.

In the analysis of $ I_2 $, let us denote
\[ \xi_i = \dfrac{\mu^a \phi_{J_0k}[H(Y_i)] h(Y_i) g^{a-1}(Y_i)}{w^a(Y_i)} - c_{J_0k}, \]
for $ i = 1, 2, \ldots, n $. It is immediate that $ \xi_1, \xi_2, \ldots, \xi_n $ are \iid\ Moreover, $ \E(\xi_i) = 0 $ and $ \var(\xi_i) \ 1 $. In fact, the zero mean comes from \eqref{eq:cjk} and
\begin{eqnarray}\label{asymp:var-xii}
\begin{split}
\var(\xi_i) & = \E\pa{\xi_i^2} \\
& = \E\co{\dfrac{\mu^a \phi_{J_0k}[H(Y_i)] h(Y_i) g^{a-1}(Y_i)}{w^a(Y_i)} - c_{J_0k}}^2 \\
& \lesssim \E\co{\dfrac{\mu^a \phi_{J_0k}[H(Y_i)] h(Y_i) g^{a-1}(Y_i)}{w^a(Y_i)}}^2 + c_{J_0k}^2 \\
& \lesssim \E\co{\dfrac{\mu^a \phi_{J_0k}[H(Y_i)] h(Y_i) g^{a-1}(Y_i)}{w^a(Y_i)}}^2 \\
& \asymp \int_{0}^{1} \phi_{J_0k}^2[H(y)] h(y) dy \\
& = \int_{0}^{1} \phi_{J_0k}^2(x) dx = 1,
\end{split}
\end{eqnarray}
where the fourth inequality comes from the fact that
\begin{eqnarray*}
c_{J_0k}^2 & = & \ch{\E\co{\mu^a \phi_{J_0k}[H(Y_i)] h(Y_i) g^{a-1}(Y_i)/w^a(Y_i)}}^2 \\
& \leq & \E\co{\mu^a \phi_{J_0k}[H(Y_i)] h(Y_i) g^{a-1}(Y_i)/w^a(Y_i)}^2
\end{eqnarray*}
and the fifth from \eqref{asymp:w}--\eqref{asymp:g}. Therefore, by \eqref{asymp:mu} and \eqref{asymp:hmu}, and because of \eqref{asymp:var-xii}, we have that
\begin{eqnarray}\label{asymp:I2}
\begin{split}
I_2  \asymp \var\pa{n^{-1} \sum_{i = 1}^{n} \xi_i}
&= n^{-1} \var(\xi_1) \lesssim n^{-1}.
\end{split}
\end{eqnarray}

Finally, note that, by \eqref{asymp:hmu},
\begin{eqnarray}\label{asymp:I31}
\begin{split}
I_3  \lesssim \E\abs{\hat{\mu}^a \pa{\dfrac{1}{\mu^a} - \dfrac{1}{\hat{\mu}^a}}}^2
& \asymp \E\abs{\dfrac{1}{\mu^a} - \dfrac{1}{\hat{\mu}^a}}^2.
\end{split}
\end{eqnarray}
Because of \eqref{asymp:mu} and \eqref{asymp:hmu}, we have that $ 1/\mu^a $ and $ 1/\hat{\mu}^a $ are Lipschitz continuous functions of $ 1/\mu $ and $ 1/\hat{\mu} $, respectively. Therefore,
\[
\abs{\dfrac{1}{\mu^a} - \dfrac{1}{\hat{\mu}^a}} \lesssim \abs{\dfrac{1}{\mu} - \dfrac{1}{\hat{\mu}}} \quad \as
\]
Hence, by \eqref{asymp:I31},
\begin{equation}\label{asymp:I3}
I_3 \lesssim \E\abs{\dfrac{1}{\mu} - \dfrac{1}{\hat{\mu}}}^2 \lesssim n^{-1}.
\end{equation}
The last inequality is verified as in the end of Proposition 4.1's proof \cite{Chesneau-2010-JotKSS-supp}.

Therefore, \eqref{asymp:esp-dif-c}, \eqref{asymp:I1}, \eqref{asymp:I2} and \eqref{asymp:I3} ensure that, for the case where $ a \neq 1 $,
\begin{equation}\label{asymp:E-diff-c-a}
\E\abs{\hat{c}_{J_0k} - c_{J_0k}}^2 \lesssim D_n + n^{-1}.
\end{equation}

The case where $ a = 1 $ is analogous but simpler. Observe that, when $ a = 1 $, $ \hat{g}^{a-1}(y) = 1 $ and $ g^{a-1}(y) = 1 $ for $ y \in [0,1] $. Therefore, in \eqref{asymp:esp-dif-c}, the term $ I_1 $ becomes null and the terms $ I_2 $ and $ I_3 $ are simplified without changing the upper bounds in \eqref{asymp:I2} and \eqref{asymp:I3}, respectively. Thus, for $ a = 1 $,
\begin{equation}\label{asymp:E-diff-c-a1}
\E\abs{\hat{c}_{J_0k} - c_{J_0k}}^2 \lesssim n^{-1}.
\end{equation}

Hence, \eqref{asymp:E-diff-c-a} and \eqref{asymp:E-diff-c-a1} yield the desired result.
\end{proof}

%%%%%%%%%%%%%%%%%%%%%%%%%%%%%%%%%%%%%%%%%%%%%%
%% Supplementary Material, if any, should   %%
%% be provided in {supplement} environment  %%
%% with title and short description.        %%
%%%%%%%%%%%%%%%%%%%%%%%%%%%%%%%%%%%%%%%%%%%%%%
%\begin{supplement}
%\stitle{???}
%\sdescription{???.}
%\end{supplement}

%% if your bibliography is in bibtex format, uncomment commands:
%\bibliographystyle{imsart-number} % Style BST file (imsart-number.bst or imsart-nameyear.bst)
%\bibliography{bibliography}       % Bibliography file (usually '*.bib')

%% or include bibliography directly:

%\bibliographystyle{apalike}
%\bibliography{biblio1}

\end{document}